\pdfoutput=1 
\documentclass[a4paper,twoside,10pt,fleqn]{article}
\usepackage{helvet}

\usepackage[latin1]{inputenc}
\usepackage{amsfonts}
\usepackage{amsmath}
\usepackage{amssymb,marvosym,units}
\usepackage{amsthm}
\usepackage[cmtip,arrow, matrix, curve]{xy}
\usepackage{pb-diagram,pb-xy}

\usepackage{fontenc}
\usepackage{graphicx}
\usepackage[svgnames,rgb]{xcolor}
\usepackage{epsfig}
\usepackage{tensor}

\usepackage{hyperref}
\usepackage{zref, xkeyval, ifpdf, ifthen, calc, marginnote, pdfcomment}

\usepackage{longtable} 
\usepackage{lscape}
\usepackage{bbm}

\hypersetup{
    pagebackref=true,
    bookmarks=true,       
    unicode=false,          
    pdftoolbar=true,        
    pdfmenubar=true,        
    pdffitwindow=true,     
    pdfstartview={FitH},    
    pdftitle={A new formulation of the holonomy-flux $^*$-algebra of  Loop Quantum Gravity},    
    pdfauthor={Diana Kaminski},     
    pdfsubject={},   
    pdfkeywords={Loop Quantum Gravity, abstract cross-product algebra}, 
    pdfnewwindow=true,      
    colorlinks=false,       
    linkcolor=black,          
    pdfborder= 0 0 0,
}

\title{Algebras of Quantum Variables for Loop Quantum Gravity\\[5pt]
\textbf{IV. A new formulation of the holonomy-flux $^*$-algebra}}
\author{Diana Kaminski\\[3pt]
kaminski@math.uni-paderborn.de\\ 
\small{Europe - Germany}}
\date{August 19, 2011}

\newcommand{\Ab}{\begin{large}\bar{\mathcal{A}}\end{large}}

\newcommand{\Alg}{\begin{large}\mathfrak{A}\end{large}}

\newcommand{\Aut}{\begin{large}\mathfrak{Aut}\end{large}}

\newcommand{\CB}{\mathbb{C}}

\newcommand{\DD}{\mathcal{D}}

\newcommand{\E}{\mathcal{E}}
\newcommand{\Ep}{\mathbbm{E}}

\newcommand{\FF}{\mathfrak{F}}

\newcommand{\Goid}{\mathcal{G}}

\newcommand{\GG}{G}

\newcommand{\HS}{\mathcal{H}}

\newcommand{\KD}{\mathcal{K}}

\newcommand{\la}{\langle}
\newcommand{\LD}{\mathcal{L}}

\newcommand{\Lop}{\mathfrak{L}}

\newcommand{\N}{\mathbb{N}}

\newcommand{\op}{\mathfrak{o}}

\newcommand{\PD}{\mathcal{P}}

\newcommand{\ra}{\rangle}

\newcommand{\TD}{\mathcal{T}}

\newcommand{\FD}{\mathcal{F}}

\newcommand{\R}{\mathbb{R}}

\newcommand{\SimGroup}{\mathfrak{G}}

\newcommand{\UD}{\mathcal{U}}

\newcommand{\WF}{\mathfrak{W}}

\newcommand{\ZD}{\mathcal{Z}}

\newcommand{\ho}{\mathfrak{h}}

\newcommand{\go}{\mathfrak{g}}
\newcommand{\gop}{\mathbbm{g}}
\newcommand{\Go}{\mathfrak{G}}

\DeclareMathOperator{\Act}{Act}

\DeclareMathOperator{\adm}{A}

\DeclareMathOperator{\Der}{Der}
\DeclareMathOperator{\dif}{d}
\DeclareMathOperator{\diff}{surf}

\DeclareMathOperator{\Diff}{Diff}

\DeclareMathOperator{\Exp}{Exp}

\DeclareMathOperator{\Hol}{Hol}
\DeclareMathOperator{\Hom}{Hom}

\DeclareMathOperator{\id}{id}

\DeclareMathOperator{\Map}{Map}

\DeclareMathOperator{\Pol}{Pol}

\DeclareMathOperator{\ori}{or}
\DeclareMathOperator{\pr}{pr}

\DeclareMathOperator{\Rep}{Rep}

\newcommand{\gp}{{\gamma^\prime}}
\newcommand{\gpi}{{\gamma^\prime_i}}
\newcommand{\gpj}{{\gamma^\prime_j}}
\newcommand{\gpk}{{\gamma^\prime_K}}

\newcommand{\gppi}{{\gamma^{\prime\prime}_i}}
\newcommand{\gppj}{{\gamma^{\prime\prime}_j}}
\newcommand{\gppje}{{\gamma^{\prime\prime}_{j+1}}}
\newcommand{\gppl}{{\gamma^{\prime\prime}_l}}
\newcommand{\gpe}{{\gamma^\prime_1}}
\newcommand{\gpz}{{\gamma^\prime_2}}
\newcommand{\gpm}{{\gamma^\prime_M}}
\newcommand{\gppe}{{\gamma^{\prime\prime}_1}}

\newcommand{\gppm}{{\gamma^{\prime\prime}_M}}

\newcommand{\gpppn}{{\gamma^{\prime\prime\prime}_N}}
\newcommand{\tg}{{\tilde\gamma}}

\newcommand{\Gp}{{\Gamma^\prime}}

\newcommand{\Gpp}{\Gamma^{\prime\prime}}
\newcommand{\Gppp}{\Gamma^{\prime\prime\prime}}

\newcommand{\idf}{\mathbbm{1}}

\newcommand{\bra}{[}
\newcommand{\ket}{]}

\newcommand{\beq}{\begin{equation}\begin{aligned}}
\newcommand{\beqs}{\begin{equation*}\begin{aligned}}
\newcommand{\be}{\begin{flalign}}
\newcommand{\bes}{\begin{equation*}}
\newcommand{\eq}{\end{aligned}\end{equation}}
\newcommand{\eqs}{\end{aligned}\end{equation*}}
\newcommand{\ee}{\end{flalign}}
\newcommand{\ees}{\end{equation}}

\newcommand{\limi}{\underset{i\rightarrow\infty}{\underrightarrow{\lim}}}

\newcommand{\limpPD}{\underset{\PD_\Gamma\in \PD_{\Gamma_\infty}}{\underleftarrow{\lim}}}

\newtheorem{theo}{Theorem }[section]
\newtheorem{lem}[theo]{Lemma}
\newtheorem{rem}[theo]{Remark}
\newtheorem{prop}[theo]{Proposition}
\newtheorem{cor}[theo]{Corollary}

\newtheorem{defi}[theo]{Definition}

\newenvironment{proofs}[1][Proof ]{\noindent\textbf{#1}: }{\ \begin{flushright}
                                                                         \rule{0.5em}{0.5em}
                                                                        \end{flushright}}

\newcounter{exa}[section]
 \newenvironment{exa}{\refstepcounter{exa}
  \textbf{Example} \thesection.\arabic{exa}: }{ {\begin{flushright}
                                                                         \rule{0.2em}{0.2em}
                                                                        \end{flushright}}}

\newcounter{problem}[subsection]
 \newenvironment{problem}{\refstepcounter{problem}
  \textbf{Problem} \thesection.\arabic{problem}: }{{\begin{flushright}
                                                                         \rule{0.2em}{0.2em}
                                                                        \end{flushright}}}
\newcommand{\GGi}{\xymatrix{
  \Goid_1  \ar@<-2pt>[r] \ar@<2pt>[r] &  \Goid^0_1    \\
}}
\newcommand{\GGii}{\xymatrix{
  \Goid_2  \ar@<-2pt>[r] \ar@<2pt>[r] &  \Goid^0_2    \\
}}
\newcommand{\GGm}{\xymatrix{
  \Goid  \ar@<-1pt>[r]^{s} \ar@<1pt>[r]_{t} &  \Goid^0    \\
}}
\newcommand{\GGim}{\xymatrix{
  \Goid_1  \ar@<-1pt>[r]^{s_1} \ar@<1pt>[r]_{t_1} &  \Goid^0_1    \\
}}
\newcommand{\GGiim}{\xymatrix{
  \Goid_2  \ar@<-1pt>[r]^{s_2} \ar@<1pt>[r]_{t_2} &  \Goid^0_2    \\
}}

\newcommand{\PGm}{\xymatrix{
  \PD  \ar@<-1pt>[r]^{s} \ar@<1pt>[r]_{t} &  \Sigma    \\
}}

\newcommand{\PGs}{\PD\Sigma\rightrightarrows\Sigma}
\newcommand{\PGoS}{\PD\rightrightarrows\Sigma}
\newcommand{\fPGm}{\xymatrix{
  \PD_\Gamma  \ar@<-1pt>[r]^{s} \ar@<1pt>[r]_{t} &  V_\Gamma    \\
}}
\newcommand{\PGsm}{\xymatrix{
  \PD\Sigma \ar@<-1pt>[r]^{s_{\PD\Sigma}} \ar@<1pt>[r]_{t_{\PD\Sigma}} &  \Sigma   \\
}}
\newcommand{\fPGms}{\xymatrix{
  \PD^s_\Gamma  \ar@<-1pt>[r]^{s} \ar@<1pt>[r]_{t} &  V_\Gamma    \\
}}

\newcommand{\fPG}{\PD_\Gamma\Sigma \rightrightarrows V_\Gamma
}

\newcommand{\fPSGm}{\xymatrix{
  \PD_\Gamma\Sigma  \ar@<-1pt>[r]^{s} \ar@<1pt>[r]_{t} &  V_\Gamma    \\
}}
\newcommand{\fPSG}{\xymatrix{
  \PD_\Gamma\Sigma  \ar@<-2pt>[r] \ar@<2pt>[r] &  V_\Gamma    \\
}}
\newcommand{\fgHGm}{\xymatrix{
  H(\Gamma)  \ar@<-1pt>[r]^/0.3em/{\hat s_H} \ar@<1pt>[r]_/0.3em/{\hat t_H} &  V_\Gamma    \\
}}

\newcommand{\fHGm}{\xymatrix{
  H_\Gamma  \ar@<-1pt>[r]^/0.3em/{\hat s_H} \ar@<1pt>[r]_/0.3em/{\hat t_H} &  V_\Gamma    \\
}}

\newcommand{\fGGm}{\xymatrix{
  \G^G_\Gamma  \ar@<-1pt>[r]^/0.3em/{s_P} \ar@<1pt>[r]_/0.3em/{t_P} &  V_\Gamma    \\
}}
\newcommand{\fGHm}{\xymatrix{
  \G^H_\Gamma  \ar@<-1pt>[r]^/0.3em/{s_P} \ar@<1pt>[r]_/0.3em/{t_P} &  V_\Gamma    \\
}}

\newcounter{count}
\begin{document}
\maketitle
\begin{abstract}\noindent In this article the holonomy-flux $^*$-algebra, which has been introduced by Lewandowski, Oko\l{}ow, Sahlmann and Thiemann \cite{LOST06}, is modificated. The new $^*$-algebra is called the holonomy-flux cross-product $^*$-algebra. This algebra is an abstract cross-product $^*$-algebra. It is given by the universal algebra of the algebra of continuous and differentiable functions on the configuration space of generalised connections and the universal enveloping flux algebra associated to a surface set, and some canonical commutator relations. There is a uniqueness result for a certain path- and graph-diffeomorphism invariant state of the holonomy-flux cross-product $^*$-algebra. This new $^*$-algebra  is not the only $^*$-algebra, which is generated by the algebra of certain continuous and differentiable functions on the configuration space of generalised connections and the universal enveloping flux algebra associated to a surface set. The theory of abstract cross-product algebras allows to define different new $^*$-algebras. Some of these algebras are presented in this article. 
\end{abstract}

\thispagestyle{plain}
\pdfbookmark[0]{\contentsname}{toc}
\tableofcontents

\section{Introduction}

In \cite{AshCorZap98} Ashtekar, Corichi and Zapata have introduced the concept of an algebra generated by holonomies along paths, which is Lie group-valued, and quantum fluxes associated to surfaces and paths, which take values in the Lie algebra of the Lie group. A further analysed $^*$-algebra and representations of this $^*$-algebra have been presented by Sahlmann \cite{Sahlmann02,Sahl02}, or by Oko\l{}\'{o}w and Lewandowski \cite{OkolLew03,OkolowLewa05}. Finally, the holonomy-flux $^*$-algebra has been presented by the project group Lewandowski, Oko\l{}\'{o}w, Sahlmann and Thiemann in \cite{LOST06}. In this article the holonomy-flux $^*$-algebra is reformulated in a slightly different way such that the resulting $^*$-algebra is different from the $^*$-algebra presented by Lewandowski, Oko\l{}\'{o}w, Sahlmann and Thiemann. 

The ideas for the definition of the quantum configuration and momentum variables in the project \textit{AQV} have been presented in \cite{Kaminski0} and are shortly introduced in section \ref{quantum variables}. In this article the Lie algebra-valued or enveloping algebra-valued quantum flux operators are used intensively. These operators replace and generalise the flux-like variables of  Lewandowski, Oko\l{}\'{o}w, Sahlmann and Thiemann. Moreover, new $^*$-algebras are derived from certain functions depending on holonomies along paths in a graph and Lie algebra-valued or enveloping algebra-valued quantum flux operators and some canonical commutator relation among them. The concept of the construction is introduced in section \ref{subsec constrholfluxalg}.

The new $^*$-algebras are in particular abstract cross-product algebras, which have been presented by Schm\"udgen and Klimyk \cite{KlimSchmued94} in the context of Hopf algebras. The holonomy-flux $^*$-algebra is generated by all multiplication operators defined by functions in $C(\Ab)$ and left (or right) vector fields $e^L$ (or $e^R$) on $C^\infty(\Ab)$. Similarly to the $^*$-algebras in Quantum Mechanics the new \textit{holonomy-flux cross-product $^*$-algebra} is generated by the identity $\idf$, the holonomies along paths and the Lie algebra-valued quantum flux operators satisfying certain canonical commutator relations. If the surfaces are restricted to a certain set of surfaces, then this algebra is called the \textit{holonomy-flux cross-product $^*$-algebra associated to a surface set}. In contrast to the holonomy-flux $^*$-algebra the construction of the holonomy-flux cross-product $^*$-algebra is independent of the Hilbert space and the representation of the operators on the Hilbert space. For the definition of the \textit{holonomy-flux cross-product $^*$-algebra associated to a graph $\Gamma$ and a surface set $\breve S$} the enveloping flux algebra $\bar\E_{\breve S,\Gamma}$ associated to a surface set $\breve S$ and a graph $\Gamma$ is necessary. Then this abstract cross-product $^*$-algebra is given by the tensor vector space of the analytic holonomy $C^*$-algebra restricted to a graph and the enveloping flux algebra associated to a surface set equipped with a multiplication operation, which is derived from a certain action of enveloping flux algebra associated to a surface set on the analytic holonomy $C^*$-algebra restricted to a graph. In particular, it is used that the analytic holonomy $C^*$-algebra restricted to a graph is a right (or left) $\bar\E_{\breve S,\Gamma}$-module algebra. 

The $^*$-representation of the enveloping flux algebra associated to a surface set and a graph is given by \textit{infinitesimal representation of the flux group associated to the surface set $\breve S$ and the graph $\Gamma$} on the Hilbert space $\HS_\Gamma$. The $^*$-representation $\pi$ of the holonomy-flux cross-product $^*$-algebra associated to the graph $\Gamma$ and the surface set $\breve S$ is given by this representation $\dif U_{\overrightarrow{L}}$ of the enveloping flux algebra associated to the surface set and the graph and the representation $\Phi_M$ of the analytic holonomy $C^*$-algebra restricted to the graph. Consequently, an element $f_\Gamma\otimes E_S(\gamma)$ is represented on the Hilbert space $\HS_\Gamma$ by
\beqs \pi(f_\Gamma\otimes E_S(\gamma)):=\frac{1}{2}\Phi_M(e^{\overrightarrow{L}}(f_\Gamma)) + \frac{1}{2} \Phi_M(f_\Gamma)\dif U_{\overrightarrow{L}}(E_S(\gamma))
\eqs
where $e^{\overrightarrow{L}}$ denotes the right-invariant vector field and $E_S(\gamma)$ is an element of the enveloping flux algebra associated to a surface set $\breve S$ and a graph $\Gamma$. The representation extends to a representation of the  holonomy-flux cross-product $^*$-algebra associated to a surface set $\breve S$. In theorem \ref{theo stateholfluxalg} it is shown that, the corresponding state is the unique surface-orientation-preserving graph-diffeomorphism invariant state of the holonomy-flux cross-product $^*$-algebra associated to the surface set $\breve S$. Moreover, for a restricted notion of graph-diffeomorphism invariance other $^*$-representations and states are studied in section \ref{subsec Repholfluxcross}. Simple tensor products of the holonomy-flux cross-product $^*$-algebra are shortly presented in section \ref{subsec tensorholflux}.

On the other hand, the Heisenberg double has been introduced by Schm\"udgen and Klimyk \cite{KlimSchmued94}. The Heisenberg double $\HS(C^\infty(G),\E)$ depends on a Lie algebra $G$ and the enveloping algebra $\E$ of the Lie algebra $\go$ associated to $G$. The universal enveloping algebra $\E$ of the complexified Lie algebra $\go^{\CB}$ is equipped with an antilinear and antimultiplicative involution $Y\mapsto Y^+$ such that $X^+=-X$ for all $X\in\go$. Therefore, the universal enveloping flux algebra itself is a unital $^*$-algebra, which is isomorphic to a particular $O^*$-algebra in the sense of Inoue \cite{Inoue} and Schm\"udgen \cite{Schmuedgen90}. The holonomy-flux cross-product $^*$-algebra associated to surfaces presented in section \ref{subsec constrholfluxalg}, is regarded as an abstract cross-product algebra, which is constructed from holonomies ($G$-valued) and quantum fluxes ($\go$- or $\E$-valued).  But it is discovered in section \ref{sec summaryofholfluxstar} that, the holonomy-flux cross-product $^*$-algebra is not equivalent to the Heisenberg algebra $\HS(C^\infty(\bar G_{\breve S,\Gamma}),\bar\E_{\breve S,\Gamma})$ associated to a surface set and a graph, which is presented in section \ref{subsec Heisenbergalgebras}. Indeed there are three different $\bar \E_{\breve S,\Gamma}$-module algebras that define three different holonomy-flux cross-product $^*$-algebras. The definitons of these algebras are related to different bilinear maps that define left (or right) module algebras and different multiplication operations on vector spaces that define cross-products.

In section \ref{sec tableQMholflux} the construction of algebras of quantum configuration and momentum variables are compared for Quantum Mechanics, the Weyl $C^*$-algebras associated to a graph and surface sets and the holonomy-flux cross-product $^*$-algebra associated to a graph and surface sets.

To summarise the holonomy-flux $^*$-algebra \cite{LOST06} is reformulated and generalised to the holonomy-flux cross-product $^*$-algebra by using the theory of abstract cross-product algebras invented by Schm\"udgen and Klimyk \cite{KlimSchmued94}. In particular this concept is based on Hopf algebras and has been used in the context of quantum groups. A further project is the reformulation of the algebra defined by Oko\l{}\'{o}w and Lewandowski \cite{LewOk08} in this new context.
\section{The basic quantum operators}\label{quantum variables}
\subsection{Finite path groupoids and graph systems}\label{subsec fingraphpathgroup}
Let $c:[0,1]\rightarrow\Sigma$ be continuous curve in the domain $\bra 0,1\ket$, which is (piecewise) $C^k$-differentiable ($1\leq k\leq \infty$), analytic ($k=\omega$) or semi-analytic ($k=s\omega$) in $\bra 0,1\ket$ and oriented such that the source vertex is $c(0)=s(c)$ and the target vertex is $c(1)=t(c)$. Moreover assume that, the range of each subinterval of the curve $c$ is a submanifold, which can be embedded in $\Sigma$. An \textbf{edge} is given by a \hyperlink{rep-equiv}{reparametrisation invariant} curve of class (piecewise) $C^k$. The maps $s_{\Sigma},t_{\Sigma}:P\Sigma\rightarrow\Sigma$ where $P\Sigma$ is the path space are surjective maps and are called the source or target map.    

A set of edges $\{e_i\}_{i=1,...,N}$ is called \textbf{independent} iff the only intersections points of the edges are source $s_{\Sigma}(e_i)$ or $t_{\Sigma}(e_i)$ target points. Composed edges are called \textbf{paths}. An \textbf{initial segment} of a path $\gamma$ is a path $\gamma_1$ such that there exists another path $\gamma_2$ and $\gamma=\gamma_1\circ\gamma_2$. The second element $\gamma_2$ is also called a \textbf{final segment} of the path $\gamma$.

\begin{defi}
A \textbf{graph} $\Gamma$ is a union of finitely many independent edges $\{e_i\}_{i=1,...,N}$ for $N\in\N$. The set $\{e_1,...,e_N\}$ is called the \textbf{generating set for $\Gamma$}. The number of edges of a graph is denoted by $\vert \Gamma\vert$. The elements of the set $V_\Gamma:=\{s_{\Sigma}(e_k),t_{\Sigma}(e_k)\}_{k=1,...,N}$ of source and target points are called \textbf{vertices}.
\end{defi}

A graph generates a finite path groupoid in the sense that, the set $\PD_\Gamma\Sigma$ contains all independent edges, their inverses and all possible compositions of edges. All the elements of $\PD_\Gamma\Sigma$ are called paths associated to a graph. Furthermore the surjective source and target maps $s_{\Sigma}$ and $t_{\Sigma}$ are restricted to the maps $s,t:\PD_\Gamma\Sigma\rightarrow V_\Gamma$, which are required to be surjective.

\begin{defi}\label{path groupoid} Let $\Gamma$ be a graph. Then a \textbf{finite path groupoid} $\PD_\Gamma\Sigma$ over $V_\Gamma$ is a pair $(\PD_\Gamma\Sigma, V_\Gamma)$ of finite sets equipped with the following structures: 
\begin{enumerate}
 \item two surjective maps \(s,t:\PD_\Gamma\Sigma\rightarrow V_\Gamma\), which are called the source and target map,
\item the set \(\PD_\Gamma\Sigma^2:=\{ (\gamma_i,\gamma_j)\in\PD_\Gamma\Sigma\times\PD_\Gamma\Sigma: t(\gamma_i)=s(\gamma_j)\}\) of finitely many composable pairs of paths,
\item the  composition \(\circ :\PD_\Gamma^2\Sigma\rightarrow \PD_\Gamma\Sigma,\text{ where }(\gamma_i,\gamma_j)\mapsto \gamma_i\circ \gamma_j\), 
\item the inversion map \(\gamma_i\mapsto \gamma_i^{-1}\) of a path,
\item the object inclusion map \(\iota:V_\Gamma\rightarrow\PD_\Gamma\Sigma\) and
\item $\PD_\Gamma\Sigma$ is defined by the set $\PD_\Gamma\Sigma$ modulo the algebraic equivalence relations generated by
\beq\label{groupoid0} \gamma_i^{-1}\circ \gamma_i\simeq \idf_{s(\gamma_i)}\text{ and }\gamma_i\circ \gamma_i^{-1}\simeq \idf_{t(\gamma_i)}
\eq 
\end{enumerate}
Shortly write $\fPSGm$. 
\end{defi} 
Clearly, a graph $\Gamma$ generates freely the paths in $\PD_\Gamma\Sigma$. Moreover the map $s \times t: \PD_\Gamma\Sigma\rightarrow V_\Gamma\times V_\Gamma$ defined by $(s\times t)(\gamma)=(s(\gamma),t(\gamma))$ for all $\gamma\in\PD_\Gamma\Sigma$ is assumed to be surjective ($\PD_\Gamma\Sigma$ over $V_\Gamma$ is a transitive groupoid), too. 

A general groupoid $\GG$ over $\GG^{0}$ defines a small category where the set of morphisms is denoted in general by $\GG$ and the set of objects is denoted by $\GG^{0}$. Hence in particular the path groupoid can be viewed as a category, since,
\begin{itemize}
\item the set of morphisms is identified with $\PD_\Gamma\Sigma$,
\item the set of objects is given by $V_\Gamma$ (the units) 
\end{itemize}

From the condition (\ref{groupoid0}) it follows that, the path groupoid satisfies additionally 
\begin{enumerate}
 \item $ s(\gamma_i\circ \gamma_j)=s(\gamma_i)\text{ and } t(\gamma_i\circ \gamma_j)=t(\gamma_j)\text{ for every } (\gamma_i,\gamma_j)\in\PD_\Gamma^2\Sigma$
\item $s(v)= v= t(v)\text{ for every } v\in V_\Gamma$
\item\label{groupoid1} $ \gamma \circ\idf_{s(\gamma)} = \gamma = \idf_{t(\gamma)}\circ \gamma\text{ for every } \gamma\in \PD_\Gamma\Sigma\text{ and }$
\item $\gamma \circ (\gamma_i\circ \gamma_j)=(\gamma \circ \gamma_i) \circ \gamma_j$
\item $\gamma \circ (\gamma^{-1}\circ \gamma_1)=\gamma_1= (\gamma_1 \circ \gamma) \circ \gamma^{-1}$
\end{enumerate}

The condition \ref{groupoid1} implies that the vertices are units of the groupoid. 

\begin{defi}
Denote the set of all finitely generated paths by
\beqs \PD_\Gamma\Sigma^{(n)}:=\{(\gamma_1,...,\gamma_n)\in \PD_\Gamma\times ...\PD_\Gamma: (\gamma_i,\gamma_{i+1})\in\PD^{(2)}, 1\leq i\leq n-1 \}\eqs
The set of paths with source point $v\in V_\Gamma$ is given by
\beqs \PD_\Gamma\Sigma^{v}:=s^{-1}(\{v\})\eqs
The set of paths with target  point $v\in V_\Gamma$ is defined by
\beqs \PD_\Gamma\Sigma_{v}:=t^{-1}(\{v\})\eqs
The set of paths with source point $v\in V_\Gamma$ and target point $u\in V_\Gamma$ is 
\beqs \PD_\Gamma\Sigma^{v}_u:=\PD_\Gamma\Sigma^{v}\cap \PD_\Gamma\Sigma_{u}\eqs
\end{defi}

A graph $\Gamma$ is said to be \hypertarget{disconnected}{\textbf{disconnected}} iff it contains only mutually pairs $(\gamma_i,\gamma_j)$ of non-composable independent paths $\gamma_i$ and $\gamma_j$ for $i\neq j$ and $i,j=1,...,N$. In other words for all $1\leq i,l\leq N$ it is true that $s(\gamma_i)\neq t(\gamma_l)$ and $t(\gamma_i)\neq s(\gamma_l)$ where $i\neq l$ and $\gamma_i,\gamma_l\in\Gamma$.

\begin{defi}
Let $\Gamma$ be a graph. A \textbf{subgraph $\Gp$ of $\Gamma$} is a given by a finite set of independent paths in $\PD_\Gamma\Sigma$. 
\end{defi}
For example let $\Gamma:=\{\gamma_1,...,\gamma_N\}$ then $\Gp:=\{\gamma_1\circ\gamma_2,\gamma_3^{-1},\gamma_4\}$ where $\gamma_1\circ\gamma_2,\gamma_3^{-1},\gamma_4\in\PD_\Gamma\Sigma$ is a subgraph of $\Gamma$, whereas the set $\{\gamma_1,\gamma_1\circ\gamma_2\}$ is not a subgraph of $\Gamma$. Notice if additionally $(\gamma_2,\gamma_4)\in\PD_\Gamma^{(2)}$ holds, then $\{\gamma_1,\gamma_3^{-1},\gamma_2\circ\gamma_4\}$ is a subgraph of $\Gamma$, too. Moreover for $\Gamma:=\{\gamma\}$ the graph $\Gamma^{-1}:=\{\gamma^{-1}\}$ is a subgraph of $\Gamma$. As well the graph $\Gamma$ is a subgraph of $\Gamma^{-1}$. A subgraph of $\Gamma$ that is generated by compositions of some paths, which are not reversed in their orientation, of the set $\{\gamma_1,...,\gamma_N\}$ is called an \textbf{orientation preserved subgraph of a graph}. For example for $\Gamma:=\{\gamma_1,...,\gamma_N\}$ orientation preserved subgraphs are given by $\{\gamma_1\circ\gamma_2\}$, $\{\gamma_1,\gamma_2,\gamma_N\}$ or $\{\gamma_{N-2}\circ\gamma_{N-1}\}$ if $(\gamma_1,\gamma_2)\in\PD_\Gamma\Sigma^{(2)}$ and $(\gamma_{N-2},\gamma_{N-1})\in\PD_\Gamma\Sigma^{(2)}$.   

\begin{defi}
A \textbf{finite graph system $\PD_\Gamma$ for $\Gamma$} is a finite set of subgraphs of a graph $\Gamma$. A finite graph system $\PD_{\Gp}$ for $\Gp$ is a \hypertarget{finite graph subsystem}{\textbf{finite graph subsystem}} of $\PD_\Gamma$ for $\Gamma$ iff the set $\PD_{\Gp}$ is a subset of $\PD_{\Gamma}$ and $\Gp$ is a subgraph of $\Gamma$. Shortly write $\PD_{\Gp}\leq\PD_{\Gamma}$.

A \hypertarget{finite orientation preserved graph system}{\textbf{finite orientation preserved graph system}} $\PD^{\op}_\Gamma$ for $\Gamma$ is a finite set of orientation preserved subgraphs of a graph $\Gamma$. 
\end{defi}

Recall that, a finite path groupoid is constructed from a graph $\Gamma$, but a set of elements of the path groupoid need not be a graph again. For example let $\Gamma:=\{\gamma_1\circ\gamma_2\}$ and $\Gp=\{\gamma_1\circ\gamma_3\}$, then $\Gpp=\Gamma\cup\Gp$ is not a graph, since this set is not independent. Hence only appropriate unions of paths, which are elements of a fixed finite path groupoid, define graphs. The idea is to define a suitable action on elements of the path groupoid, which corresponds to an action of diffeomorphisms on the manifold $\Sigma$. The action has to be transfered to graph systems. But the action of bisection, which is defined by the use of the groupoid multiplication, cannot easily generalised for graph systems. 

\begin{problem}\label{problem group structure on graphs systems}
Let $\breve\Gamma:=\{\Gamma_i\}_{i=1,..,N}$ be a finite set such that each $\Gamma_i$ is a set of not necessarily independent paths such that 
\begin{enumerate}
\item the set contains no loops and
\item each pair of paths satisfies one of the following conditions
\begin{itemize}
\item the paths intersect each other only in one vertex,
\item the paths do not intersect each other or
\item one path of the pair is a segment of the other path.
\end{itemize}
\end{enumerate}

Then there is a map $\circ:\breve\Gamma\times \breve\Gamma\rightarrow\breve\Gamma$ of two elements $\Gamma_1$ and $\Gamma_2$ defined by
\beqs \{\gamma_1,...,\gamma_M\}\circ\{\tg_1,...,\tg_M\}:= &\Big\{ \gamma_i\circ\tg_j:t(\gamma_i)=s(\tg_j)\Big\}_{1\leq i,j\leq M}\\
\eqs for $\Gamma_1:=\{\gamma_1,...,\gamma_M\},\Gamma_2:=\{\tg_1,...,\tg_M\}$. 
Moreover define a map $^{-1}:\breve\Gamma\rightarrow\breve\Gamma$ by
\beqs  \{\gamma_1,...,\gamma_M\}^{-1}:= \{\gamma^{-1}_1,...,\gamma^{-1}_M\}\eqs 

Then the following is derived
\beqs \{\gamma_1,...,\gamma_M\}\circ\{\gamma^{-1}_1,...,\gamma^{-1}_M\}&=\Big\{ \gamma_i\circ\gamma^{-1}_j: t(\gamma_i)=t(\gamma_j)\Big\}_{1\leq i,j\leq M}\\
&=\Big\{ \gamma_i\circ\gamma^{-1}_j:t(\gamma_i)=t(\gamma_j)\text{ and }i\neq j\Big\}_{1\leq i,j\leq M}\\
&\quad\cup\{\idf_{s_{\gamma_j}}\}_{1\leq j\leq M}\\
\neq &\quad\cup\{\idf_{s_{\gamma_j}}\}_{1\leq j\leq M}
\eqs The equality is true, if the set $\breve\Gamma$ contains only graphs such that all paths are mutually non-composable. Consequently this does not define a well-defined multiplication map. Notice that, the same is discovered if a similar map and inversion operation are defined for a finite graph system $\PD_\Gamma$. 
\end{problem}

Consequently the property of paths being independent need not be dropped for the definition of a suitable multiplication and inversion operation. In fact the independence property is a necessary condition for the construction of the holonomy algebra for analytic paths. Only under this circumstance each analytic path is decomposed into a finite product of independent piecewise analytic paths again. 

\begin{defi}
A finite path groupoid $\PD_{\Gp}\Sigma$ over $V_{\Gp}$ is a \textbf{finite path subgroupoid} of $\PD_{\Gamma}\Sigma$ over $V_\Gamma$ iff the set $V_{\Gp}$ is contained in $V_\Gamma$ and the set $\PD_{\Gp}\Sigma$ is a subset of $\PD_{\Gamma}\Sigma$. Shortly write $\PD_{\Gp}\Sigma\leq\PD_{\Gamma}\Sigma$.
\end{defi}

Clearly for a subgraph $\Gamma_1$ of a graph $\Gamma_2$, the associated path groupoid $\PD_{\Gamma_1}\Sigma$ over $V_{\Gamma_1}$ is a subgroupoid of $\PD_{\Gamma_2}\Sigma$ over $V_{\Gamma_2}$.  This is a consequence of the fact that, each path in $\PD_{\Gamma_1}\Sigma$ is a composition of paths or their inverses in $\PD_{\Gamma_2}\Sigma$. 

\begin{defi}
A \textbf{family of finite path groupoids} $\{\PD_{\Gamma_i}\Sigma\}_{i=1,...,\infty}$, which is a set of finite path groupoids $\PD_{\Gamma_i}\Sigma$ over $V_{\Gamma_i}$, is said to be \textbf{inductive} iff for any $\PD_{\Gamma_1}\Sigma,\PD_{\Gamma_2}\Sigma$ exists a $\PD_{\Gamma_3}\Sigma$ such that $\PD_{\Gamma_1}\Sigma,\PD_{\Gamma_2}\Sigma\leq\PD_{\Gamma_3}\Sigma$.

A \textbf{family of graph systems} $\{\PD_{\Gamma_i}\}_{i=1,...,\infty}$, which is a set of finite path systems $\PD_{\Gamma_i}$ for $\Gamma_i$, is said to be \textbf{inductive} iff for any $\PD_{\Gamma_1},\PD_{\Gamma_2}$ exists a $\PD_{\Gamma_3}$ such that $\PD_{\Gamma_1},\PD_{\Gamma_2}\leq \PD_{\Gamma_3}$.
\end{defi}

\begin{defi}
Let $\{\PD_{\Gamma_i}\Sigma\}_{i=1,...,\infty}$ be an inductive family of path groupoids and $\{\PD_{\Gamma_i}\}_{i=1,...,\infty}$ be an inductive family of graph systems.

The \textbf{inductive limit path groupoid $\PD$ over $\Sigma$} of an inductive family of finite path groupoids such that $\PD:=\limi\PD_{\Gamma_i}\Sigma$ is called the \textbf{(algebraic) path groupoid} $\PGoS$.

Moreover there exists an \textbf{inductive limit graph $\Gamma_\infty$} of an inductive family of graphs such that $\Gamma_\infty:=\limi \Gamma_i$.

The \textbf{inductive limit graph system} $\PD_{\Gamma_\infty}$ of an inductive family of graph systems such that $\PD_{\Gamma_\infty}:=\limi \PD_{\Gamma_i}$
\end{defi}

Assume that, the inductive limit $\Gamma_\infty$ of a inductive family of graphs is a graph, which consists of an infinite countable number of independent paths. The inductive limit $\PD_{\Gamma_\infty}$ of a inductive family $\{\PD_{\Gamma_i}\}$ of finite graph systems contains an infinite countable number of subgraphs of $\Gamma_\infty$ and each subgraph is a finite set of arbitrary independent paths in $\Sigma$. 

\subsection{Holonomy maps for finite path groupoids, graph systems and transformations}\label{subsec holmapsfinpath}
In section \ref{subsec fingraphpathgroup} the concept of finite path groupoids for analytic paths has been given. Now the holonomy maps are introduced for finite path groupoids and finite graph systems. The ideas are familar with those presented by Thiemann \cite{Thiembook07}. But for example the finite graph systems have not been studied before. Ashtekar and Lewandowski \cite{AshLew93} have defined the analytic holonomy $C^*$-algebra, which they have based on a finite set of independent hoops. The hoops are generalised for path groupoids and the independence requirement is implemented by the concept of finite graph systems. 

\subsubsection{Holonomy maps for finite path groupoids}\label{subsubsec holmap}

\paragraph*{Groupoid morphisms for finite path groupoids}\hspace{10pt} 

Let $\GGim, \GGiim$ be two arbitrary groupoids.

\begin{defi}
A \hypertarget{groupoid-morphism}{\textbf{groupoid morphism}} between two groupoids $\GG_1$ and $\GG_2$ consists of two maps  $\ho:\GG_1\rightarrow\GG_2$  and $h:\GG_1^0\rightarrow\GG_2^0$ such that
\beqs (\hypertarget{G1}{G1})\qquad \ho(\gamma\circ\gp)&= \ho(\gamma)\ho(\gp)\text{ for all }(\gamma,\gp)\in \GG_1^{(2)}\eqs
\beqs (\hypertarget{G2}{G2})\qquad s_{2}(\ho(\gamma))&=h(s_{1}(\gamma)),\quad t_2(\ho(\gamma))=h(t_{1}(\gamma))\eqs 
 
A \textbf{strong groupoid morphism} between two groupoids $\GG_1$ and $\GG_2$ additionally satisfies
\beqs (\hypertarget{SG2}{SG})\qquad \text{ for every pair }(\ho(\gamma),\ho(\gp))\in\GG_2^{(2)}\text{ it follows that }(\gamma,\gp)\in \GG_1^{(2)}\eqs
\end{defi}

Let $G$ be a Lie group. Then $G$ over $e_G$ is a groupoid, where the group multiplication $\cdot: G^2\rightarrow G$ is defined for all elements  $g_1,g_2,g\in G$ such that $g_1\cdot g_2 = g$. A groupoid morphism between a finite path groupoid $\PD_\Gamma\Sigma$ to $G$ is given by the maps
\[\ho_\Gamma: \PD_\Gamma\Sigma\rightarrow G,\quad h_\Gamma:V_\Gamma\rightarrow e_G \] Clearly
\beq \ho_\Gamma(\gamma\circ\gp)&= \ho_\Gamma(\gamma)\ho_\Gamma(\gp)\text{ for all }(\gamma,\gp)\in \PD_\Gamma\Sigma^{(2)}\\
s_G(\ho_\Gamma(\gamma))&=h_\Gamma(s_{\PD_\Gamma\Sigma}(\gamma)),\quad t_G(\ho_\Gamma(\gamma))=h_\Gamma(t_{\PD_\Gamma\Sigma}(\gamma))
\eq But for an arbitrary pair $(\ho_\Gamma(\gamma_1),\ho_\Gamma(\gamma_2))=:(g_1,g_2)\in G^{(2)}$ it does not follows that, $(\gamma_1,\gamma_2)\in \PD_\Gamma\Sigma^{(2)}$ is true. Hence $\ho_\Gamma$ is not a strong groupoid morphism.

\begin{defi}\label{def sameholanal}Let $\fPG$ be a finite path groupoid.

Two paths $\gamma$ and $\gp$ in $\PD_\Gamma\Sigma$ have the \textbf{same-holonomy for all connections} iff 
\beqs \ho_\Gamma(\gamma)=\ho_\Gamma(\gp)\text{ for all }&(\ho_\Gamma,h_\Gamma)\text{ groupoid morphisms }\\ & \ho_\Gamma:\PD_\Gamma\Sigma\rightarrow G, h:V_\Gamma\rightarrow\{e_G\}
\eqs Denote the relation by $\sim_{\text{s.hol.}}$.
\end{defi}
\begin{lem}
The same-holonomy for all connections relation is an equivalence relation. 
\end{lem}
Notice that, the quotient of the finite path groupoid and the same-holonomy relation for all connections replace the hoop group, which has been used in \cite{AshLew93}.
\begin{defi}\label{genrestgroupoidforgraph}
Let $\fPG$ be a finite path groupoid modulo same-holonomy for all connections equivalence.

A \hypertarget{holonomy map for a finite path groupoid}{\textbf{holonomy map for a finite path groupoid}} $\PD_\Gamma\Sigma$ over $V_\Gamma$ is a groupoid morphism consisting of the maps $(\ho_\Gamma,h_\Gamma)$, where
\(\ho_\Gamma:\PD_\Gamma\Sigma\rightarrow G,h_\Gamma:V_\Gamma\rightarrow \{e_G\}\). 
The set of all holonomy maps is abbreviated by $\Hom(\PD_\Gamma\Sigma,G)$.
\end{defi}

For a short notation observe the following.
In further sections it is always assumed that, the finite path groupoid $\fPG$ is considered modulo same-holonomy for all connections equivalence although it is not stated explicitly.

\paragraph*{Admissable maps and equivalent groupoid morphisms}\hspace{10pt}

Now consider a finite path groupoid morphism $(\ho_\Gamma,h_\Gamma)$ from a finite path groupoid $\PD_\Gamma\Sigma$ over $V_\Gamma$ to the groupoid $G$ over $\{e_G\}$, which is contained in $\Hom(\PD_\Gamma\Sigma,G)$.

Consider an arbitrary map $\go_\Gamma: \PD_\Gamma\Sigma\rightarrow G$. Then there is a groupoid morphism defined by
\beq\label{eq similarity_2b} \Go_\Gamma(\gamma)&:= \go_\Gamma(\gamma)\ho_\Gamma(\gamma)\go_\Gamma(\gamma^{-1})^{-1}\text{ for all }\gamma\in\PD_\Gamma\Sigma\\
\eq if and only if 
\beqs \Go_\Gamma(\gamma_1\circ\gamma_2)&=\Go_\Gamma(\gamma_1)\Go_\Gamma(\gamma_2)\text{ for all }(\gamma_1,\gamma_2)\in\PD_\Gamma\Sigma^{(2)}\eqs holds. Then $\Go_\Gamma\in \Hom(\PD_\Gamma\Sigma,G)$.

Hence for all $(\gamma_1,\gamma_2)\in\PD_\Gamma\Sigma^{(2)}$ it is necessary that
\beqs  \Go_\Gamma(\gamma_1\circ\gamma_2)
&= \go_\Gamma(\gamma_1\circ\gamma_2)\ho_\Gamma(\gamma_1\circ\gamma_2)\go_\Gamma(\gamma_2^{-1}\circ\gamma_1^{-1})^{-1}\\
&=\go_\Gamma(\gamma_1\circ\gamma_2)\ho_\Gamma(\gamma_1)\ho_\Gamma(\gamma_2)\go_\Gamma(\gamma_2^{-1}\circ\gamma_1^{-1})^{-1}\\
&\overset{!}{=}\go_\Gamma(\gamma_1)\ho_\Gamma(\gamma_1)\go_\Gamma(\gamma_1^{-1})^{-1}\go_\Gamma(\gamma_2)\ho_\Gamma(\gamma_2)\go_\Gamma(\gamma_2^{-1})^{-1}
\eqs is satisfied. Therefore the map is required to fulfill
\beq\label{equ2 g_gamma} &\go_\Gamma(\gamma_1)=\go_\Gamma(\gamma_1\circ\gamma_2)\text{, }\go_\Gamma(\gamma_2^{-1})=\go_\Gamma((\gamma_1\circ\gamma_2)^{-1})\text{ and }\\
&\go_\Gamma(\gamma_1^{-1})^{-1}\go_\Gamma(\gamma_2)=e_G\text{ for all }(\gamma_1,\gamma_2)\in\PD_\Gamma\Sigma^{(2)}\text{ in particular, }\\
&\go_\Gamma(\gamma^{-1})^{-1}\go_\Gamma(\gamma)=e_G\text{ for all }(\gamma^{-1},\gamma)\in\PD_\Gamma\Sigma^{(2)}
\eq
for every refinement $\gamma_1\circ\gamma_2$ of each $\gamma$ in $\PD_\Gamma\Sigma$ and $\gamma_1$ being an initial segment of $\gamma_1\circ \gamma_2$ and $\gamma_2^{-1}$ an final segment of $(\gamma_1\circ\gamma_2)^{-1}$.
In comparison with Fleischhack's definition in \cite[Def. 3.7]{Fleischhack06} such maps are called admissible.
\begin{defi}\label{def admiss}
The set of maps $\go_\Gamma:\PD_\Gamma\Sigma\rightarrow G$ satisfying \eqref{equ2 g_gamma} for all pairs of decomposable paths in $\PD_\Gamma^{(2)}\Sigma$ is called the \textbf{set of admissible maps} and is denoted by $\Map^{\adm}(\PD_\Gamma\Sigma,G)$. 
\end{defi} 

Consider a map $g_\Gamma:V_\Gamma\rightarrow G$ such that 
\beqs(g_\Gamma,\ho_\Gamma)\in \Map(V_\Gamma,G)\times \Hom(\PD_\Gamma\Sigma,G)
\eqs which is also called a local gauge map.
Then the map $ \tilde\Go_\Gamma$ defined by
\beq\label{eq similarity1} \tilde\Go_\Gamma(\gamma)&:= g_\Gamma(s(\gamma))\ho_\Gamma(\gamma)g_\Gamma(s(\gamma^{-1}))^{-1}\text{ for all }\gamma\in\PD_\Gamma\Sigma\eq is a groupoid morphism. This is a result of the computation: 
\beqs \tilde\Go_\Gamma(\gamma_1\gamma_2)&= g_\Gamma(s(\gamma_1))\ho_\Gamma(\gamma_1\gamma_2)g_\Gamma(t(\gamma_2))^{-1}\\&= g_\Gamma(s(\gamma_1))\ho_\Gamma(\gamma_1)g_\Gamma(t(\gamma_1))^{-1}
g_\Gamma(s(\gamma_2))\ho_\Gamma(\gamma_2)g_\Gamma(t(\gamma))^{-1}\eqs
since $t(\gamma_1)=s(\gamma_2)$.  

\begin{defi}\label{def similargroupoidhom}
Two groupoid morphisms $(\ho_\Gamma,h_\Gamma)$ and $(\Go_\Gamma,h_\Gamma)$, or respectively $(\tilde\Go_\Gamma,h_\Gamma)$, between the groupoids $\PD_\Gamma$ over $V_\Gamma$ and the groupoid $G$ over $\{e_G\}$, which are defined for $(\go_\Gamma,\ho_\Gamma)\in\Map(\PD_\Gamma\Sigma,G)\times\Hom(\PD_\Gamma\Sigma,G)$ by \eqref{eq similarity_2b}, or respectively for $(g_\Gamma,\ho_\Gamma)\in \Map(V_\Gamma,G)\times \Hom(\PD_\Gamma\Sigma,G)$ by \eqref{eq similarity1}, are said to be \textbf{similar or equivalent groupoid morphisms}.
\end{defi}

\subsubsection{Holonomy maps for finite graph systems}\label{subsec graphhol}

Ashtekar and Lewandowski \cite{AshLew93} have presented the loop decomposition into a finite set of independent hoops (in the analytic category). This structure is replaced by a graph, since a graph is a set of independent edges. Notice that, the set of hoops that is generated by a finite set of independent hoops, is generalised to the set of finite graph systems. A finite path groupoid is generated by the set of edges, which defines a graph $\Gamma$, but a set of elements of the path groupoid need not be a graph again. The appropriate notion for graphs constructed from sets of paths is the finite graph system, which is defined in section \ref{subsec fingraphpathgroup}. Now the concept of holonomy maps is generalised for finite graph systems. Since the set, which is generated by a finite number of independent edges, contains paths that are composable, there are two possibilities to identify the image of the holonomy map for a finite graph system on a fixed graph with a subgroup of $G^{\vert\Gamma\vert}$. One way is to use the generating set of independend edges of a graph, which has been also used in \cite{AshLew93}. On the other hand, it is also possible to identify each graph with a disconnected subgraph of a fixed graph, which is generated by a set of independent edges. Notice that, the author implements two situations. One case is given by a set of paths that can be composed further and the other case is related to paths that are not composable. This is necessary for the definition of an action of the flux operators. Precisely the identification of the image of the holonomy maps along these paths is necessary to define a well-defined action of a flux element on the configuration space. This issue has been studied in \cite{Kaminski1,KaminskiPHD}.

First of all consider a graph $\Gamma$ that is generated by the set $\{\gamma_1,...,\gamma_N\}$ of edges. Then each subgraph of a graph $\Gamma$ contain paths that are composition of edges in $\{\gamma_1,...,\gamma_N\}$ or inverse edges. For example the following set $\Gp:=\{\gamma_1\circ\gamma_2\circ\gamma_3,\gamma_4\}$ defines a subgraph of $\Gamma:=\{\gamma_1,\gamma_2,\gamma_3,\gamma_4\}$. Hence there is a natural identification available.

\begin{defi}
A subgraph $\Gp$ of a graph $\Gamma$ is always generated by a subset $\{\gamma_1,...,\gamma_M\}$ of the generating set $\{\gamma_1,...,\gamma_N\}$ of independent edges that generates the graph $\Gamma$. Hence each subgraph is identified with a subset of $\{\gamma_1^{\pm 1},...,\gamma_N^{\pm 1}\}$. This is called the \hypertarget{natural identification}{\textbf{natural identification of subgraphs}}.
\end{defi}

\begin{exa}\label{exa natidentif}
For example consider a subgraph $\Gp:=\{\gamma_1\circ\gamma_2,\gamma_3\circ\gamma_4,...,\gamma_{M-1}\circ\gamma_M\}$, which is identified naturally with a set $\{\gamma_1,...,\gamma_M\}$. The set $\{\gamma_1,...,\gamma_M\}$ is a subset of $\{\gamma_1,...,\gamma_N\}$ where $N=\vert \Gamma\vert$ and $M\leq N$. 

Another example is given by the graph $\Gpp:=\{\gamma_1,\gamma_2\}$ such that $\gamma_2=\gpe\circ\gpz$, then $\Gpp$ is identified naturally with $\{\gamma_1,\gpe,\gpz\}$. This set is a subset of $\{\gamma_1,\gpe,\gpz,\gamma_3,...,\gamma_{N-1}\}$. 
\end{exa}

\begin{defi}
Let $\Gamma$ be a graph, $\PD_\Gamma$ be the finite graph system. Let $\Gp:=\{\gamma_1,...,\gamma_M\}$be a subgraph of $\Gamma$.

A \hypertarget{holonomy map for a finite graph system}{\textbf{holonomy map for a finite graph system}} $\PD_\Gamma$ is a given by a pair of maps $(\ho_\Gamma,h_\Gamma)$ such that there exists a holonomy map\footnote{In the work the holonomy map for a finite graph system and the holonomy map for a finite path groupoid is denoted by the same pair $(\ho_\Gamma,h_\Gamma)$.} $(\ho_\Gamma,h_\Gamma)$ for the finite path groupoid $\fPG$ and
\beqs &\ho_\Gamma:\PD_\Gamma\rightarrow G^{\vert \Gamma\vert},\quad \ho_\Gamma(\{\gamma_1,...,\gamma_M\})=(\ho_\Gamma(\gamma_1),...,\ho_\Gamma(\gamma_M), e_G,...,e_G)\\
&h_\Gamma:V_\Gamma\rightarrow \{e_G\}
\eqs 
The set of all holonomy maps for the finite graph system is denoted by $\Hom(\PD_\Gamma,G^{\vert \Gamma\vert})$.

The image of a map $\ho_\Gamma$ on each subgraph $\Gp$ of the graph $\Gamma$ is given by
\beqs (\ho_\Gamma(\gamma_1),...,\ho_\Gamma(\gamma_M),e_G,...,e_G)
\eqs is an element of $G^{\vert \Gamma\vert}$. The set of all images of maps on subgraphs of $\Gamma$ is denoted by $\Ab_\Gamma$.
\end{defi}
The idea is now to study two different restrictions of the set $\PD_\Gamma$ of subgraphs. For a short notation of a ''set of  holonomy maps for a certain restricted set of subgraphs of a graph'' in this article the following notions are introduced.
\begin{defi}
If the subset of all disconnected subgraphs of the finite graph system $\PD_\Gamma$ is considered, then the restriction of $\Ab_\Gamma$, which is identified with $G^{\vert \Gamma\vert}$ appropriately, is called the \hypertarget{non-standard identification}{\textbf{non-standard identification of the configuration space}}. If the subset of all natural identified subgraphs of the finite graph system $\PD_\Gamma$ is considered, then the restriction of $\Ab_\Gamma$, which is identified with $G^{\vert \Gamma\vert}$ appropriately, is called the \hypertarget{natural identification}{\textbf{natural identification of the configuration space}}.
\end{defi}

A comment on the non-standard identification of $\Ab_\Gamma$ is the following. If $\Gp:=\{\gamma_1\circ\gamma_2\}$ and $\Gpp:=\{\gamma_2\}$ are two subgraphs of $\Gamma:=\{\gamma_1,\gamma_2,\gamma_3\}$. The graph $\Gp$ is a subgraph of $\Gamma$. Then evaluation of a map $\ho_\Gamma$ on a subgraph $\Gp$ is given by
\beqs \ho_\Gamma(\Gp)=(\ho_\Gamma(\gamma_1\circ\gamma_2),\ho_\Gamma(s(\gamma_2)),\ho_\Gamma(s(\gamma_3)))=(\ho_\Gamma(\gamma_1)\ho_\Gamma(\gamma_2),e_G,e_G)\in G^3
\eqs and the holonomy map of the subgraph $\Gpp$ of $\Gp$ is evaluated by
\beqs \ho_\Gamma(\Gpp)=(\ho_\Gamma(s(\gamma_1)),\ho_\Gamma(s(\gamma_2))\ho_\Gamma(\gamma_2),\ho_\Gamma(s(\gamma_3)))=(\ho_\Gamma(\gamma_2),e_G,e_G)\in G^3
\eqs

\begin{exa}
Recall example \thesubsection.\ref{exa natidentif}.
For example for a subgraph $\Gp:=\{\gamma_1\circ\gamma_2,\gamma_3\circ\gamma_4,...,\gamma_{M-1}\circ\gamma_M\}$, which is naturally identified with a set $\{\gamma_1,...,\gamma_M\}$. Then the holonomy map is evaluated at $\Gp$ such that \[\ho_\Gamma(\Gp)=(\ho_\Gamma(\gamma_1),\ho_\Gamma(\gamma_2),....,\ho_\Gamma(\gamma_M),e_G,...,e_G)\in G^N\] where $N=\vert \Gamma\vert$. For example, let $\Gp:=\{\gamma_1,\gamma_2\}$ such that $\gamma_2=\gpe\circ\gpz$ and which is naturally identified with $\{\gamma_1,\gpe,\gpz\}$. Hence \[\ho_\Gamma(\Gp)=(\ho_\Gamma(\gamma_1),\ho_\Gamma(\gpe),\ho_\Gamma(\gpz),e_G,...,e_G)\in G^N\] is true.

Another example is given by the disconnected graph $\Gp:=\{\gamma_1\circ\gamma_2\circ\gamma_3,\gamma_4\}$, which is a subgraph of $\Gamma:=\{\gamma_1,\gamma_2,\gamma_3,\gamma_4\}$. Then the non-standard identification is given by
\[\ho_\Gamma(\Gp)=(\ho_\Gamma(\gamma_1\circ\gamma_2\circ\gamma_3),\ho_\Gamma(\gamma_4),e_G,e_G)\in G^4\]

If the natural identification is used, then $\ho_\Gamma(\Gp)$ is idenified with 
\[(\ho_\Gamma(\gamma_1),\ho_\Gamma(\gamma_2),\ho_\Gamma(\gamma_3),\ho_\Gamma(\gamma_4))\in G^4\]

Consider the following example. Let $\Gppp:=\{\gamma_1,\alpha,\gamma_2,\gamma_3\}$ be a graph such that 
 \begin{center}
\includegraphics[width=0.2\textwidth]{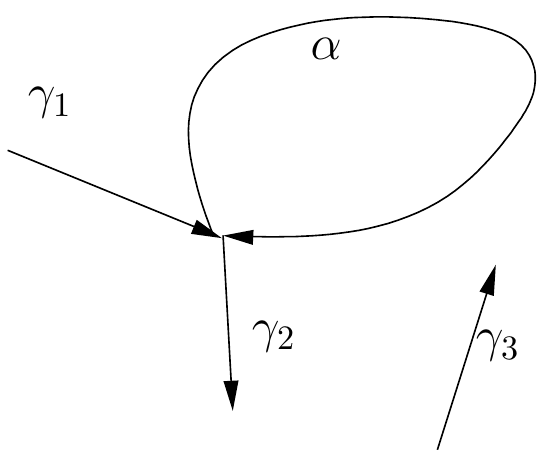}
\end{center}
Then notice the sets $\Gamma_1:=\{\gamma_1\circ\alpha,\gamma_3\}$ and $\Gamma_2:=\{\gamma_1\circ\alpha^{-1},\gamma_3\}$. In the non-standard identification of the configuration space $\Ab_{\Gppp}$ it is true that,
\beqs \ho_{\Gppp}(\Gamma_1)=(\ho_{\Gppp}(\gamma_1\circ\alpha),\ho_{\Gppp}(\gamma_3),e_G,e_G)\in G^4,\\
\ho_{\Gppp}(\Gamma_2)=(\ho_{\Gppp}(\gamma_1\circ\alpha^{-1}),\ho_{\Gppp}(\gamma_3),e_G,e_G)\in G^4
\eqs holds. Whereas in the natural identification of $\Ab_{\Gppp}$
 \beqs \ho_{\Gppp}(\Gamma_1)=(\ho_{\Gppp}(\gamma_1),\ho_{\Gppp}(\alpha),\ho_{\Gppp}(\gamma_3),e_G)\in G^4,\\
\ho_{\Gppp}(\Gamma_2)=(\ho_{\Gppp}(\gamma_1),\ho_{\Gppp}(\alpha^{-1}),\ho_{\Gppp}(\gamma_3),e_G)\in G^4
\eqs yields.
\end{exa}

The equivalence class of similar or equivalent groupoid morphisms defined in definition \ref{def similargroupoidhom} allows to define the following object.
The set of images of all holonomy maps of a finite graph system modulo the similar or equivalent groupoid morphisms equivalence relation is denoted by $\Ab_\Gamma/\bar\SimGroup_\Gamma$. 

\subsubsection{Transformations in finite path groupoids and finite graph systems}\label{subsubsec bisections}

The aim of this section is to clearify the graph changing operators in LQG framework and the role of finite diffeomorphisms in $\Sigma$. 
Therefore operations, which add, delete or transform paths, are introduced.  In particular translations in a finite path graph groupoid and in the groupoid $G$ over $\{e_G\}$ are studied. 

\paragraph*{Transformations in finite path groupoid\\[5pt]}

\begin{defi}
Let $\varphi$ be a $C^k$-diffeomorphism on $\Sigma$, which maps surfaces into surfaces. 

Then let $(\Phi_\Gamma,\varphi_\Gamma)$ be a pair of bijective maps, where $\varphi\vert_{V_\Gamma}=\varphi_\Gamma$ and 
\beq\Phi_\Gamma:\PD_\Gamma\Sigma\rightarrow\PD_\Gamma\Sigma\text{ and }\varphi_\Gamma:V_\Gamma\rightarrow V_\Gamma\eq 
such that 
\beq (s\circ\Phi_\Gamma)(\gamma)=(\varphi_\Gamma\circ s)(\gamma),\quad (t\circ \Phi_\Gamma)(\gamma)=(\varphi_\Gamma\circ t)(\gamma)\text{ for all }\gamma\in\PD_\Gamma\Sigma\eq holds such that $(\Phi_\Gamma,\varphi_\Gamma)$ defines a groupoid morphism.

Call the pair $(\Phi_\Gamma,\varphi_\Gamma)$ a \textbf{path-diffeomorphism of a finite path groupoid} $\PD_\Gamma\Sigma$ over $V_\Gamma$. Denote the set of finite path-diffeomorphisms by $\Diff(\PD_\Gamma\Sigma)$.
\end{defi}

Notice that, for $(\gamma,\gp)\in\PD_\Gamma\Sigma^{(2)}$ it is true that
\beq\label{eq requcombi0} \Phi_\Gamma(\gamma\circ\gp)=\Phi_\Gamma(\gamma)\circ\Phi_\Gamma(\gp)
\eq requires that
\beq\label{eq requcombi} (t\circ\Phi_\Gamma)(\gamma)=(s\circ\Phi_\Gamma)(\gp)
\eq Hence from \eqref{eq requcombi0} and \eqref{eq requcombi} it follows that, $\Phi_\Gamma(\idf_v)=\idf_{\varphi_\Gamma(v)}$ is true.

A path-diffeomorphism $(\Phi_\Gamma,\varphi_\Gamma)$ is lifted to $\Hom(\PD_\Gamma\Sigma,G)$. \\
The pair $(\ho_\Gamma\circ\Phi_\Gamma,h_\Gamma\circ\varphi_\Gamma)$ defined by
\beqs \ho_\Gamma\circ\Phi_\Gamma&: \PD_\Gamma\Sigma\rightarrow G,\quad \gamma\mapsto (\ho_\Gamma\circ\Phi_\Gamma)(\gamma)\\
h_\Gamma\circ\varphi_\Gamma&: V_\Gamma\rightarrow \{e_G\},\quad (h_\Gamma\circ\varphi_\Gamma)(v)=e_G
\eqs such that
\beqs &s_{\Hol}((\ho_\Gamma\circ\Phi_\Gamma)(\gamma))=(h_\Gamma\circ\varphi_\Gamma)(s(\gamma))=e_G,\\
&t_{\Hol}(\ho_\Gamma\circ\Phi_\Gamma(\gamma))=(h_\Gamma\circ\varphi_\Gamma)(t(\gamma))=e_G\text{ for all }\gamma\in\PD_\Gamma\Sigma
\eqs whenever $(\ho_\Gamma,h_\Gamma)\in\Hom(\PD_\Gamma\Sigma,G)$ and $(\Phi_\Gamma,\varphi_\Gamma)$ is a path-diffeomorphism, is a \hyperlink{holonomy map for a finite path groupoid}{holonomy map for a finite path groupoid} $\PD_\Gamma\Sigma$ over $V_\Gamma$.

\begin{defi}
A \textbf{left-translation in the finite path groupoid} $\PD_\Gamma\Sigma$ over $V_\Gamma$ at a vertex $v$ is a map defined by
\beqs L_\theta:\PD_\Gamma\Sigma^v\rightarrow \PD_\Gamma\Sigma^{w},\quad \gamma\mapsto L_{\theta}(\gamma):=\theta\circ\gamma
\eqs
for some $\theta\in\PD_\Gamma\Sigma_{v}^{w}$ and all $\gamma\in\PD_\Gamma\Sigma^v$.
\end{defi} 
In analogy a right-translation $R_\theta$ and an inner-translation $I_{\theta,\theta^\prime}$ in the finite path groupoid $\PD_\Gamma\Sigma$ over $V_\Gamma$ at a vertex $v$ can be defined.
\begin{rem}
Let $(\Phi_\Gamma,\varphi_\Gamma)$ be a path-diffeomorphism on a finite path groupoid $\PD_\Gamma\Sigma$ over $V_\Gamma$. Then a left-translation in the finite path groupoid $\PD_\Gamma\Sigma$ over $V_\Gamma$ at a vertex $v$ is defined by a path-diffeomorphism $(\Phi_\Gamma,\varphi_\Gamma)$ and the following object
\beq L_{\Phi_\Gamma}:\PD_\Gamma\Sigma^v\rightarrow \PD_\Gamma\Sigma^{\varphi_\Gamma(v)},\quad \gamma\mapsto L_{\Phi_\Gamma}(\gamma):=\Phi_\Gamma(\gamma)\text{ for }\gamma\in\PD_\Gamma\Sigma^v
\eq
Furthermore a right-translation in the finite path groupoid $\PD_\Gamma\Sigma$ over $V_\Gamma$ at a vertex $v$ is defined by a path-diffeomorphism $(\Phi_\Gamma,\varphi_\Gamma)$ and the following object
\beq R_{\Phi_\Gamma}:\PD_\Gamma\Sigma_v\rightarrow \PD_\Gamma\Sigma_{\varphi_\Gamma(v)},\quad \gamma\mapsto R_{\Phi_\Gamma}(\gamma):=\Phi_\Gamma(\gamma)\text{ for }\gamma\in\PD_\Gamma\Sigma_v
\eq

Finally an inner-translation in the finite path groupoid $\PD_\Gamma\Sigma$ over $V_\Gamma$ at the vertices $v$ and $w$ is defined by
\beqs  I_{\Phi_\Gamma}:\PD_\Gamma\Sigma^v_w\rightarrow \PD_\Gamma\Sigma^{\varphi_\Gamma(v)}_{\varphi_\Gamma(w)},\quad \gamma\mapsto I_{\Phi_\Gamma}(\gamma)=\Phi_\Gamma(\gamma)\text{ for }\gamma\in\PD_\Gamma\Sigma^v_w
\eqs where $(s\circ\Phi_\Gamma)(\gamma)=\varphi_\Gamma(v)$ and $(t\circ\Phi_\Gamma)(\gamma)=\varphi_\Gamma(w)$.
\end{rem}

In the following considerations the right-translation in a finite path groupoid is focused, but there is a generalisation to left-translations and inner-translations.
\begin{defi}
A \hypertarget{bisection of a finite path groupoid}{\textbf{bisection of a finite path groupoid}} $\PD_\Gamma\Sigma$ over $V_\Gamma$ is a map $\sigma:V_\Gamma\rightarrow\PD_\Gamma\Sigma$, which is right-inverse to the map $s:\PD_\Gamma\Sigma\rightarrow V_\Gamma$ (i.o.w. $s\circ\sigma=\id_{V_\Gamma}$) and such that $t\circ\sigma:V_\Gamma\rightarrow V_{\Gamma}$ is a bijective map\footnote{Note that in the infinite case of path groupoids an additional condition for the map $t\circ\sigma:\Sigma\rightarrow\Sigma$ has to be required. The map has to be a diffeomorphism. Observe that, the map $t\circ\sigma$ defines the finite diffeomorphism $\varphi_\Gamma:V_\Gamma\rightarrow V_\Gamma$.}. The set of bisections on $\PD_\Gamma\Sigma$ over $V_\Gamma$ is denoted $\mathfrak{B}(\PD_\Gamma\Sigma)$.
\end{defi}
\begin{rem}\label{rem defiofright}
Discover that, a bisection $\sigma\in\mathfrak{B}(\PD_\Gamma\Sigma)$ defines a path-diffeomorphism $(\varphi_\Gamma,\Phi_\Gamma)\in \Diff(\PD_\Gamma\Sigma)$, where $\varphi_\Gamma=t\circ\sigma$ and $\Phi_\Gamma$ is given by the right-translation $R_{\sigma(v)}:\PD_\Gamma\Sigma_v\rightarrow\PD_\Gamma\Sigma_{\varphi_\Gamma(v)}$ in $\fPG$, where $R_{\sigma(v)}(\gamma)=\Phi_\Gamma(\gamma)$ for all $\gamma\in\PD_\Gamma\Sigma_v$ and for a fixed $v\in V_\Gamma$. The right-translation is defined by 
\beq\label{eq Rendv}R_{\sigma(v)}(\gamma):=
\left\{\begin{array}{ll}
 \gamma\circ\sigma(v) & v=t(\gamma)\\
\gamma\circ\idf_{t(\gamma)} & v\neq t(\gamma)\\
\end{array}\right.\\
\eq whenever $t(\gamma)$ is the target vertex of a non-trivial path $\gamma$ in $\Gamma$. For a trivial path $\idf_v$ the right-translation is defined by $R_{\sigma(v)}(\idf_v)=\idf_{(t\circ\sigma)(v)}$ and $R_{\sigma(v)}(\idf_w)=\idf_{w}$ whenever $v\neq w$. The right-translation $R_{\sigma(v)}$ is required to be bijective. Before this result is proven in lemma \ref{lem path-diffeom} notice the following considerations. 
\end{rem}

Note that, $(R_{\sigma(v)},t\circ \sigma)$ transfers to the holonomy map such that
\beq\label{eq righttransl} (\ho_\Gamma\circ R_{\sigma(t(\gp))}(\gamma\circ\gp)&=\ho_\Gamma(\gamma\circ\gp\circ\sigma(t(\gp)))\\
&=\ho_\Gamma(\gamma)\ho_\Gamma(\gp\circ\sigma(t(\gp)))
\eq is true.
There is a bijective map between a right-translation $R_{\sigma(v)}:\PD_{\Gamma}\Sigma_v\rightarrow\PD_\Gamma\Sigma_{(t\circ\sigma)(v)}$ and a path-diffeomorphism $(\varphi_\Gamma,\Phi_\Gamma)$. In particular observe that, $\sigma\in\mathfrak{B}(\PD_\Gamma\Sigma_v)$ and $(\varphi_\Gamma,\Phi_\Gamma)\in\Diff(\PD_\Gamma\Sigma_v)$. Simply speaking the path-diffeomorphism does not change the source and target vertex at the same time. The path-diffomorphism changes the target vertex by a (finite) diffeomorphism and, therefore, the path is transformed. 

Bisections $\sigma$ in a finite path groupoid can be transfered, likewise path-diffeomorphisms, to holonomy maps. The pair $(\ho_\Gamma\circ \Phi_\Gamma, h_\Gamma\circ\varphi_\Gamma)$ of the maps
defines a pair of maps $(\ho_\Gamma\circ \Phi_\Gamma,h_\Gamma\circ\varphi_\Gamma)$ by 
\beq \ho_\Gamma\circ \Phi_\Gamma:\PD_\Gamma\Sigma_v\rightarrow G\text{ and } h_\Gamma\circ\varphi_\Gamma: V_\Gamma\rightarrow\{e_G\}
\eq which is a \hyperlink{holonomy map for a finite path groupoid}{holonomy map for a finite path groupoid} $\PD_\Gamma\Sigma$ over $V_\Gamma$.

\begin{lem}\label{lem groupbisection}The set $\mathfrak{B}(\PD_\Gamma\Sigma)$ of bisections on the finite path groupoid $\PD_\Gamma\Sigma$ over $V_\Gamma$ forms a group.
\end{lem}
\begin{proofs}The group multiplication is given by
\beqs (\sigma\ast\sigma^\prime)(v) =\sigma^\prime(v)\circ\sigma(t(\sigma^\prime(v)))\text{ for }v\in V_\Gamma
\eqs whenever $\sigma^\prime(v)\in\PD_\Gamma\Sigma^v_{\varphi_\Gamma^\prime(v)}$ and $\sigma(t(\sigma^\prime(v)))\in\PD_\Gamma\Sigma^{(t\circ\sigma^\prime)(v)}_{\varphi_\Gamma(v)}$.

Clearly the group multiplication is associative.
The unit $\id$ is equivalent to the object inclusion $v\mapsto\idf_v$ of the groupoid $\fPG$, where $\idf_v$ is the constant loop at $v$, and the inversion is given by
\beqs \sigma^{-1}(v)=\sigma((t\circ\sigma)^{-1}(v))^{-1}\text{ for }v\in V_\Gamma
\eqs 
\end{proofs}

The group property of bisections $\mathfrak{B}(\PD_\Gamma\Sigma)$ carries over to holonomy maps. Using the group multiplication $\cdot $ of $G$ conclude that
\beqs (\ho_\Gamma\circ R_{(\sigma\ast\sigma^\prime)(v)})(\idf_{v})=\ho_\Gamma\circ (R_{\sigma^\prime(v)}\circ R_{\sigma(t(\sigma^\prime(v)))})(\idf_{v})
= \ho_\Gamma(\sigma^\prime(v))\cdot\ho_\Gamma(\sigma(t(\sigma^\prime(v))))\text{ for }v\in V_\Gamma
\eqs is true. 
\begin{rem}
Moreover right-translations define path-diffeomorphisms, i.e. $R_{(\sigma)(v)}=\Phi_\Gamma$ and $\varphi_\Gamma=t\circ\sigma$ whenever $v\in V_\Gamma$.
But for two bisections $\sigma_\Gp,\breve\sigma_\Gp\in\mathfrak{B}(\PD_\Gamma\Sigma)$ the object $\sigma_\Gp(v)\circ\breve\sigma_\Gp(v)$ is not comparable with $(\sigma_\Gp\ast\breve\sigma_\Gp)(v)$. Then for the composition $\Phi_1(\gamma)\circ\Phi_2(\gamma)$, there exists no path-diffeomorphism $\Phi$ such that $\Phi_1(\gamma)\circ\Phi_2(\gamma)=\Phi(\gamma)$ yields in general. Moreover generally the object $\Phi_1(\gamma)\circ\Phi_2(\gp)=\Phi(\gamma\circ\gp)$ is not well-defined.

But the following is defined
\beq\label{def compositionofdiffeo} R_{(\sigma\ast\sigma^\prime)(v)}(\gamma)=\Phi_\Gamma^\prime(\gamma)\circ\Phi_\Gamma(\idf_{\varphi_\Gamma^\prime(v)})=:(\Phi_\Gamma^\prime\ast\Phi_\Gamma)(\gamma)
\eq whenever $\gamma\in\PD_\Gamma\Sigma_v$, $(\varphi_\Gamma,\Phi_\Gamma)\in\Diff(\PD_\Gamma\Sigma_v)$ and $(\varphi_\Gamma^\prime,\Phi_\Gamma^\prime)\in\Diff(\PD_\Gamma\Sigma_{\varphi_\Gamma^\prime(v)})$ are path-diffeomorphisms  such that $\varphi_\Gamma=t\circ\sigma$, $\Phi_\Gamma=R_{\sigma(\varphi_\Gamma^\prime(v))}$ and  $\varphi_\Gamma^\prime=t\circ\sigma^\prime$, $\Phi_\Gamma^\prime=R_{\sigma^\prime(v)}$.

Moreover for $(\gamma,\gp)\in\PD_\Gamma\Sigma^{(2)}$ and $\gp\in\PD_\Gamma\Sigma_v$ it is true that
\beqs (\Phi_\Gamma^\prime\ast\Phi_\Gamma)(\gamma\circ\gp)=\Phi_\Gamma^\prime(\gamma\circ\gp)\circ\Phi_\Gamma(\idf_{\varphi_\Gamma^\prime(v)}) = \Phi_\Gamma^\prime(\gamma)\circ\Phi_\Gamma^\prime(\gp)\circ\Phi_\Gamma(\idf_{\varphi_\Gamma^\prime(v)}) = \Phi_\Gamma^\prime(\gamma)\circ(\Phi_\Gamma^\prime\ast\Phi_\Gamma)(\gp)
\eqs holds.
\end{rem}

Then the following lemma easily follows.

\begin{lem}\label{lem path-diffeom}Let $\sigma$ be a bisection contained in $\mathfrak{B}(\PD_\Gamma\Sigma)$ and $v\in V_\Gamma$.

The pair $(R_{\sigma(v)},t\circ\sigma)$ of maps such that
\beqs &R_{\sigma(v)}:\PD_\Gamma\Sigma_v\rightarrow\PD_\Gamma\Sigma_{(t\circ \sigma)(v)},\quad &
s\circ R_{\sigma(v)} = (t\circ\sigma)\circ s\\
&t\circ\sigma:V_\Gamma\rightarrow V_\Gamma,\quad &t\circ R_{\sigma(v)}= (t\circ\sigma)\circ t
\eqs defined in remark \ref{rem defiofright} is a path-diffeomorphism in $\fPG$.
\end{lem}
\begin{proofs}This follows easily from the derivation
\beqs R_{\sigma(t(\gp))}(\gamma\circ\gp)&=\gamma\circ\gp\circ\sigma(t(\gp))
= R_{\sigma(t(\gp))}(\gamma)\circ R_{\sigma(t(\gp))}(\gp)
\eqs 
\beqs R_{\sigma(t(\gamma))}(\idf_{s(\gamma)}\circ\gamma)&=R_{\sigma(t(\gamma))}(\idf_{s(\gamma)})\circ R_{\sigma(t(\gamma))}(\gamma)=\idf_{s(\gamma)}\circ\gamma\circ\sigma(t(\gamma))\\
R_{\sigma(t(\gamma))}(\gamma\circ\idf_{t(\gamma)})&=R_{\sigma(t(\gamma))}(\gamma)\circ R_{\sigma(t(\gamma))}(\idf_{t(\gamma)})=\gamma\circ\sigma(t(\gamma))\circ\idf_{(t\circ\sigma)(t(\gamma))}
\eqs The inverse map satisfies 
\beqs R^{-1}_{\sigma(v)}(\gamma\circ\sigma(v))=R_{\sigma^{-1}(v)}(\gamma\circ\sigma(v))
=\gamma\circ\sigma(v)\circ\sigma^{-1}(v)=\gamma
\eqs whenever $v=t(\gamma)$, 
\beqs R^{-1}_{\sigma(v)}(\gamma)=\gamma
\eqs whenever $v\neq t(\gamma)$ and
\beqs R^{-1}_{\sigma(v)}(\idf_{(t\circ\sigma)(v)})=\idf_{v}
\eqs

Moreover derive
\beqs ( s\circ R_{\sigma(v)})(\gp)= ((t\circ\sigma)\circ s)(\gp)
\eqs for all $\gp\in\PD_\Gamma\Sigma_v$ and a fixed bisection $\sigma\in \mathfrak{B}(\PD_\Gamma\Sigma)$.
 \end{proofs}

Notice that, $L_{\sigma(v)}$ and $I_{\sigma(v)}$ similarly to the pair $(R_{\sigma(v)},t\circ\sigma)$ can be defined. Summarising the pairs $(R_{\sigma(v)},t\circ\sigma)$, $(L_{\sigma(v)},t\circ\sigma)$ and $(I_{\sigma(v)},t\circ\sigma)$ for a bisection $\sigma \in\mathfrak{B}(\PD_\Gamma\Sigma)$ are path-diffeomorphisms of a finite path groupoid $\fPG$.

In general a right-translation $(R_\sigma,t\circ\sigma)$ in the finite path groupoid $\PD_\Gamma\Sigma$ over $\Sigma$ for a bisection $\sigma\in\mathfrak{B}(\PD_\Gamma\Sigma)$ is defined by the bijective maps $R_\sigma$ and $t\circ\sigma$, which are given by 
\beqs 
&R_{\sigma}:\PD_\Gamma\Sigma\rightarrow\PD_\Gamma\Sigma,\quad &
s\circ R_{\sigma} =  s\qquad\quad\text{ }\\
&t\circ\sigma:V_\Gamma\rightarrow V_\Gamma,\quad &t\circ R_{\sigma}= (t\circ\sigma)\circ t\\
&R_\sigma(\gamma):=\gamma\circ\sigma(t(\gamma))\quad\forall \gamma\in \PD_\Gamma\Sigma;\quad R^{-1}_\sigma:=R_{\sigma^{-1}}
\eqs For example for a fixed suitable bisection $\sigma$ the right-translation is $R_\sigma(\idf_v)=\gamma$, then $R^{-1}_\sigma(\gamma)=\gamma\circ\gamma^{-1}=\idf_v$ for $v=s(\gamma)$. Clearly the right-translation $(R_\sigma,t\circ\sigma)$ is not a groupoid morphism in general. 

\begin{defi}
Define for a given bisection $\sigma$ in $\mathfrak{B}(\PD_\Gamma\Sigma)$, the \textbf{right-translation in the groupoid $G$ over $\{e_G\}$} through
\beqs &\ho_\Gamma\circ R_{\sigma}:\PD_\Gamma\Sigma\rightarrow G, \quad \gamma\mapsto (\ho_\Gamma\circ R_\sigma)(\gamma):=\ho_\Gamma(\gamma\circ\sigma(t(\gamma)))= \ho_\Gamma(\gamma)\cdot\ho_\Gamma(\sigma(t(\gamma))) \\
&h_\Gamma\circ t\circ\sigma:V_\Gamma\rightarrow e_G 
\eqs 

Furthermore for a fixed $\sigma\in\mathfrak{B}(\PD_\Gamma\Sigma)$ define
the \textbf{left-translation in the groupoid $G$ over $\{e_G\}$} by
\beqs &\ho_\Gamma\circ L_\sigma:\PD_\Gamma\Sigma\rightarrow G,\quad \gamma\mapsto \ho_\Gamma(\sigma((t\circ\sigma)^{-1}(s(\gamma)))\circ\gamma)=\ho_\Gamma(\sigma((t\circ\sigma)^{-1}(s(\gamma))))\cdot\ho_\Gamma(\gamma)\\
&h_\Gamma\circ t\circ\sigma:V_\Gamma\rightarrow e_G 
\eqs
and the \textbf{inner-translation in the groupoid $G$ over $\{e_G\}$}
\beqs &\ho_\Gamma\circ I_\sigma:\PD_\Gamma\Sigma\rightarrow G,\quad \gamma\mapsto \ho_\Gamma(\sigma((t\circ\sigma)^{-1}(s(\gamma)))\circ\gamma\circ\sigma(t(\gamma)))=\ho_\Gamma(\sigma((t\circ\sigma)^{-1}(s(\gamma))))\cdot\ho_\Gamma(\gamma)\cdot\ho_\Gamma(\sigma(t(\gamma)))\\
&h_\Gamma\circ t\circ\sigma:V_\Gamma\rightarrow e_G 
\eqs such that $I_\sigma=L_{\sigma^{-1}}\circ R_\sigma$.
\end{defi}

The pairs $(R_\sigma,t\circ\sigma)$ and $(L_\sigma,t\circ\sigma)$ are not groupoid morphisms. Whereas the pair $(I_\sigma,t\circ\sigma)$ is a groupoid morphism, since for all pairs $(\gamma,\gp)\in\PD_\Gamma\Sigma^{(2)}$ such that $t(\gamma)=s(\gp)$ it is true that $\sigma(t(\gamma)) \circ\sigma((t\circ\sigma)^{-1}(t(\gamma)))^{-1}=\idf_{t(\gamma)}$ holds. Notice that, in this situation $\sigma(t(\gamma))=\sigma(t(\gamma\circ\gp))$ is satisfied.

\begin{prop}
The map $\sigma\mapsto R_\sigma$ is a group isomorphism, i.e. $R_{\sigma\ast\sigma^\prime}=R_\sigma\circ R_{\sigma^\prime}$ and where $\sigma\mapsto t\circ\sigma$ is a group isomorphism from $\mathfrak{B}(\PD_\Gamma\Sigma)$ to the group of finite diffeomorphisms $\Diff(V_\Gamma)$ in a finite subset $V_\Gamma$ of $\Sigma$. 

The maps $\sigma\mapsto L_\sigma$ and $\sigma\mapsto I_\sigma$ are group isomorphisms.
\end{prop}

There is a generalisation of path-diffeomorphisms in the finite path groupoid, which coincide with the graphomorphism presented by Fleischhack in \cite{Fleischhack06}. In this approach the diffeomorphism $\varphi:\Sigma\rightarrow\Sigma$ changes the source and target vertex of a path $\gamma$. Consequently the path-diffeomorphism $(\Phi,\varphi)$, which implements the inner-translation $I_{\Phi}$ in the path groupoid $\PGs$, is a graphomorphism in the context of Fleischhack. Some element of the set of graphomorphisms is directly related to a right-translation $R_\sigma$ in the path groupoid.  Precisely for every $v\in\Sigma$ and $\sigma\in\mathfrak{B}(\PD\Sigma)$ the pairs $(R_{\sigma(v)},t\circ\sigma)$, $(L_{\sigma(v)},t\circ\sigma)$ and $(I_{\sigma(v)},t\circ\sigma)$ define graphomorphism. Furthermore the right-translation $R_{\sigma(v)}$,  the left-translation $L_{\sigma(v)}$ and the inner-translation $I_{\sigma(v)}$ are required to be bijective maps, and hence the maps cannot map non-trivial paths to trivial paths. This property restricts the set of all graphomorphism, which is generated by these translations. In particular in this article graph changing operations, which change the number of edges of a graph, are studied. Hence the left- or right-translation in a finite path groupoid is used in the further development. Notice that in general, these objects do not define graphomorphism.
Finally notice that, in particular for the graphomorphism $(R_{\sigma(v)},t\circ\sigma)$ and a holonomy map for the path groupoid $\PGs$ a similar relation \eqref{eq righttransl} holds. The last equation is fundamental for the construction of $C^*$-dynamical systems, which contain the analytic holonomy $C^*$-algebra restricted to a finite path groupoid $\fPG$ and a point norm continuous action of the finite path-diffeomorphism group $\Diff(V_\Gamma)$ on this algebra. Clearly the right-, left- and inner-translations $R_\sigma$, $L_\sigma$ and $I_\sigma$ are constructed such that \eqref{eq righttransl} generalises. But note that, in the infinite case considered by Fleischhack the action of the bisections $\mathfrak{B}(\PD\Sigma)$ are not point-norm continuous implemented. The advantage of the usage of bisections is that, the map $\sigma\mapsto t\circ\sigma$ is a group morphism between the group $\mathfrak{B}(\PD\Sigma)$ of bisections in $\PGs$ and the group $\Diff(\Sigma)$ of diffeomorphisms in $\Sigma$. Consequently there is an action of the group of diffeomorphisms in $\Sigma$ on the finite path groupoid, which is used to define an action of the group of diffemorphisms in $\Sigma$ on the analytic holonomy $C^*$-algebra. 

\paragraph*{Transformations in finite graph systems\\[5pt]}
To proceed it is necessary to transfer the notion of bisections and right-translations to finite graph systems. A right-translation $R_{\sigma_\Gamma}$  is a mapping that maps graphs to graphs. Each graph is a finite union of independent edges. This causes problems. Since the definition of right-translation in a finite graph system $\PD_\Gamma$ is often not well-defined for all bisections in the finite graph system and all graphs. For example if the graph $\Gamma:=\{\gamma_1,\gamma_2\}$ is disconnected and the bisection $\tilde\sigma$ in the finite path groupoid $\PD_{\Gamma}\Sigma$ over $V_\Gamma$ is defined by $\tilde\sigma(s(\gamma_1))=\gamma_1$, $\tilde\sigma(s(\gamma_2))=\gamma_2$, $\tilde\sigma(t(\gamma_1))=\gamma_1^{-1}$ and $\tilde\sigma(t(\gamma_2))=\gamma_2^{-1}$ where $V_\Gamma:=\{s(\gamma_1),t(\gamma_1),s(\gamma_2),t(\gamma_2)\}$. Let $\idf_\Gamma$ be the set given by the elements $\idf_{s(\gamma_1)}$,$\idf_{s(\gamma_2)}$,$\idf_{t(\gamma_1)}$ and $\idf_{t(\gamma_2)}$. Then notice that, a bisection $\sigma_\Gamma$, which maps a set of vertices in $V_\Gamma$ to a set of paths in $\PD_\Gamma\Sigma$, is given for example by $\sigma_\Gamma(V_\Gamma)= \{\gamma_1,\gamma_2,\gamma_1^{-1},\gamma_2^{-1}\}$. In this case the right-translation $R_{\sigma_\Gamma(V_\Gamma)}(\idf_\Gamma)$ is equivalent to $\{\gamma_1,\gamma_2,\gamma_1^{-1},\gamma_2^{-1}\}$, which is not a set of independent edges and hence not a graph. Loosely speaking the graph-diffeomorphism acts on all vertices in the set $V_\Gamma$ and hence implements four new edges. But a bisection $\sigma_\Gamma$, which maps a subset $V:= \{s(\gamma_1),s(\gamma_2)\}$ of $V_\Gamma$ to a set of paths,  leads to a translation $R_{\sigma_\Gamma(V)}(\{\idf_{s(\gamma_1)},\idf_{s(\gamma_2)}\})=\{\gamma_1,\gamma_2\}$, which is indeed a graph. Set $\Gp:=\{\gamma_1\}$ and $V^\prime=\{s(\gamma_1)\}$. Then observe that, for a restricted bisection, which maps a set $V^\prime$ of vertices in $V_\Gamma$ to a set of paths in $\PD_\Gp\Sigma$, the right-translation become $R_{\sigma_\Gp(V^\prime)}(\{\idf_{s(\gamma_1)}\})=\{\gamma_1\}$, which defines a graph, too. Notice that $\idf_{s(\gamma_1)}$ is a subgraph of $\Gp$. Hence in the simpliest case new edges are emerging. The next definition of the right-tranlation shows that composed paths arise, too.

\begin{defi}\label{defi bisecongraphgroupioid}
Let $\Gamma$ be a graph, $\fPG$ be a finite path groupoid and let $\PD_\Gamma$ be a finite graph system. Moreover the set $V_\Gamma$ is given by $\{v_1,...,v_{2N}\}$.

A \hypertarget{bisection of a finite graph system}{\textbf{bisection of a finite graph system}} $\PD_\Gamma$ is a map $\sigma_\Gamma:V_\Gamma\rightarrow\PD_\Gamma$ such that there exists a bisection $\tilde \sigma\in\mathfrak{B}(\PD_\Gamma\Sigma)$ such that $\sigma_\Gamma(V)=\{\tilde\sigma(v_i):v_i\in V\}$ whenever $V$ is a subset of $V_\Gamma$.

Define a restriction $\sigma_\Gp:V_\Gp\rightarrow\PD_\Gp$ of a bisection $\sigma_\Gamma$ in $\PD_\Gamma$ by
\beqs \sigma_\Gp(V):=\{ \tilde\sigma(w_k) : & w_k\in V\}
\eqs for each subgraph $\Gp$ of $\Gamma$ and $V\subseteq V_\Gp$. 

A \textbf{right-translation in the finite graph system} $\PD_\Gamma$ is a map $R_{\sigma_\Gp}: \PD_\Gp\rightarrow \PD_\Gp$, which is given by a  bisection $\sigma_\Gp:V_\Gp\rightarrow \PD_\Gp$ such that
\beqs
&R_{\sigma_\Gp}(\Gpp)= R_{\sigma_\Gp}(\{\gppe,...,\gppm,\idf_{w_i}:w_i\in\{s(\gpe),...,s(\gpk)\in V^s_{\Gp}:s(\gpi)\neq s(\gpj) \forall i\neq j\}\setminus V_{\Gpp}\})\\[4pt]
&:=\left\{
\begin{array}{ll}
\gppe,...,\gppj,\gppje\circ\tilde\sigma(t(\gppje)),...,\gppm\circ\tilde\sigma(t(\gppm)),\idf_{w_i}\circ\tilde\sigma(w_i) : & \\[4pt]
w_i\in \{s(\gpe),...,s(\gpk)\in V^s_{\Gp}:s(\gpi)\neq s(\gpj) \forall i\neq j\}\setminus V_{\Gpp}, \quad t(\gppi)\neq t(\gppl)\quad\forall i\neq l; i,l\in\bra j+1,M\ket &\\
\end{array}\right\}\\
&=\Gamma^{\prime\prime}_\sigma\\
\eqs where $\tilde\sigma\in\mathfrak{B}(\PD_{\Gamma}\Sigma)$, $K:=\vert\Gp\vert$ and $M:=\vert\Gpp\vert$, $V^s_{\Gp}$ is the set of all source vertices of $\Gp$ and such that $\Gpp:=\{\gppe,...,\gppm\}$ is a subgraph of $\Gp:=\{\gpe,...,\gpk\}$ and $\Gamma^{\prime\prime}_\sigma$ is a subgraph of $\Gp$.
\end{defi} 

Derive that, for $\tilde\sigma(t(\gamma_i))=\gamma_i^{-1}$ it is true that
$(t\circ\tilde\sigma)(s(\gamma_i^{-1}))=s(\gamma_i)=(t\circ\tilde\sigma)(t(\gamma_i))$ holds.

\begin{exa}
Let $\Gamma$ be a disconnected graph.
Then for a bisection $\tilde\sigma\in\mathfrak{B}(\PD_\Gamma\Sigma)$ such that $\sigma(t(\gamma_i))=\gamma_i^{-1}$ for all $1\leq i\leq \vert\Gamma\vert$ it is true that 
\beqs R_{\sigma_\Gamma}(\Gamma)&=\Big\{\gamma_1\circ\tilde\sigma(t(\gamma_1)),...,\gamma_N\circ\tilde\sigma(t(\gamma_N)),\idf_{s(\gamma_1)}\circ\tilde\sigma(s(\gamma_1)),...,\idf_{s(\gamma_N)}\circ\tilde\sigma(s(\gamma_N))\Big\}\\
&=\{\idf_{s(\gamma_1)},...,\idf_{s(\gamma_N)}\}
\eqs yields. Set $\Gp:=\{\gpe,...,\gpm\}$, then derive
\beqs R_{\sigma_\Gamma}(\Gp)&=\Big\{\gpe\circ\tilde\sigma(t(\gpe)),...,\gpm\circ\tilde\sigma(t(\gpm)),\idf_{s(\gamma_1)}\circ\tilde\sigma(s(\gamma_1)),...,\idf_{s(\gamma_{N-M})}\circ\tilde\sigma(s(\gamma_{N-M}))\Big\}\\
&=\{\idf_{s(\gpe)},...,\idf_{s(\gpm)},\gamma_{1},...,\gamma_{N-M}\}
\eqs if $\Gamma=\Gp\cup\{\gamma_1,...,\gamma_{N-M}\}$.
\end{exa}
To understand the definition of the right-translation notice the following problem.

\begin{problem}\label{prob withoutcond} Consider a subgraph $\Gamma$ of $\tilde\Gamma:=\{\gamma_1,\gamma_2,\gamma_3,\gamma_4\}$, a map $\tilde\sigma:V_{\tilde\Gamma}\rightarrow \PD_{\tilde\Gamma}\Sigma$. Then the map 
\beqs R_{\sigma_{\tilde\Gamma}}(\Gamma)=\{\gamma_1\circ\gamma_1^{-1},\gamma_2\circ\idf_{t(\gamma_2)},\gamma_3\circ\idf_{t(\gamma_3)},\idf_{s(\gamma_1)}\circ\gamma_4 \}=:\Gamma_\sigma
\eqs 
 \begin{center}
\includegraphics[width=0.5\textwidth]{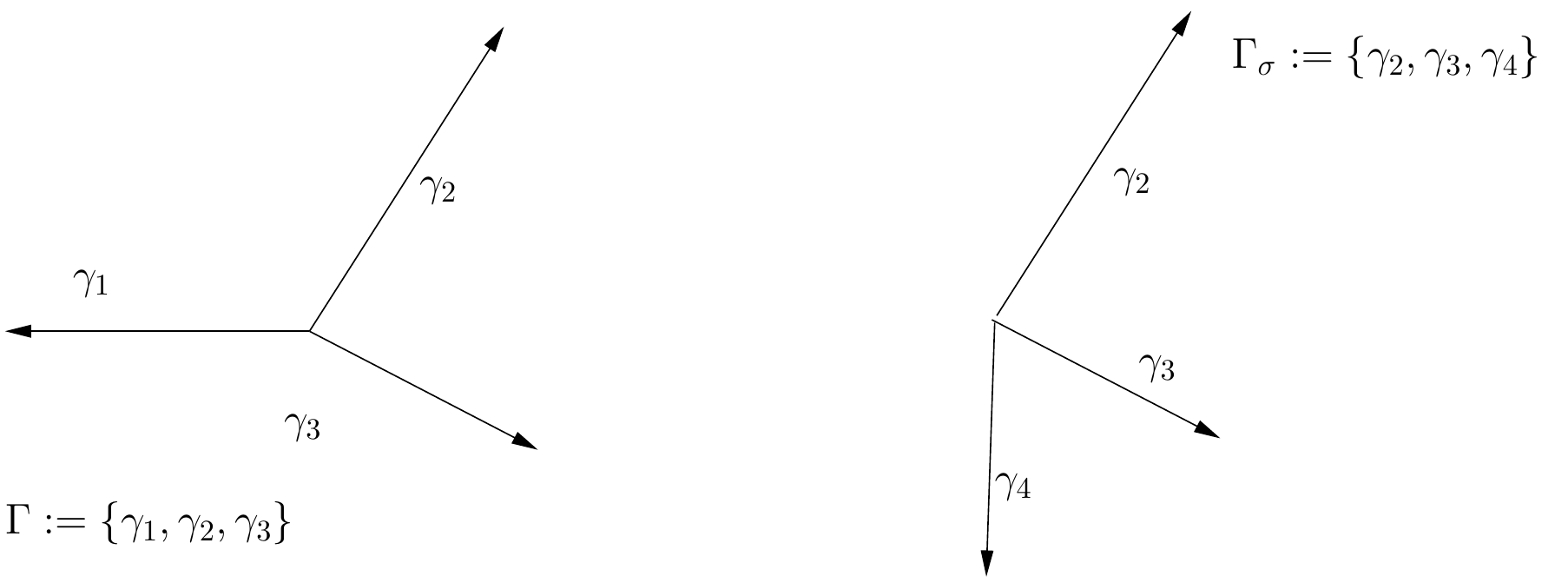}
\end{center} is not a right-translation. This follows from the following fact. Notice that,  the map $\sigma$ maps $t(\gamma_1)\mapsto s(\gamma_1)$, $t(\gamma_2)\mapsto t(\gamma_2)$, $t(\gamma_3)\mapsto t(\gamma_3)$ and $s(\gamma_1)\mapsto t(\gamma_4)$. Then the map $\tilde\sigma$ is not a bisection in the finite path groupoid $\PD_{\tilde\Gamma}\Sigma$ over $V_{\tilde\Gamma}$ and does not define a right-translation $R_{\sigma_{\tilde\Gamma}}$ in the finite graph system $\PD_{\tilde\Gamma}$.

This is a general problem. For every bisection $\tilde\sigma$ in a finite path groupoid such that a graph $\Gamma:=\{\gamma\}$ is translated to $\{\gamma\circ\tilde\sigma(t(\gamma),\tilde\sigma(s(\gamma))\}$. Hence either such translations in the graph system are excluded or the definition of the bisections has to be restricted to maps such that the map $t\circ\tilde\sigma$ is not bijective.  Clearly, the restriction of the right-translation such that $\Gamma$ is mapped to $\{\gamma\circ\tilde\sigma(t(\gamma),\idf_{s(\gamma)}\}$ implies that a simple path orientation transformation is not implemented by a right-translation.

Furthermore there is an ambiguity for graph containing to paths $\gamma_1$ and $\gamma_2$ such that $t(\gamma_1)=t(\gamma_2)$. Since in this case a bisection $\sigma$, which maps $t(\gamma_1)$ to $t(\gamma_3)$, the right-translation is $\{\gamma_1\circ\gamma_3,\gamma_2\circ\gamma_3\}$, is not a graph anymore. 
\end{problem}

\begin{exa}
Otherwise there is for example a subgraph $\Gp$ of $\tilde\Gamma:=\{\gamma_1,\gamma_2,\gamma_3,\gamma_4\}$ and a bisection $\tilde\sigma_{\tilde\Gamma}$ such that 
\beqs \Gamma^\prime_\sigma:=\{\gamma_1\circ\gamma_1^{-1},\gamma_2\circ\idf_{s(\gamma_2)},\gamma_3\circ\idf_{s(\gamma_3)}
\}
\eqs
 \begin{center}
\includegraphics[width=0.5\textwidth]{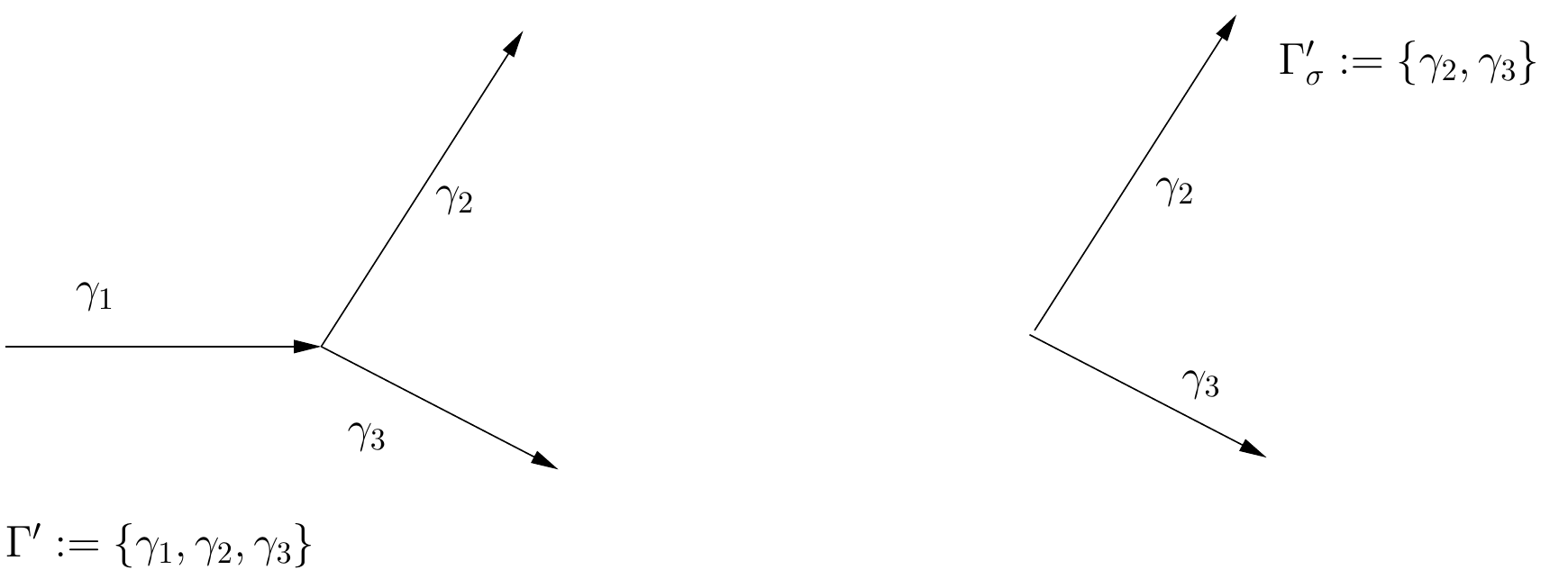}
\end{center} Notice that, $t(\gamma_1)\mapsto s(\gamma_1)$, $t(\gamma_2)\mapsto t(\gamma_2)$, $t(\gamma_3)\mapsto t(\gamma_3)$ and $t(\gamma_4)\mapsto t(\gamma_4)$. Hence the the map $\tilde\sigma_{\tilde\Gamma}:V_{\tilde\Gamma}\rightarrow\PD_{\tilde\Gamma}\Sigma$ is bijective map and consequently a bisection. The bisection $\sigma_{\tilde\Gamma}$ in the graph system $\PD_{\tilde\Gamma}$ defines a right-translation $R_{\sigma_{\tilde\Gamma}}$ in $\PD_{\tilde\Gamma}$.

Moreover for a subgraph $\Gpp:=\{\gamma_2,\gamma_3\}$ of the graph $\breve\Gamma:=\{\gamma_1,\gamma_2,\gamma_3\}$ there exists a map $\sigma_{\breve \Gamma}:V_{\breve\Gamma}\rightarrow\PD_{\breve\Gamma}$ such that
\beqs R_{\sigma_{\breve\Gamma}}(\Gpp)=\{\gamma_2,\gamma_3,\tilde\sigma(s(\gamma_1))\}=\{\gamma_2,\gamma_3,\gamma_1\}
\eqs
\begin{center}
\includegraphics[width=0.5\textwidth]{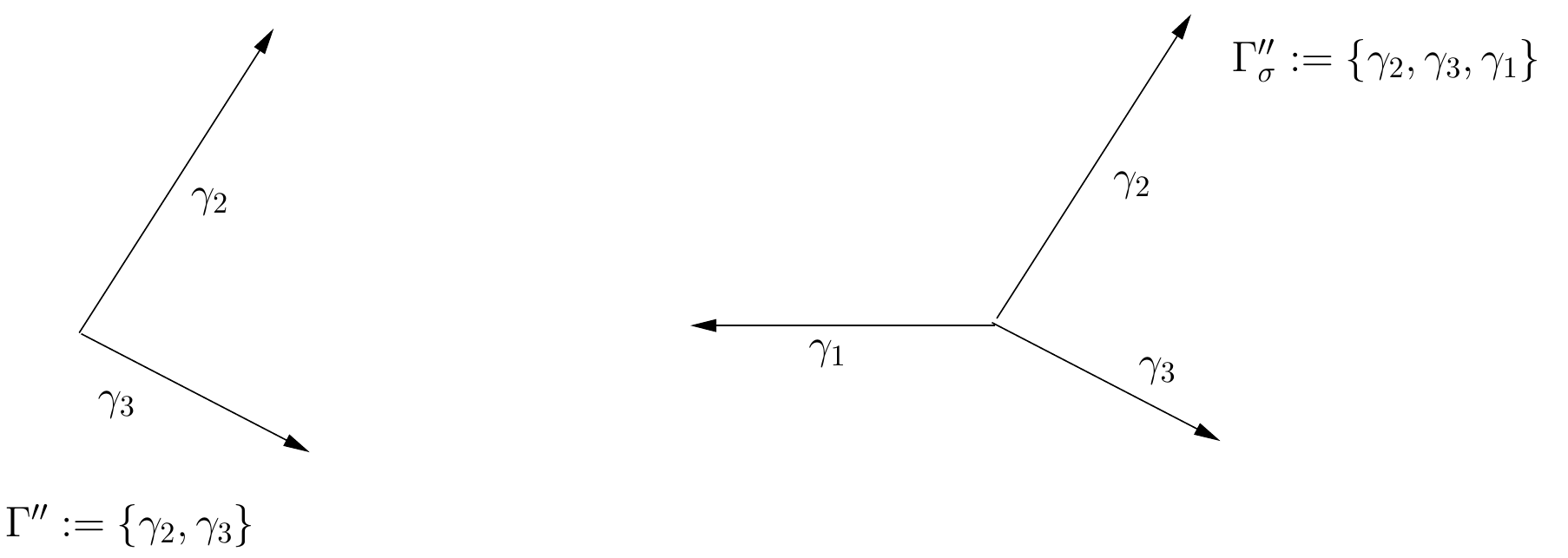}
\end{center}
where $t(\gamma_2)\mapsto t(\gamma_2)$, $t(\gamma_3)\mapsto t(\gamma_3)$ and $s(\gamma_1)\mapsto t(\gamma_1)$. Consequently in this example the map $\tilde\sigma_{\breve \Gamma}$ is a bisection, which defines a right-translation in $\PD_{\breve\Gamma}$.

Note that, for a graph $\Gamma$ such that $\tilde\Gamma$ and $\breve\Gamma$ are subgraphs the bisection $\sigma_{\tilde\Gamma}$ extends to a bisection $\sigma$ in $\PD_\Gamma$ and $\sigma_{\breve\Gamma}$ extends to a bisection $\breve\sigma$ in $\PD_\Gamma$.
\end{exa}

Moreover the bisections of a finite graph system are transfered, analogously, to bisections of a finite path groupoid $\fPG$ to the group $G^{\vert\Gamma\vert}$. Let $\sigma\in\mathfrak{B}(\PD_\Gamma)$ and $(\ho_\Gamma,h_\Gamma)\in\Hom(\PD_\Gamma,G^ {\vert\Gamma\vert})$. Thus there are two maps
\beq \ho_\Gamma\circ R_\sigma:\PD_\Gamma\rightarrow G^{\vert\Gamma\vert}\text{ and }h_\Gamma\circ(t\circ\sigma):V_\Gamma\rightarrow \{e_G\}
\eq which defines a \hyperlink{holonomy map for a finite graph system}{holonomy map for a finite graph system} if $\sigma$ is suitable.

Now, a similar right-translation in a finite graph system in comparison to the right-translation $R_{\sigma(v)}$ in a finite path groupoid is studied.
Let $\sigma_\Gp:V_\Gp\rightarrow\PD_\Gp$ be a restriction of $\sigma_\Gamma\in\mathfrak{B}(\PD_\Gamma)$.  Moreover let $V$ be a subset of $V_ \Gp$, let $\Gpp$ be a subgraph of $\Gp$ and $\Gppp$ be a subgraph of $\Gpp$. Then a right-translation is given by
\beqs &R_{\sigma_\Gp(V)}(\Gpp)\\
&:= 
\left\{\begin{array}{ll}
 R_{\sigma_\Gp}(\{\gppe,...,\gppm,\idf_{w_i}:w_i\in\{s(\gpe),...,s(\gpk)\in V_{\Gp}:s(\gpi)\neq s(\gpj) \forall i\neq j\}\setminus V_{\Gpp}\}& : V_{\Gpp}\subset V\\
R_{\sigma_\Gp}(\{\gppe,...,\gpppn,\idf_{w_i}:w_i\in\{s(\gpe),...,s(\gpk)\in V_{\Gpp}:s(\gpi)\neq s(\gpj) \forall i\neq j\}\setminus V_{\Gppp}\})& \\
\quad \cup\{\idf_{x_i}:x_i\in V\setminus  V_{\Gpp}\}\cup\{\Gpp\setminus\Gppp\}
&: V_{\Gpp}\not\subset V,V_{\Gppp}\subset V\\
       \end{array}\right.
\eqs Loosely speaking, the action of a path-diffeomorphism is somehow localised on a fixed vertex set $V$. 

For example note that for a subgraph $\Gp:=\{\gamma\circ\gp\}$ of $\Gamma:=\{\gamma,\gp\}$ and a subset $V:=\{t(\gp)\}$ of $V_\Gamma$, it is true that
\beqs (\ho_\Gamma\circ R_{\sigma_{\Gamma}(V)})(\gamma\circ\gp)= (\ho_\Gamma\circ R_{\sigma_{\Gamma}(V)})(\gamma)\cdot (\ho_\Gamma\circ R_{\sigma_{\Gamma}(V)})(\gp)= \ho_\Gamma(\gamma)\cdot (\ho_\Gamma\circ R_{\sigma_{\Gamma}(V)})(\gp)
=\ho_\Gamma(\gamma\circ\gp\circ\sigma(t( \gp)))
\eqs yields whenever $\sigma_{\Gamma}\in\mathfrak{B}(\PD_\Gamma\Sigma)$. For a special bisection $\breve\sigma_\Gamma$ it is true that,
\beqs (\ho_\Gamma\circ R_{\breve\sigma_{\Gamma}})(\gamma)= \ho_\Gamma(\gamma\circ\gp)= (\ho_\Gamma\circ R_{\breve\sigma_{\Gamma}})(\gamma)\cdot (\ho_\Gamma\circ R_{\breve\sigma_{\Gamma}})(\gp)
\eqs holds whenever $\breve\sigma_\Gamma\in\mathfrak{B}(\PD_\Gamma\Sigma)$, $\breve\sigma_\Gamma(t(\gp))=\idf_{t(\gp)}$ and $\breve\sigma_\Gamma(t(\gamma))=\gp$. Let $\tilde\sigma$ be the bisection in the finite path groupoid $\PD_\Gamma\Sigma$ that defines the bisection $\breve\sigma$ in $\PD_\Gamma$. Then the last statement is true, since $R_{\breve\sigma_{\Gamma}}(\gp)=\gp\circ\gp^{-1}$ requires $\tilde\sigma_\Gamma:t(\gp)\mapsto s(\gp)$ and $R_{\breve\sigma_{\Gamma}}(\gamma)=\gamma\circ\gp$ needs $\tilde\sigma_\Gamma: t(\gamma)\mapsto t(\gp)$, where $s(\gp)=t(\gamma)$. Then
$R_{\breve\sigma_\Gamma}(\gamma)$ and $R_{\sigma_\Gamma(t(\gp))}(\gamma)$ coincide if $\breve\sigma_\Gamma(t(\gamma))=\sigma_\Gamma(t(\gamma))$ and $\breve\sigma_\Gamma(t(\gp))=\idf_{t(\gp)}$ holds. 
 
\begin{problem}\label{prob righttranslgrouoid}Let $\Gp$ be a subgraph of the graph $\Gamma$, $\sigma_\Gamma$ be a bisection in $\PD_\Gamma$, $\sigma_\Gp:V_\Gp\rightarrow\PD_\Gp$ be a restriction of $\sigma_\Gamma\in\mathfrak{B}(\PD_\Gamma)$.  Moreover let $V$ be a subset of $V_ \Gp$, let $\Gpp:=\{\gamma\circ\gp\}$ be a subgraph of $\Gp$. Let $(\gamma,\gp)\in\PD_\Gp\Sigma^{(2)}$.

Then even for a suitable bisection $\sigma_\Gp$ in $\PD_\Gamma$ it follows that,
\beq\label{eq ineqsolv} R_{\sigma_{\Gp}(V)}(\gamma\circ\gp)\neq R_{\sigma_{\Gp}(V)}(\gamma)\circ R_{\sigma_{\Gp}(V)}(\gp) 
\eq yields. This is a general problem. In comparison with problem \ref{subsec fingraphpathgroup}.\ref{problem group structure on graphs systems} the multiplication map $\circ$ is not well-defined and hence
 \beqs R_{\sigma_{\Gp}(V)}(\gamma)\circ R_{\sigma_{\Gp}(V)}(\gp)
\eqs  is not well-defined. Recognize that, $R_{\sigma_{\Gp}(V)}:\PD_{\Gamma}\rightarrow\PD_\Gamma$.

Consequently in general it is not true that, 
 \beq\label{eq ineqsolv2}(\ho_\Gamma\circ R_{\sigma_{\Gp}(V)})(\gamma\circ\gp)=\ho_\Gamma(R_{\sigma_{\Gp}(V)}(\gamma)\circ R_{\sigma_{\Gp}(V)}(\gp))=(\ho_\Gamma\circ R_{\sigma_{\Gp}(V)})(\gamma)\cdot (\ho_\Gamma\circ R_{\sigma_{\Gp}(V)})(\gp)
\eq yields.
\end{problem}

With no doubt the left-tranlation $L_{\sigma_\Gp}$ and the inner autmorphisms $I_{\sigma_\Gp}$ in a finite graph system $\PD_\Gamma$ for every $\Gp\in\PD_\Gamma$ are defined similarly.

\begin{defi}Let $\sigma_\Gamma\in \mathfrak{B}(\PD_\Gamma)$ be a bisection in the finite graph system $\PD_\Gamma$. Let $R_{\sigma_\Gamma(V)}$ be a right-tranlation, where $V$ is a subset of $V_\Gamma$.

Then the pair $(\Phi_\Gamma,\varphi_\Gamma)$ defined by $\Phi_\Gamma=R_{\sigma_\Gamma(V)}$ (or, respectively, $\Phi_\Gamma=L_{\sigma_\Gamma(V)}$, or $\Phi_\Gamma=I_{\sigma_\Gamma(V)}$) for a subset $V\subseteq V_\Gamma$ and $\varphi_\Gamma=t\circ \sigma_\Gamma$ is called a \textbf{graph-diffeomorphism of a finite graph system}. 
Denote the set of finite graph-diffeomorphisms by $\Diff(\PD_\Gamma)$.
\end{defi}

Let $\Gp$ be a subgraph of $\Gamma$ and $\sigma_\Gp$ be a restriction of  bisection $\sigma_\Gamma$ in $\PD_\Gamma$. Then for example another graph-diffeomorphism $(\Phi_\Gp,\varphi_\Gp)$ in $\Diff(\PD_\Gamma)$ is defined by $\Phi_\Gp=R_{{\sigma_\Gp}(V)}$ for a subset $V\subseteq V_\Gp$ and $\varphi_\Gp=t\circ \sigma_\Gp$.

Remembering that the set of bisections of a finite path groupoid forms a group (refer \ref{lem groupbisection}) one may ask if the bisections of a finite graph system form a group, too. 

\begin{prop}\label{lemma bisecform}The set of bisections $\mathfrak{B}(\PD_\Gamma)$ in a finite graph system $\PD_\Gamma$ forms a group.
\end{prop}
\begin{proofs}Let $\Gamma$ be a graph and let $V_\Gamma$ be equivalent to the set $\{v_1,...,v_{2N}\}$.
  
First two different multiplication operations are studied. The studies are comparable with the results of the definition \ref{defi bisecongraphgroupioid} of a right-translation in a finite graph system. The easiest multiplication operation is given by $\ast_1$, which is defined by 
\beqs (\sigma\ast_1\sigma^\prime)(V_\Gamma)&:=\{
(\tilde\sigma\ast \tilde\sigma^\prime)(v_1),...,(\tilde\sigma\ast \tilde\sigma^\prime)(v_{2N}):v_i\in V_\Gamma\}
\eqs where $\ast$ denotes the multiplication of bisections on the finite path groupoid $\fPG$. Notice that, this operation is not well-defined in general. In comparison with the definition of the right-translation in a finite graph system one has to take care. First the set of vertices doesn't contain any vertices twice, the map $\sigma$ in the finite path system is bijective, the mapping $\sigma$ maps each set to a set of vertices containing no  
vertices twice and the situation in problem \thesection.\ref{prob withoutcond} has to be avoided.

Fix a bisection $\tilde\sigma$ in a finite path groupoid $\fPG$.  Let $V_{\sigma^\prime}$ be a subset of $V_\Gamma$ where $\Gamma:=\{\gamma_1,...,\gamma_N\}$ and for each $v_i$ in $V_{\sigma^\prime}$ it is true that $v_i\neq v_j$ and $v_i\neq (t\circ\tilde\sigma^\prime)(v_j)$ for all $i\neq j$. Define the set $V_{\sigma,\sigma^\prime}$ to be equal to a subset of the set of all vertices $\{v_k\in V_{\sigma^\prime}: 1\leq k\leq 2N\}$ such that each pair $(v_i,v_j)$ of vertices in $V_{\sigma,\sigma^\prime}$ satisfies $(t\circ(\tilde\sigma\ast\tilde\sigma^\prime))(v_i)\neq (t\circ\tilde\sigma^\prime)(v_j)$ and $(t\circ\tilde\sigma^\prime)(v_i)\neq (t\circ\tilde\sigma^\prime)(v_j)$ for all $i\neq j$.
Define
\beqs W_{\sigma,\sigma^\prime}:=\Big\{
w_i\in \{V_{\sigma}\cap V_{\sigma^\prime}\}\setminus V_{\sigma,\sigma^\prime}:
& (t\circ\tilde\sigma)(w_j)\neq (t\circ\tilde\sigma^\prime)(w_i)\quad\forall i\neq j,\quad 1\leq i,j\leq l\Big\}
\eqs
The set $V_{\sigma,\sigma^\prime,\breve \sigma}$ is a subset of all vertices $\{v_k\in V_{\sigma,\sigma^\prime}: 1\leq k\leq 2N\}$ such that each pair $(v_i,v_j)$ of vertices in $V_{\sigma,\sigma^\prime,\breve\sigma}$ satisfies $(t\circ(\breve\sigma\ast\tilde\sigma\ast\tilde\sigma^\prime))(v_i)\neq (t\circ(\tilde\sigma\ast\tilde\sigma^\prime))(v_j)$ and $(t\circ(\tilde\sigma\ast\tilde\sigma^\prime))(v_i)\neq (t\circ(\tilde\sigma\ast\tilde\sigma^\prime))(v_j)$ for all $i\neq j$.

Consequently define a second multiplication on $\mathfrak{B}(\PD_\Gamma)$ similarly to the operation $\ast_1$. This is done by the following definition. Set
\beqs (\sigma\ast_2\sigma^\prime)(V_{\sigma^\prime})
:=&\left\{
(\tilde\sigma\ast \sigma^\prime)(v_1),...,(\tilde\sigma\ast\sigma^\prime)(v_{k}):
v_1,...,v_k\in V_{\sigma,\sigma^\prime}, 1\leq k\leq 2N\right\}\\
&\cup\left\{ \tilde\sigma(w_1),\sigma^\prime(w_1)...,\tilde\sigma(w_{l}),\sigma^\prime(w_{l}): w_1,...,w_l\in W_{\sigma,\sigma^\prime}, 1\leq l\leq 2N
\right\}\\
&\cup\left\{\idf_{p_1},...,\idf_{p_n}:p_1,...,p_n\in V_{\sigma^\prime}\setminus \{V_{\sigma,\sigma^\prime}\cup W_{\sigma,\sigma^\prime}\}, 1\leq n\leq 2N\right\}
\eqs 

Hence the inverse is supposed to be $\sigma^{-1}(V_\Gamma)=\sigma((t\circ\sigma)^{-1}(V_\Gamma))^{-1}$ such that
\beqs (\sigma\ast_2\sigma^{-1})(V_{\sigma^{-1}})=&\{ (\tilde\sigma\ast \tilde\sigma^{-1})(v_1),...,(\tilde\sigma\ast \tilde\sigma^{-1})(v_{2N}):v_i\in V_{\sigma,\sigma^{-1}}\}\\
&\cup\left\{ \tilde\sigma(w_1),\sigma^\prime(w_1)^{-1}...,\tilde\sigma(w_{l}),\sigma^\prime(w_{l})^{-1}: w_1,...,w_l\in W_{\sigma,\sigma^{-1}}, 1\leq l\leq 2N
\right\}\\
&\cup\left\{\idf_{p_1},...,\idf_{p_n}:p_1,...,p_n\in V_{\sigma^\prime}\setminus \{V_{\sigma,\sigma^{-1}}\cup W_{\sigma,\sigma^{-1}}\}, 1\leq n\leq 2N\right\}
\eqs
\end{proofs} 
Notice that, the problem \thesection.\ref{prob righttranslgrouoid} is solved by a multiplication operation $\circ_2$, which is defined similarly to $\ast_2$. Hence the equality of \eqref{eq ineqsolv} is available and consequently \eqref{eq ineqsolv2} is true. Furthermore a similar remark to \ref{def compositionofdiffeo} can be done.

\begin{exa}
Now consider the following example. Set $\Gp:=\{\gamma_1,\gamma_3\}$, let $\Gamma:=\{\gamma_1,\gamma_2,\gamma_3\}$ and\\ $V_\Gamma:=\{s(\gamma_1),t(\gamma_1),s(\gamma_2),t(\gamma_2),s(\gamma_3),t(\gamma_3): s(\gamma_i)\neq s(\gamma_j),t(\gamma_i)\neq t(\gamma_j)\text{ }\forall i\neq j\}$. \\
Set $V$ be equal to $\{s(\gamma_1),s(\gamma_2),s(\gamma_3)\}$. Take two maps $\sigma$ and $\sigma^\prime$ such that $\sigma^\prime(V)=\{\gamma_1,\gamma_3\}$, $\sigma(V)=\{ \gamma_2\}$, where $(t\circ\tilde\sigma)(s(\gamma_3))=t(\gamma_3)$, $\tilde\sigma^\prime(s(\gamma_3))=\gamma_3$, $\tilde\sigma^\prime(s(\gamma_1))=\gamma_1$ and  $\tilde\sigma(t(\gamma_3))=\gamma_2$. Then $s(\gamma_3)\in V_{\sigma_{\Gp},\sigma^\prime_\Gp}$ and $s(\gamma_1)\in W_{\sigma_{\Gp},\sigma^\prime_\Gp}$.
Derive 
\beqs (\sigma\ast_1\sigma^\prime)(V)= \{\gamma_3\circ\gamma_2,\gamma_1\}
\eqs 
Then conclude that,
\beqs (\sigma\ast_2\sigma^\prime)(V_\Gamma)= \{\gamma_3\circ\gamma_2,\gamma_1\}
\eqs holds. Notice that
\beqs (\sigma\ast_2\sigma^\prime)(V)\neq (\sigma^\prime\ast_2\sigma)(V)= \{\gamma_2,\gamma_1,\gamma_3\}
\eqs is true.
Finally obtain
\beqs (\sigma\ast_2\sigma^{-1})(V_\Gamma)= \{\gamma_3\circ\gamma_3^{-1},\gamma_1\circ\gamma_1^{-1}\}=\{\idf_{s(\gamma_3)},\idf_{s(\gamma_1)}\}
\eqs 
Let $\sigma^\prime(V_\Gamma)=\{\gamma_1,\gamma_3\}$ and $\breve \sigma(V_\Gamma)=\{\gamma_2,\gamma_4\}$. Then notice that,
\beqs (\breve \sigma\ast_1\sigma^\prime)(V_\Gamma)=\{\gamma_3\circ\gamma_2,\gamma_1\}
\eqs and 
\beqs (\breve \sigma\ast_2\sigma^\prime)(V_\Gamma)=\{\gamma_3\circ\gamma_2,\gamma_1,\gamma_4\}
\eqs yields.

Furthermore assume supplementary that $t(\gamma_3)=t(\gamma_1)$ holds. Then calculate the product of the maps $\sigma$ and $\sigma^\prime$:
\beqs (\sigma\ast_1\sigma^\prime)(V)= \{\gamma_3\circ\gamma_2,\gamma_1\circ\gamma_2\}\notin\PD_\Gamma
\eqs  and
\beqs (\sigma\ast_2\sigma^\prime)(V_\Gamma)= \{\idf_{t(\gamma_1)}, \idf_{t(\gamma_3)}\}\in\PD_\Gamma
\eqs
\end{exa}

The group structure of $\mathfrak{B}(\PD_\Gamma)$ transferes to $G$. Let $\tilde\sigma$ be a bisection in the finite path groupoid $\fPSGm$, which defines a bisection $\sigma$ in $\PD_\Gamma$ and let $\tilde\sigma^\prime$ be a bisection in $\fPSGm$, which defines another bisection $\sigma^\prime$ in $\PD_\Gamma$. Let $V_{\sigma,\sigma^\prime}$ be equal to $V_\Gamma$, then derive
\beq &\ho_\Gamma\left((\sigma\ast_2\sigma^\prime)(V_\Gamma)\right) 
= \{\ho_\Gamma((\tilde\sigma\ast\sigma^\prime)(v_1)),...,\ho_\Gamma((\tilde\sigma\ast\sigma^\prime)(v_{2N})) \}\\ 
&=\ho_\Gamma(\sigma^\prime(V_\Gamma)\circ\sigma(t(\sigma^\prime(V_\Gamma))))
=\{\ho_\Gamma(\sigma^\prime(v)\circ\tilde\sigma(t(\sigma^\prime(v_1)))),...,\ho_\Gamma(\sigma^\prime(v_{N})\circ\tilde\sigma(t(\sigma^\prime(v_{N}))))\}\\
&=\{\ho_\Gamma(\sigma^\prime(v))\ho_\Gamma(\tilde\sigma(t(\sigma^\prime(v_1)))),...,\ho_\Gamma(\sigma^\prime(v_{N}))\ho_\Gamma(\tilde\sigma(t(\sigma^\prime(v_{N}))))\}\\
&= \ho_\Gamma(\sigma^\prime(V_\Gamma))\ho_\Gamma(\sigma(V_\Gamma))
\eq
Consequently the right-translation in the finite product $G^{\vert\Gamma\vert}$ is definable.

\begin{defi}Let $\sigma_\Gp$ be in $\mathfrak{B}(\PD_\Gamma)$, $\Gp$ a subgraph of $\Gamma$, $\Gpp$ a subgraph of $\Gp$ and $R_{\sigma_\Gp}$ a right-translation, $L_{\sigma_\Gp}$ a left-translation and $I_{\sigma_\Gp}$ an inner-translation in $\PD_\Gamma$.

Then the \textbf{right-translation in the finite product $G^{\vert\Gamma\vert}$} is given by
\beqs \ho_\Gamma\circ R_{\sigma_\Gp}:\PD_\Gamma\rightarrow G^{\vert\Gamma\vert}, \quad \Gpp\mapsto (\ho_\Gamma\circ R_{\sigma_\Gp})(\Gpp)
\eqs
Furthermore define the \textbf{left-translation in the finite product $G^{\vert\Gamma\vert}$} by
\beqs \ho_\Gamma\circ L_{\sigma_\Gp}:\PD_\Gamma\rightarrow G^{\vert\Gamma\vert},\quad \Gpp\mapsto (\ho_\Gamma\circ L_{\sigma_\Gp})(\Gpp)
\eqs
and the \textbf{inner-translation in the finite product $G^{\vert\Gamma\vert}$}
\beqs \ho_\Gamma\circ I_{\sigma_\Gp}:\PD_\Gamma\rightarrow G^{\vert\Gamma\vert},\quad \Gpp\mapsto (\ho_\Gamma\circ I_{\sigma_\Gp})(\Gpp)
\eqs such that $I_{\sigma_\Gp}=L_{\sigma_\Gp^{-1}}\circ R_{\sigma_\Gp}$.
\end{defi}
\begin{lem}It is true that $R_{\sigma_\Gp\ast_2\sigma_\Gp^\prime}=R_{\sigma_\Gp}\circ R_{\sigma_\Gp^\prime}$, $L_{\sigma_\Gp\ast_2\sigma_\Gp^\prime}=L_{\sigma_\Gp}\circ L_{\sigma_\Gp^\prime}$ and $I_{\sigma_\Gp\ast_2\sigma_\Gp^\prime}=I_{\sigma_\Gp}\circ I_{\sigma_\Gp^\prime}$ for all bisections $\sigma_\Gp$ and $\sigma^\prime_\Gp$ in $\mathfrak{B}(\PD_\Gamma)$.
\end{lem}

There is an action of $\mathfrak {B}(\PD_\Gamma)$ on $G^{\vert \Gamma\vert}$ by
\beqs (\zeta_{\sigma_\Gp}\circ\ho_\Gamma)(\Gpp):= (\ho_\Gamma\circ R_{\sigma_\Gp})(\Gpp)
\eqs whenever $\sigma_\Gp\in \mathfrak {B}(\PD_\Gamma)$, $\Gpp\in\PD_\Gp$ and $\Gp\in \PD_\Gamma$.
Then for another $\breve\sigma\in \mathfrak {B}(\PD_\Gamma)$ it is true that,
\beqs ((\zeta_{\breve\sigma_\Gp}\circ\zeta_{\sigma_\Gp})\circ\ho_\Gamma)(\Gpp)= (\ho_\Gamma\circ R_{\breve\sigma\ast_2\sigma_\Gp})(\Gpp)=(\zeta_{\breve\sigma_\Gp\ast_2\sigma_\Gp}\circ\ho_\Gamma)(\Gpp)
\eqs yields.

Recall that, the map $\tilde\sigma\mapsto t\circ\tilde\sigma$ is a group isomorphism between the group of bisections $\mathfrak {B}(\PD_\Gamma\Sigma)$ and the group $\Diff(V_\Gamma)$ of finite diffeomorphisms in $V_\Gamma$. Therefore if the graphs $\Gp=\Gpp$ contain only the path $ \gamma$, then the action $\zeta_{\sigma_\Gp}$ is equivalent to an action of the finite diffeomorphism group $\Diff(V_\Gamma)$. Loosely speaking, the graph-diffeomorphisms $(R_{\sigma_\Gp(V)},t\circ\sigma_\Gp)$ on a subgraph $\Gpp$ of $\Gp$ transform graphs and respect the graph structure of $\Gp$. The diffeomorphism $t\circ\tilde\sigma$ in the finite path groupoid only implements the finite diffeomorphism in $\Sigma$, but it doesn't adopt any path groupoid or graph preserving structure. Summarising the bisections of a finite graph system respect the graph structure and implement the finite diffeomorphisms in $\Sigma$. There is another reason why the group of bisections is more fundamental than the path- or graph-diffeomorphism group. In \cite{Kaminski1,KaminskiPHD} the concept of  $C^*$-dynamical systems is studied. It turns out that, there are three different $C^*$-dynamical systems, each is build from the analytic holonomy $C^*$-algebra and a point-norm continuous action of the group of bisections of a finite graph system. The actions are implemented by one of the three translations, i.e. the left-, right- or inner-translation in the finite product $G^{\vert\Gamma\vert}$. 

Finally the left or right-translations in a finite path groupoid can be studied in the context of natural or non-standard identification of the configuration space. This new concept leads to two different notions of diffeomorphism-invariant states. The actions of path- and graph-diffeomorphism and the concepts of natural or non-standard identification of the configuration space was not used in the context of LQG before.

\subsection{The Lie algebra-valued quantum flux operators associated to surfaces and graphs}\label{sec fluxdef} 

The quantum analogue of a classical connection $A_a(v)$ is given by the holonomy along a path $\gamma$ and is denoted by $\ho(\gamma)$. The quantum flux operator $E_S(\gamma)$, which replaces the classical flux variable $E(S,f^S)$, is given by a map $E_S$ from a graph to the Lie algebra $\go$. Let $\Exp$ be the exponential map from the Lie algebra $\go$ to $G$ and set $U_t(E_S(\gamma)):=\Exp(tE_S(\gamma))$. Then the quantum flux operator $E_S(\gamma)$ and the quantum holonomies $\ho(\gamma)$ satisfy the following canonical commutator relation: 
\[ E_S(\gamma)\ho(\gamma) = i\frac{\dif}{\dif t}\Big\vert_{t=0}U_t(E_S(\gamma))\ho(\gamma)\] where $\gamma$ is a path that intersects the surface $S$ in the target vertex of the path and lies below with respect to the surface orientation of $S$.

In this section different definitions of the quantum flux operator, which is associated to a fixed surface $S$, are presented. For example, the quantum flux operator $E_S$ is defined to be a map from a graph $\Gamma$ to a direct sum $\go\oplus\go$ of the Lie algebra $\go$ associated to the Lie group $G$. This is related to the fact that, one distinguishes between paths that are ingoing and paths that are outgoing with resepect to the surface orientation of $S$. If there are no intersection points of the surface $S$ and the source or target vertex of a path $\gamma_i$ of a graph $\Gamma$, then the map maps the path $\gamma_i$ to zero in both entries. For different surfaces or for a fixed surface different maps refer to different quantum flux operators.  Furthermore, the quantum flux operators are also defined as maps form the graph $\Gamma$ to direct sum $\E\oplus\E$ of the universal enveloping algebra $\E$ of $\go$.
 
\begin{defi}\label{defi intersefunc}Let $\breve S$ be a finite set $\{S_i\}$ of surfaces in $\Sigma$, which is closed under a flip of orientation of the surfaces. Let $\Gamma$ be a graph such that each path in $\Gamma$ satisfies one of the following conditions 
\begin{itemize}
 \item the path intersects each surface in $\breve S$ in the source vertex of the path and there are no other intersection points of the path and any surface contained in $\breve S$,
 \item the path intersects each surface in $\breve S$ in the target vertex of the path and there are no other intersection points of the path and any surface contained in $\breve S$,
 \item the path intersects each surface in $\breve S$ in the source and target vertex of the path and there are no other intersection points of the path and any surface contained in $\breve S$,
 \item the path does not intersect any surface $S$ contained in $\breve S$.
\end{itemize}

Then define the intersection functions $\iota_L:\breve S\times \Gamma\rightarrow \{\pm 1,0\}$ such that
\beqs \iota_L(S,\gamma):=
\left\{\begin{array}{ll}
1 &\text{ for a path }\gamma\text{ lying above and outgoing w.r.t. }S\\
-1 &\text{ for a path }\gamma\text{ lying below and outgoing w.r.t. }S\\
0 &\text{ the path }\gamma\text{ is not outgoing w.r.t. }S
\end{array}\right.
\eqs
and the intersection functions $\iota_R:\breve S\times \Gamma\rightarrow\{\pm 1,0\}$ such that
\beqs \iota_L(S,\gamma):= \left\{\begin{array}{ll}
-1 &\text{ for a path }\gp\text{ lying above and ingoing w.r.t. }S\\
1 &\text{ for a path }\gp\text{ lying below and ingoing w.r.t. }S\\
0 &\text{ the path }\gp\text{ is not ingoing w.r.t. }S
\end{array}\right.
\eqs whenever $S\in\breve S$ and $\gamma\in\Gamma$.

Define a map $\sigma_L:\breve S\rightarrow \go$ such that
\beqs \sigma_L(S)&=\sigma_L(S^ {-1})
\eqs whenever $S\in\breve S$ and $S^ {-1}$ is the surface $S$ with reversed orientation. Denote the set of such maps by $\breve\sigma_L$. Respectively, the map $\sigma_R:\breve S\rightarrow \go$ such that
\beqs \sigma_R(S)&=\sigma_R(S^ {-1})
\eqs whenever $S\in\breve S$. Denote the set of such maps by $\breve\sigma_R$.
Moreover, there is a map $\sigma_L\times \sigma_R:\breve S\rightarrow \go\oplus\go$ such that
\beqs (\sigma_L,\sigma_R)(S)&=(\sigma_L,\sigma_R)(S^ {-1})
\eqs whenever $S\in\breve S$. Denote the set of such maps by $\breve\sigma$.

Finally, define the \textbf{Lie algebra-valued quantum flux set for paths}
\beqs \gop_{\breve S,\Gamma}
:=\bigcup_{\sigma_L\times\sigma_R\in\breve\sigma}\bigcup_{S\in\breve S}\Big\{& (E^L,E^R)\in\Map(\Gamma,\go\oplus\go): 
&(E^L, E^R)(\gamma):=(\iota_L(S,\gamma)\sigma_L(S),\iota_R(S,\gamma)\sigma_R(S))\Big\}
\eqs
where $\Map(\Gamma,\go\oplus\go)$ is the set of all maps from the graph $\Gamma$ to the direct sum  $\go\oplus\go$ of Lie algebras.   
\end{defi}

Observe that, $(\iota_L\times \iota_R)(S^{-1},\gamma)=(-\iota_L\times -\iota_R)(S,\gamma)$ holds for every $\gamma\in\Gamma$. 

Remark that, the condition $E^L(\gamma)=E^R(\gamma^{-1})$ is not required. 

\begin{exa}\label{exa Exa1}
Analyse the following example. Consider a graph $\Gamma$ and two disjoint surface sets $\breve S$ and $\breve T$.
\begin{center}
\includegraphics[width=0.45\textwidth]{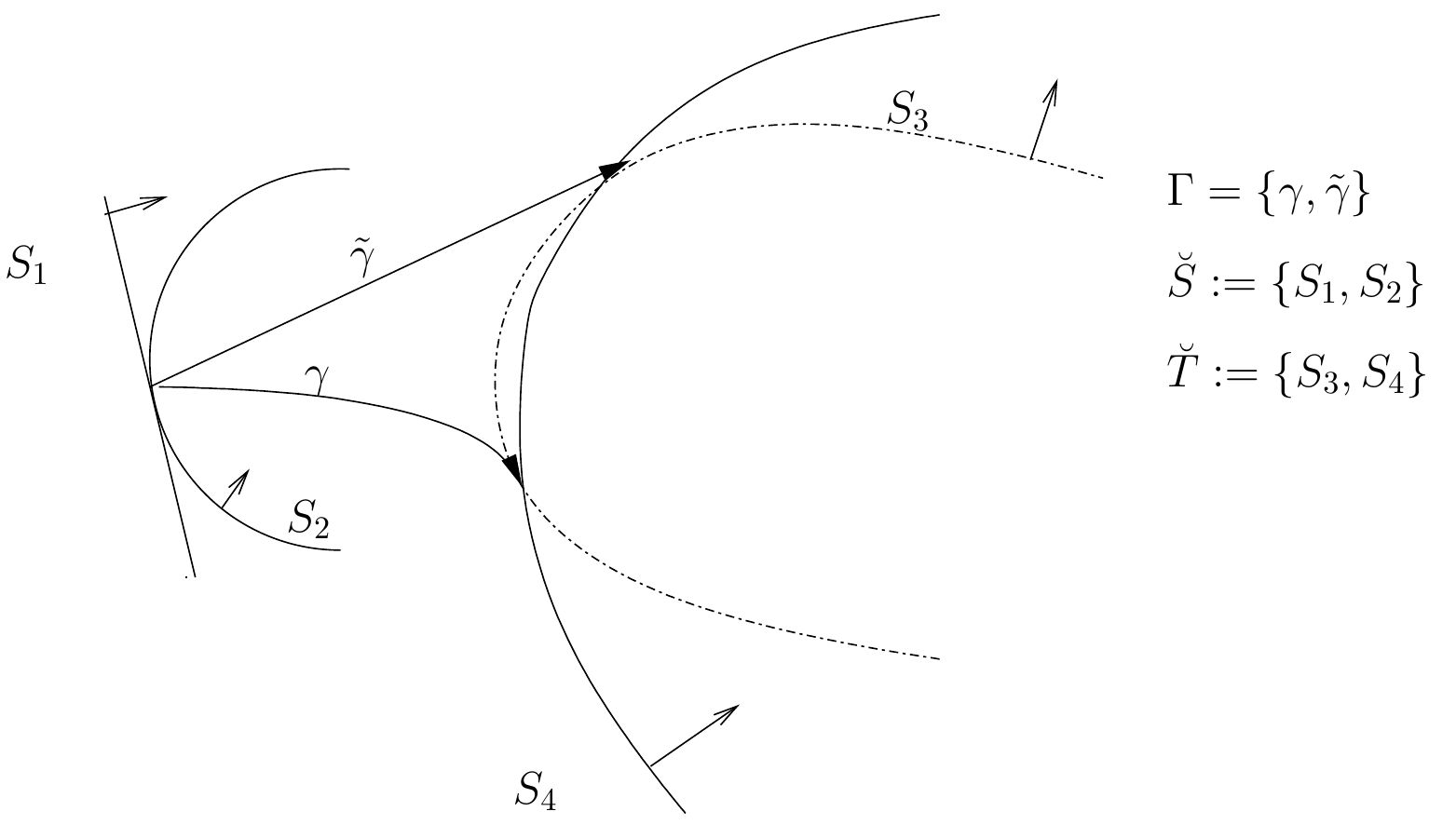}
\end{center}
Then the elements of $\gop_{\breve S,\Gamma}$ are for example given by the maps $E^L_{i}\times E^R_{i}$ for $i=1,2$ such that 
\beqs 
E_1(\gamma)&:= (E^L_{1}, E^R_{1})(\gamma)=(\iota_L(S_1,\gamma)\sigma_L(S_1),\iota_R(S_1,\gamma)\sigma_R(S_1))=(X_{1},0)\\
E_1(\tg)&:= (E^L_{1}, E^R_{1})(\tg)=(\iota_L(S_1,\tg)\sigma_L(S_1),\iota_R(S_1,\tg)\sigma_R(S_1))= (X_1,0)\\
E_2(\gamma)&:= (E^L_{2}, E^R_{2})(\gamma)=(\iota_L(S_2,\gamma)\sigma_L(S_2),\iota_R(S_2,\gamma)\sigma_R(S_2))
=(X_{2},0)\\
E_2(\tg)&:= (E^L_{2}, E^R_{2})(\tg)=(\iota_L(S_2,\tg)\sigma_L(S_2),\iota_R(S_2,\tg)\sigma_R(S_2))
=(X_{2},0)\\
E_3(\gamma)&:= (E^L_{3}, E^R_{3})(\gamma)=(\iota_L(S_3,\gamma)\sigma_L(S_3),\iota_R(S_3,\gamma)\sigma_R(S_3))
=(0,-Y_{3})\\
E_3(\tg)&:= (E^L_{3}, E^R_{3})(\tg)=(\iota_L(S_3,\tg)\sigma_L(S_3),\iota_R(S_3,\tg)\sigma_R(S_3))= (0,-Y_3)\\
E_4(\gamma)&:= (E^L_{4}, E^R_{4})(\gamma)=(\iota_L(S_4,\gamma)\sigma_L(S_4),\iota_R(S_4,\gamma)\sigma_R(S_4))
=(0,Y_{4})\\
E_4(\tg)&:= (E^L_{4}, E^R_{4})(\tg)=(\iota_L(S_4,\tg)\sigma_L(S_4),\iota_R(S_4,\tg)\sigma_R(S_4))= (0,Y_4)
\eqs 

This example shows that, the surfaces $\{S_1,S_2\}$ are similar, whereas the surfaces $\{T_1,T_2\}$ produce different signatures for different paths. Moreover, the set of surfaces are chosen such that one component of the direct sum is always zero. 
\end{exa}

For a particular surface set $\breve S$, the set 
\beqs\bigcup_{\sigma_L\times\sigma_R\in\breve\sigma}\bigcup_{S\in\breve S}
\Big\{ (E^L,E^R)\in\Map(\Gamma,\go\oplus\go): \quad(E^L, E^R)(\gamma):=(\iota_L(S,\gamma)\sigma_L(S),0)\Big\}\eqs can be identified with 
\beqs\bigcup_{\sigma_L\in\breve\sigma_L}\bigcup_{S\in\breve S}\Big\{E\in\Map(\Gamma,\go): \quad
E(\gamma):=\iota_L(S,\gamma)\sigma_L(S)\Big\}
\eqs 
The same is observed for another surface set $\breve T$ and the set $\gop_{\breve T,\Gamma}$ is identifiable with 
\beqs\bigcup_{\sigma_R\in\breve\sigma_R}\bigcup_{T\in\breve T}
\Big\{E\in\Map(\Gamma,\go): \quad
E(\gamma):=\iota_R(T,\gamma)\sigma_R(T)\Big\}
\eqs

The intersection behavoir of paths and surfaces plays a fundamental role in the definition of the flux operator. There are exceptional configurations of surfaces and paths in a graph. One of them is the following.

\begin{defi}
A surface $S$ has the \textbf{surface intersection property for a graph} $\Gamma$ iff the surface intersects each path of $\Gamma$ once in the source or target vertex of the path and there are no other intersection points of $S$ and the path. 
\end{defi}
This is for example the case for the surface $S_1$ or the surface $S_3$, which are presented in example \thesection.\ref{exa Exa1}. Notice that in general, for the surface $S$ there are $N$ intersection points with $N$ paths of the graph. In the example the evaluated map $E_1(\gamma)=(X_1,0)=E_1(\tg)$ for $\gamma,\tg\in\Gamma$ if the surface $S_1$ is considered.

The property of a path lying above or below is not important for the definition of the surface intersection property for a surface. This indicates that the surface $S_4$ in the example \thesection.\ref{exa Exa1} has the surface intersection property, too.

Let a surface $S$ does not have the surface intersection property for a graph $\Gamma$, which contains only one path $\gamma$. Then for example the path $\gamma$ intersects the surface $S$ in the source and target vertices such that the path lies above the surface $S$. Then the map $E^ L\times E^ R$ is evaluated for the path $\gamma$ by
\beqs (E^ L\times E^ R)(\gamma)=(X,-Y)
\eqs
Hence, simply speaking the surface intersection property reduces the components of the map $E^ L\times E^ R$, but for different paths to different components.

Now, consider a bunch of sets of surfaces such that for each surface there is only one intersection point.
\begin{defi}\label{def intprop}
A set $\breve S$ of $N$ surfaces has the \textbf{surface intersection property for a graph $\Gamma$} with $N$ independent edges iff it contain only surfaces, for which each path $\gamma_i$ of a graph $\Gamma$ intersects each surface $S_i$ only once in the same source or target vertex of the path $\gamma_i$, there are no other intersection points of each path $\gamma_i$ and each surface in $\breve S$, and there is no other path $\gamma_j$ that intersects the surface $S_i$ for $i\neq j$ where $1 \leq i,j\leq N$.
\end{defi}
Then for example consider the following configuration.

\begin{exa} 
\begin{center}
\includegraphics[width=0.45\textwidth]{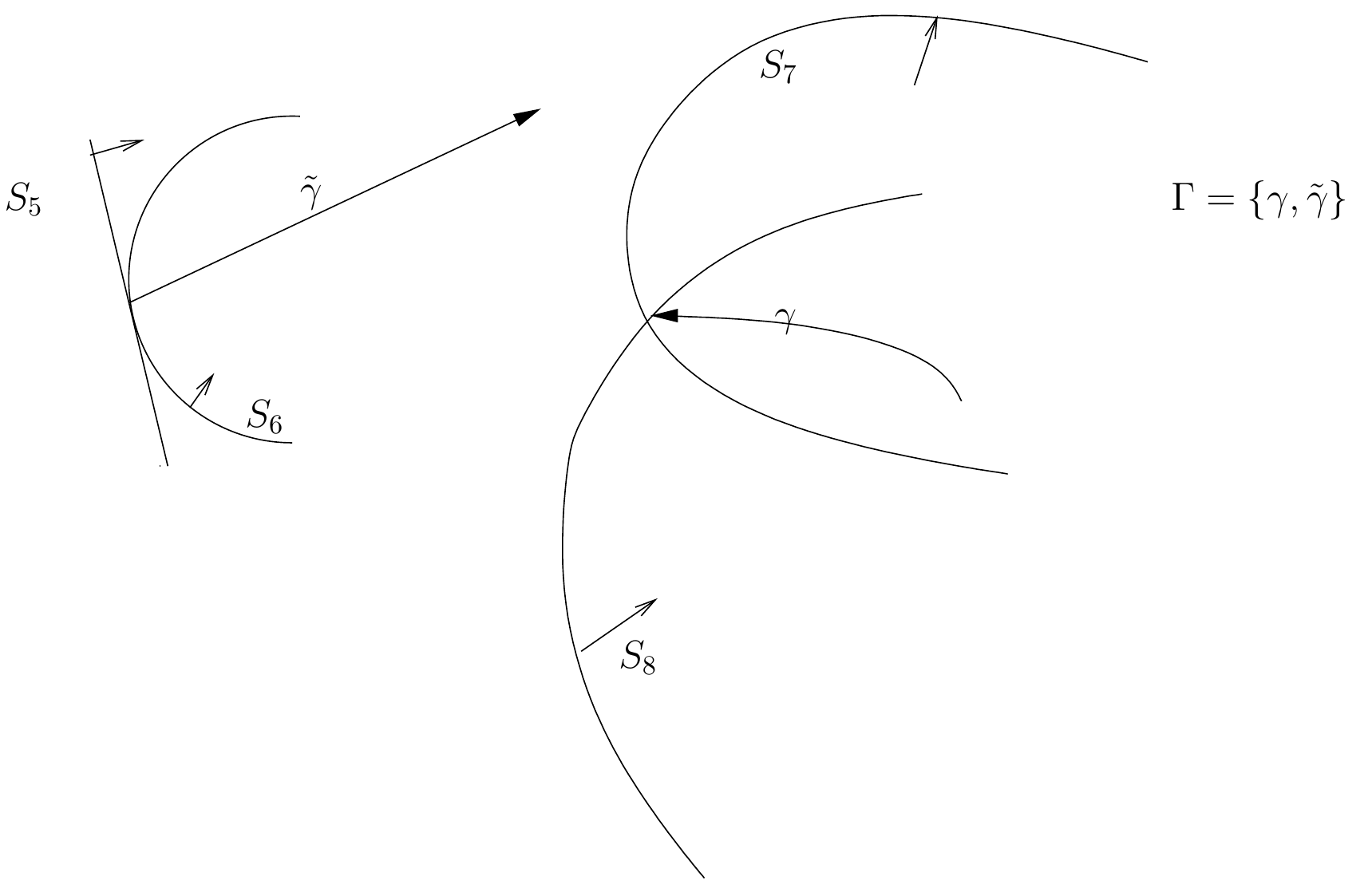}
\end{center} 
The sets $\{S_{6},S_{7}\}$ or $\{S_5,S_{8}\}$ have the surface intersection property for the graph $\Gamma$. 
The images of a map $E$ is
\beqs E_5(\tg)=(X_5,0),\quad E_{8}(\gamma)=(0,Y_{8})
\eqs
\end{exa}
Note that simply speaking, the property indicates that each map reduces to a component of $E^ L\times E^R$ but for different surfaces the map reduces to $E^L$ or $E^R$.

A set of surfaces that has the surface intersection property for a graph can be further specialised by restricting the choice to paths lying ingoing and below with respect to the surface orientations. 
\begin{defi}
A set $\breve S$ of $N$ surfaces has the \hypertarget{simple surface intersection property for a graph}{\textbf{simple surface intersection property for a graph $\Gamma$}} with $N$ independent edges iff it contains only surfaces, for which each path $\gamma_i$ of a graph $\Gamma$ intersects only one surface $S_i$ only once in the target vertex of the path $\gamma_i$, the path $\gamma_i$ lies above and there are no other intersection points of each path $\gamma_i$ and each surface in $\breve S$. 
\end{defi}
\begin{exa}Consider the following example.
\begin{center}
\includegraphics[width=0.45\textwidth]{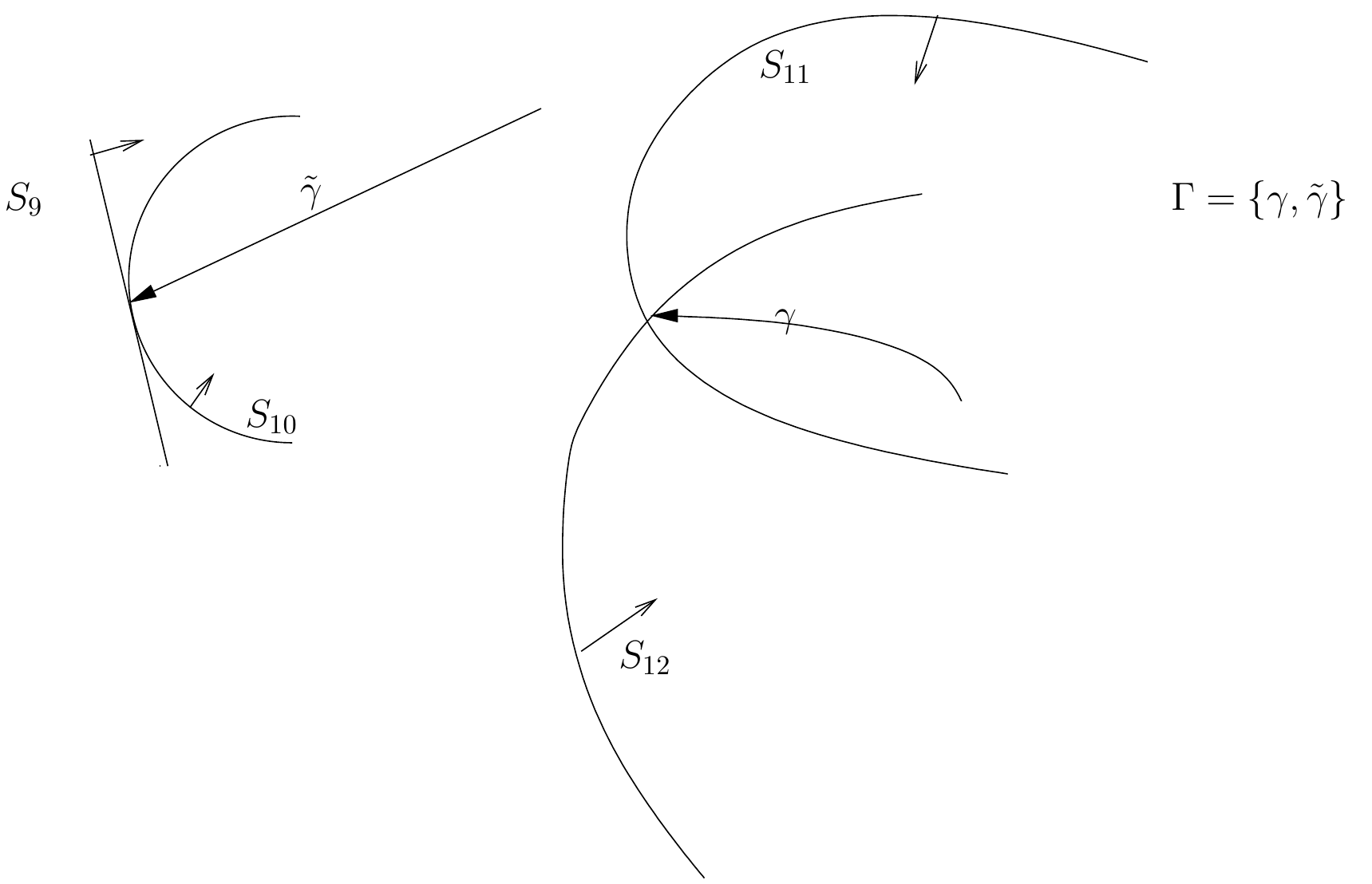}
\end{center} The sets $\{S_{9},S_{11}\}$ or $\{S_{10},S_{12}\}$ have the simple surface intersection property for the graph $\Gamma$.
Calculate
\beqs E_{9}(\tg)=(0,-Y_{9}),\quad E_{11}(\gamma)=(0,-Y_{11})
\eqs
\end{exa}
In this case the set $\gop_{\breve S,\Gamma}$ reduces to
\beqs\bigcup_{\sigma_R\in\breve\sigma_R}\bigcup_{S\in\breve S}\Big\{E\in\Map(\Gamma,\go): \quad
E(\gamma):=-\sigma_R(S)\text{ for }\gamma\cap S= t(\gamma)\Big\}
\eqs Notice that, the set $\Gamma\cap \breve S=\{t(\gamma_i)\}$ for a surface $S_i\in\breve S$ and $\gamma_i\cap S_j\cap S_i =\{\varnothing\}$ for a path $\gamma_i$ in $\Gamma$ and $i\neq j$.

On the other hand, the set of surfaces can be such that each path of a graph intersects all surfaces of the set in the same vertex. This contradicts the assumption that, each path of a graph intersects only one surface once. 
\begin{defi}Let $\Gamma$ be a graph that contains no loops.

A set $\breve S$ of surfaces has the \hypertarget{same intersection property}{\textbf{same surface intersection property for a graph}} $\Gamma$ iff each path $\gamma_i$ in $\Gamma$ intersects with all surfaces of $\breve S$ in the same source vertex $v_i\in V_\Gamma$ ($i=1,..,N$), all paths are outgoing and lie below each surface $S\in\breve S$ and there are no other intersection points of each path $\gamma_i$ and each surface in $\breve S$. 

A surface set $\breve S$ has the \textbf{same right surface intersection property for a graph} $\Gamma$ iff each path $\gamma_i$ in $\Gamma$ intersects with all surfaces of $\breve S$ in the same target vertex $v_i\in V_\Gamma$ ($i=1,..,N$), all paths are ingoing and lie above each surface $S\in\breve S$ and there are no other intersection points of each path $\gamma_i$ and each surface in $\breve S$. 
\end{defi}

Recall the example \thesection.\ref{exa Exa1}. Then the set $\{S_{1},S_{2}\}$ has the same surface intersection property for the graph $\Gamma$.

Then the set $\gop_{\breve S,\Gamma}$ reduces to
\beqs\bigcup_{\sigma_L\in\breve\sigma_L}\bigcup_{S\in\breve S}\Big\{E\in\Map(\Gamma,\go): \quad
E(\gamma):= -\sigma_L(S)\text{ for }\gamma\cap S= s(\gamma)\Big\}
\eqs Notice that, $\gamma\cap S_1\cap ...\cap S_N=s(\gamma)$ for a path $\gamma$ in $\Gamma$, whereas $\Gamma\cap\breve S=\{s(\gamma_i)\}_{1\leq i\leq N}$. Clearly, $\Gamma\cap S_i=s(\gamma_i)$ holds for a surface $S_i$ in $\breve S$. 

Simply speaking the physical intution is that, fluxes associated to different surfaces should act on the same path.
 s
Notice that both properties can be restated for other surface and path configurations. Hence, a surface set can have the simple or same surface intersection property for paths that are outgoing and lie above (or ingoing and below, or outgoing and below). The important fact is related to the question if the intersection vertices are the same for all surfaces or not.

In section \ref{subsec fingraphpathgroup} the concept of finite graph systems is introduced. The following remark shows that, the properties simply generalises to this new structure.
\begin{rem}
A set $\breve S$ has the surface intersection property for a finite graph system $\PD_\Gamma$ iff the set $\breve S$ has the surface intersection property for each subgraph of $\Gamma$ and $\Gamma$. 

A set $\breve S$ has the same surface intersection property for a finite orientation preserved graph system $\PD^{\op}_\Gamma$ associated to a graph $\Gamma$ (with no loops) iff the set $\breve S$ has the same surface intersection property the graph $\Gamma$.

A set $\breve S$ has the simple surface intersection property for a finite orientation preserved\footnote{Let $\breve S$ be equal to $S$. Then  notice that the property of all graphs being orientation preserved subgraphs is necessary, since, for a subgraph $\Gp:=\{\gp\}$ of $\Gamma$ the graph $\{\gp^{-1}\}$ is a subgraph of $\Gamma$, too. Consequently, if there is a surface $S$ intersecting a path $\gp$ such that $\gp$ is ingoing and lies above, then $S$ intersects the path $\gp^{-1}$ such that $\gp^{-1}$ is outgoing and lies above. This implies that, the surface $S$ cannot have the same surface intersection property for each subgraph of $\Gamma$.} graph system $\PD^{\op}_\Gamma$ associated to a graph $\Gamma$ iff the set $\breve S$ has the simple surface intersection property for the graph $\Gamma$.
\end{rem}

\begin{defi}Let $\breve S$ be a surface set and $\Gamma$ be a graph such that the only intersections of the graph and each surface in $\breve S$ are contained in the vertex set $V_\Gamma$.

Then the set of images $\{E(\gamma): E\in\gop_{\breve S,\Gamma}\}$ of flux maps for a fixed path $\gamma$ in $\Gamma$ is denoted by $\bar\gop_{\breve S,\gamma}$. 
\end{defi}

\begin{prop}\label{prop Liealgebrastructfluxes}Let $\breve S$ be a set of surfaces and $\Gamma$ be a fixed graph (with no loops) such that the set $\breve S$ has the same surface intersection property for a graph $\Gamma$. Moreover, let $\breve T$ be a set of surfaces and $\Gamma$ be a fixed graph such that the set $\breve T$ has the simple surface intersection property for a graph $\Gamma$.

Then the set $\bar\gop_{\breve S,\gamma}$ is equipped with a structure, which is induced from the Lie algebra structure of $\go$, such that it forms a Lie algebra. 
The the set $\bar\gop_{\breve T,\gamma}$ is equipped with a structure to form a Lie algebra, too.
\end{prop}
\begin{proofs}
\textbf{Step 1: linear space over $\CB$}\\Consider a path $\gamma$ in $\Gamma$ that lies above and ingoing w.r.t. the surface orientation of each surface $S$ in $\breve S$ and ingoing and above with respect to $T$. Then there is a map $E_S$ such that
\beqs E_S(\gamma)=- X
\eqs There exists an operation $+$ given by the map $s: \bar\gop_{\breve S,\gamma}\times\bar\gop_{\breve S,\gamma} \rightarrow \bar\gop_{\breve S,\gamma}$ such that 
\beqs (E^L_{1}(\gamma), E^L_{2}(\gamma))\mapsto s(E^L_{1}(\gamma), E^L_{2}(\gamma)):=E^L_{1}(\gamma)+ E^L_{2}(\gamma)=-\sigma^1_L(S_1)-\sigma^2_L(S_2)
=-\sigma^3_L([S])
\eqs since $\sigma_L^i\in\breve\sigma_L$ and where $[S]$ denotes an arbitrary representative of the set $\breve S$.
Respectively it is defined 
\beqs (E^L_{1}(\gamma), E^L_{2}(\gamma))\mapsto s(E^L_{1}(\gamma), E^L_{2}(\gamma)):=E^L_{1}(\gamma)+ E^L_{2}(\gamma)=-\sigma_R^1(T)-\sigma^2_R(T)=-\sigma_R^3(T)
\eqs whenever $\sigma^i_R\in\breve\sigma_R$ and $T\in\breve T$.
There is an inverse
\beqs E(\gamma) - E(\gamma)=X - X=0\eqs
and a null element
\beqs E(\gamma) +E_0(\gamma)=X\eqs
whenever $E_0(\gamma)=-\sigma_L(S) =0$.
Notice the following map 
\beq\label{eq Liealghom1}\bar\gop_{\breve S,\gamma}\times\bar\gop_{\breve S,\gp}\ni( E_{1}(\gamma), E_{2}(\gp))\mapsto E_{1}(\gamma)\dot{+} E_{2}(\gp)\in\go
\eq is not considered, since, this map is not well-defined. 
One can show easily that $(\bar\gop_{\breve S,\gamma},+)$ is an additive group. The scalar multiplication is defined by
\beqs \lambda\cdot E(\gamma)=\lambda X
\eqs for all $\lambda\in\CB$ and $X\in\go$.
Finally, prove that $(\bar\gop_{\breve S,\gamma},+)$ is a linear space over $\CB$.

\textbf{Step 2: Lie bracket} is defined by the Lie bracket of the Lie algebra $\go$ and
\beqs \bra E_{1}(\gamma),E_{2}(\gamma)\ket: =\bra X_{1},X_{2}\ket\eqs
for $E_{1}(\gamma),E_{2}(\gamma)\in\bar\gop_{\breve S,\gamma}$ and $\gamma\in\Gamma$. 
\end{proofs}

If a surface set $\breve S$ does not have the same or simple surface intersection property for the graph $\Gamma$, then the surface set can be decomposed into several sets and the graph $\Gamma$ can be decomposed into a set of subgraphs. Then for each modified surface set there is a subgraph such that required condition is fulfilled. 

\begin{defi}Let $\breve S$ a set of surfaces and $\Gamma$ be a fixed graph (with no loops) such that the set $\breve S$ has the same (or simple) surface intersection property for a graph $\Gamma$. 

The universal enveloping Lie algebra of the Lie algebra $\bar\gop_{\breve S,\gamma}$ of electric fluxes for paths of a path $\gamma$ in $\Gamma$ and all surfaces in $\breve S$ is called the \textbf{universal enveloping flux algebra $\bar\Ep_{\breve S,\gamma}$ associated to a path and a finite set of surfaces}.
\end{defi}
For a detailed construction of the universal enveloping flux algebra associated a surface set and a path refer to the section \ref{subsec constrholfluxalg} and definition \ref{def universalliefluxalg}.

Now, the definitions are rewritten for finite orientation preserved graph systems. 

\begin{defi}Let $\breve S$ be a surface set and $\Gamma$ be a graph such that the only intersections of the graph and each surface in $\breve S$ are contained in the vertex set $V_\Gamma$. $\PD_\Gamma$ denotes the finite graph system associated to $\Gamma$. Let $\E$ be the universal Lie enveloping algebra of $\go$.

Define the set of \textbf{Lie algebra-valued quantum fluxes for graphs}
\beqs \go_{\breve S,\Gamma}:= \bigcup_{\sigma_L\times\sigma_R\in\breve\sigma}\bigcup_{S\in\breve S}\Big\{ E_{S,\Gamma}\in\Map(\PD_\Gamma,\bigoplus_{\vert E_\Gamma\vert}\go\oplus \bigoplus_{\vert E_\Gamma\vert}\go):\quad 
&E_{S,\Gamma}:=E_S\times...\times E_S\\&\text{ where }E_S(\gamma):=(\iota_L(\gamma,S)\sigma_L(S),\iota_R(\gamma,S)\sigma_R(S)),\\
&E_S\in\gop_{\breve S,\Gamma},S\in\breve S,\gamma\in\Gamma\Big\}\eqs 

Moreover, define
\beqs \E_{\breve S,\Gamma}:= 
\bigcup_{\sigma_L\times\sigma_R\in\breve\sigma}\bigcup_{S\in\breve S}\Big\{ E_{S,\Gamma}\in\Map(\PD_\Gamma,\bigoplus_{\vert E_\Gamma\vert}\E\oplus \bigoplus_{\vert E_\Gamma\vert}\E):\quad 
&E_{S,\Gamma}:=E_S\times...\times E_S\\&\text{ where }E_S(\gamma):=(\iota_L(\gamma,S)\sigma_L(S),\iota_R(\gamma,S)\sigma_R(S)),\\
&E_S\in\E_{\breve S,\Gamma},S\in\breve S,\gamma\in\Gamma\Big\}
\eqs 

The set of all images of the linear hull of all maps in $\go_{\breve S,\Gamma}$ for a fixed surface set $\breve S$ and a fixed graph $\Gamma$ is denoted by $\bar\go_{\breve S,\Gamma}$.
The set of all images of the linear hull of all maps in $\go_{\breve S,\Gamma}$ for a fixed surface set $\breve S$ and a fixed subgraph $\Gp$ of $\Gamma$ is denoted by $\bar\go_{\breve S,\Gp\leq \Gamma}$. 
\end{defi}

Note that, the set of Lie algebra-valued quantum fluxes for graphs can be generalised for the inductive limit graph system $\PD_{\Gamma_\infty}$. This follows from the fact that each element of the inductive limit graph system $\PD_{\Gamma_\infty}$ is a graph.

\begin{prop}
Let $\breve S$ be a set of surfaces and $\PD^{\op}_\Gamma$ be a finite orientation preserved graph system such that the set $\breve S$ has the same surface intersection property for a graph $\Gamma$ (with no loops).  

The set $\bar\go_{\breve S,\Gamma}$ forms a Lie algebra and is called the \textbf{Lie flux algebra associated a graph and a finite surface set}.The \textbf{universal enveloping flux algebra $\bar\E_{\breve S,\Gamma}$ associated a graph and a finite surface set} is the enveloping Lie algebra of $\bar\go_{\breve S,\Gamma}$.
\end{prop}
\begin{proof}
This follows from the observation that, $\go_{\breve S,\Gamma}$ is identified with
\beqs  \bigcup_{\sigma_L\in\breve\sigma_L}\bigcup_{S\in\breve S}\Big\{ E_{S,\Gamma}\in\Map(\PD^{\op}_\Gamma,\bigoplus_{\vert E_\Gamma\vert}\go):\quad 
&E_{S,\Gamma}:=E_S\times...\times E_S\\&\text{ where }E_S(\gamma):=-\sigma_L(S),
E_S\in\gop_{\breve S,\Gamma},S\in\breve S,\gamma\in\Gamma\Big\}\eqs 
and the addition operation
\beqs
E^1_{S_1,\Gamma}(\Gamma) + E^2_{S_2,\Gamma}(\Gamma)&:=
\big(E^1_{S_1}(\gamma_1) +E^2_{S_2}(\gamma_1),...,E^1_{S_1}(\gamma_N) +E^2_{S_2}(\gamma_N)\big)\\
&=(-\sigma^1_L(S_1)-\sigma^2_L(S_2),...,-\sigma^1_L(S_1)-\sigma^2_L(S_2))\\
&=(E^3_{[S]}(\gamma_1),...,E^3_{[S]}(\gamma_N))
\eqs holds whenever $\Gamma:=\gamma_1,...,\gamma_N$.
\end{proof}
Notice that indeed it is true that,
\beqs \go_{\breve S,\Gamma}=\go_{S_i,\Gamma}
\eqs for every $S_i\in\breve S$. The more general definition is due to physical arguments.

\begin{prop}
Let $\breve T$ be a set of surfaces and $\PD^{\op}_\Gamma$ be a finite orientation preserved graph system such that the set $\breve T$ has the simple surface intersection property for $\Gamma$. 

The set $\bar\go_{\breve T,\Gamma}$ forms a Lie algebra.
\end{prop}
Notice this follows from the fact that, $\go_{\breve T,\Gamma}$ reduces to
\beqs \bigcup_{\sigma_L\in\breve\sigma_L}\Big\{ E_{\breve T,\Gamma}\in\Map(\PD^{\op}_\Gamma,\bigoplus_{\vert E_\Gamma\vert}\go):\quad 
&E_{\breve T,\Gamma}:=E_{T_1}\times...\times E_{T_N}\\&\text{ where }E_{T_i}(\gamma_i):=-\sigma_L(T_i),E_S\in\gop_{\breve S,\Gamma},T_i\in\breve T,\\
&\gamma_i\cap T_i=t(\gamma_i),\gamma\in\Gamma\Big\}\eqs 
since,
\beqs E_{S_1,\Gamma}(\Gamma)+ ...+ E_{S_N,\Gamma}(\Gamma)&= (E_{T_1}(\gamma_1),0,...,0)) + (0,E_{T_2}(\gamma_2),0,...,0)) + ...+ (0,...,0,E_{T_N}(\gamma_N))) \\
&=(E_{T_1}(\gamma_1),...,E_{T_N}(\gamma_N))=:E_{\breve T,\Gamma}(\Gamma)
\eqs yields.

If it is additionally required that ,$E^L_S(\gamma_1)=-E^R_S(\gamma_2)$ holds, then actions of $\bar\go_{\breve S,\Gamma}$ on a configuration space have to be very carefully implemented. 

The Lie flux algebra and the universal enveloping flux algebra for the inductive limit graph system $\PD_{\Gamma_\infty}$ and a fixed suitable surface set $\breve S$ are denoted by $\bar\go_{\breve S}$ and $\bar\E_{\breve S}$. 

\section{The holonomy-flux cross-product $^*$-algebras}\label{subsec constrholfluxalg}
\subsection{The holonomy-flux cross-product $^*$-algebra}

Recall that, there is a big bunch of actions on the naturally identified configuration space $\Ab_\Gamma$, which is connected to the requirement of paths lying above or below the surface and are ingoing or outgoing w.r.t. the surface orientation of a surface $S$. Recall the Lie flux algebra $\bar \go_{\breve S,\Gamma}$, which is given by the evaluation of all maps for a fixed finite orientation preserved graph system associated to a graph $\Gamma$ and a suitable surface set $\breve S$. Refer to \ref{sec fluxdef} for a precise definition. Let $\bar \go_{\breve S,\Gamma}^{\CB}$ be the complexified Lie flux algebra, then $\bar\E_{\breve S,\Gamma}$ denotes the universal enveloping flux algebra. 

For simplicity the investigations start with a graph $\Gamma$, which contains only one path $\gamma$, and one surface $S$. Clearly the following definitions can be generalised to finite orientation preserved graph system associated to an arbitrary graph $\Gamma$ and a suitable surface set $\breve S$. 
\begin{defi}
Let the graph $\Gamma$ contains only a path $\gamma$ and $S$ be a surface such that the path lies below and outgoing w.r.t. the surface orientation of this surface. Set $E_S(\Gamma)=:X_S$. Then the \textbf{right-invariant flux vector field} $e^{\overrightarrow{L}}$ is defined by
\beqs\bra E_S(\Gamma),f_\Gamma\ket:=e^{\overrightarrow{L}}(f_\Gamma)\eqs where
\beq\label{CommRel1} e^{\overrightarrow{L}}(f_\Gamma)(\ho_\Gamma(\gamma))=\frac{\dif}{\dif t}\Big\vert_{t=0} f_\Gamma(\exp(t X_S)\ho_\Gamma(\gamma))&\text{ for }X_S\in\go, \ho_\Gamma(\gamma)\in G, t\in\R
\eq whenever $f_\Gamma\in C^\infty(\Ab_\Gamma)$ and $E_S(\Gamma)\in\bar \go_{S,\Gamma}$. 

Respectively for a path $\gamma$ lying above and outgoing w.r.t. the surface orientation it is true that,
\beq\label{CommRel2} e^{\overleftarrow{L}}(f_\Gamma)(\ho_\Gamma(\gamma))=\frac{\dif}{\dif t}\Big\vert_{t=0} f_\Gamma(\exp(-tX_S)\ho_\Gamma(\gamma)),\text{ for }X_S\in\go,\ho_\Gamma(\gamma)\in G, t\in\R
\eq holds if $-E_S(\Gamma)=:X_S$.
Since $E_S\in\go_{\breve S,\Gamma}$ there exists a skew-adjoint operator $E_S(\Gamma)^+$ that satisfies
\beqs \bra E_S(\Gamma)^+,f_\Gamma\ket &= e^{\overleftarrow{L}}(f_\Gamma)
\eqs where
\beqs \bra E_S(\Gamma)^+,f_\Gamma\ket= \bra E_{S^{-1}}(\Gamma),f_\Gamma\ket\eqs
\end{defi}

A quantum flux operator of a surface $\tilde S$ and a path $\gamma$ lying below and outgoing with respect to the surface orientation of $S$ can be changed by a flip of the path orientation such that the path $\gamma$ lies below and ingoing. 

Recall the map $\breve.:C^\infty(\Ab_\Gamma)\rightarrow C^\infty(\Ab_\Gamma)$ s.t. 
\beqs f_\Gamma(\ho_\Gamma(\gamma_1),...,\ho_\Gamma(\gamma_n))\mapsto \breve{f}_\Gamma(\ho_\Gamma(\gamma_1),...,\ho_\Gamma(\gamma_n)):=f_\Gamma(\ho_\Gamma(\gamma)^{-1},...,\ho_\Gamma(\gamma_n)^{-1})\eqs
\begin{defi}
Define the \textbf{surface and graph orientation flip operator} as a map\\ $\FF:C^\infty(\Ab_\Gamma)\times\bar\go_{\breve S,\Gamma}\rightarrow C^\infty(\Ab_\Gamma)\times\bar\go_{\breve S,\Gamma} $ 
\beqs &\FF(f_\Gamma,E_S(\Gamma))=(\breve f_\Gamma,E_{S^{-1}}(\Gamma))=(\breve f_\Gamma,E_{S}^+(\Gamma)),\quad \FF(f_\Gamma,E_S(\Gamma)^+)=(\breve f_\Gamma,E_{S}(\Gamma))\\
&\FF(f_\Gamma^*,E_S(\Gamma))=(\breve f_\Gamma^*,E_{S^{-1}}(\Gamma))\\[5pt]
& (\FF\circ\pr_1)(f_\Gamma,E_S(\Gamma))=\FF(f_\Gamma)=\breve f_\Gamma,\\
& (\FF\circ\pr_1)(f_\Gamma,E_S(\Gamma))= \FF(E_S(\Gamma))=E_{S^{-1}}(\Gamma)\\
\eqs
\end{defi}
Notice that, $f_\Gamma^*=\overline{f_\Gamma}$ holds whenever $f_\Gamma\in C^\infty(\Ab_\Gamma)$.

\begin{defi}
Let the graph $\Gamma$ contains only a path $\gamma$ and $S$ be a surface such that the path lies below and outgoing w.r.t. the surface orientation of this surface. Set $s(\gamma)=v$ and $-E_S(\Gamma)=:Y_S$.

The \textbf{left-invariant flux vector field} $e^{\overleftarrow{R}}$ is realized as the following commutator
\beqs \bra E_{S}(\Gamma),f_\Gamma\ket =:e^{\overleftarrow{R}}(f_\Gamma)\eqs
where 
\beq\label{CommRel3} e^{\overleftarrow{R}}(f_\Gamma)(\ho_\Gamma(\gamma))=\frac{\dif}{\dif t}\Big\vert_{t=0}  f_\Gamma(\ho_\Gamma(\gamma) \exp(-t Y_S)))\text{ for }Y_S\in\go, t\in\R
\eq 
whenever $f_\Gamma\in C^\infty(\Ab_\Gamma)$ and $E_S\in\go_{\breve S,\Gamma}$. There exists a skew-adjoint operator $E_S(\Gamma)^+$ such that 
\beqs \bra E_S(\Gamma)^+,f_\Gamma\ket &:= e^{\overrightarrow{R}}(f_\Gamma)
\eqs where 
\beq\label{CommRel4} e^{\overrightarrow{R}}(f_\Gamma)(\ho_\Gamma(\gamma))=\frac{\dif}{\dif t}\Big\vert_{t=0}  f_\Gamma(\ho_\Gamma(\gamma) \exp(t Y_S)))\text{ for }Y_S\in\go, t\in\R
\eq 
is satisfieds.
\end{defi}

Summarising, the flux operators are implemented as differential operators $e^{\overrightarrow{L}}$ (or $e^{\overrightarrow{R}}$) on $\Ab_\Gamma$ commuting with the right (or left) shifts.

\begin{defi}Let $A$ be an (associative complex) algebra.

A \textbf{homomorphism of a Lie algebra} $\go$ in $A$ is a map $\tilde\tau:\go\rightarrow A$ such that 
\beqs &\tilde\tau(\alpha X+\beta Y)=\alpha\tilde\tau(X)+\beta\tilde\tau(Y),\\
&\tilde\tau(\bra X,Y\ket)=\tilde\tau(X)\tilde\tau(Y)-\tilde\tau(Y)\tilde\tau(X)
\eqs whenever $X,Y\in\go$ and $\alpha,\beta\in\R$.
\end{defi}

With no doubt, there is a map $\tilde\tau_1:\bar\go_{\breve S,\Gamma}\rightarrow C(\Ab_\Gamma)$ defined by $\tilde\tau_1(E_S(\Gamma))f_\Gamma=\bra E_S(\Gamma),f_\Gamma\ket$ for $f_\Gamma\in C^\infty(\Ab_\Gamma)$, which is indeed a homomorphism of a Lie flux algebra $\bar\go_{\breve S,\Gamma}$ associated to a suitable surface set $\breve S$ and a graph in $C^\infty(\Ab_\Gamma)$.

\begin{lem}\label{lem flip}Fix an element $E_S(\Gamma)\in \bar\go_{\breve S,\Gamma}$. Let $\check\tau_1:C^\infty(\Ab_\Gamma)\rightarrow C^\infty(\Ab_\Gamma)$ be a map such that $\check\tau_1(E_S(\Gamma))(f_\Gamma):=\bra E_S(\Gamma),f_\Gamma\ket$ or, equivalently, $\check\tau_1(E_S(\Gamma))(f_\Gamma):=e^{\overrightarrow{L}}(f_\Gamma)$ for each function $f_\Gamma\in C^\infty(\Ab_\Gamma)$.

Then $\check\tau_1\circ\FF$ defines an $^*$-isomorphism $\FD$ on $C^\infty(\Ab_\Gamma)$ to $C^\infty(\Ab_\Gamma)$ by 
\beqs (\FD\circ e^{\overrightarrow{L}})(f_\Gamma)= \bra E_{S^{-1}}(\Gamma),\breve f_\Gamma\ket
=e^{\overleftarrow{R}}(\breve f_\Gamma)
\eqs which implements a flip of the path orientation.
\end{lem}
\begin{proofs}This is true, since,
\beqs  e^{\overrightarrow{L}}(f_\Gamma)(\ho_\Gamma(\gamma))&=\frac{\dif}{\dif t}\Big\vert_{t=0} f_\Gamma(\exp(tX_S)\ho_\Gamma(\gamma))
=\frac{\dif}{\dif t}\Big\vert_{t=0} \breve f_\Gamma(\ho_\Gamma(\gamma)^{-1}\exp(-tX_S)) \\&=e^{\overleftarrow{R}}(\breve{f}_\Gamma)(\ho_\Gamma(\gamma)^{-1}) 
\eqs and
\beqs (\check\tau_1\circ\FF)( E_{S_0}(\Gamma))( \bra E_S(\Gamma),f_\Gamma\ket)&= e^{\overrightarrow{L}}(f_\Gamma)(\ho_\Gamma(\gamma))= e^{\overleftarrow{R}}(\breve{f}_\Gamma)(\ho_\Gamma(\gamma)^{-1})\\&= \bra (\check\tau_1\circ\FF)(E_S(\Gamma))(\idf_\Gamma),(\check\tau_1\circ\FF)( E_{S_0}(\Gamma))(f_\Gamma)\ket
\eqs hold where $S_0$ is a surface, which does not intersect any path in $\Gamma$ and $\idf_\Gamma$ is the constant function for any $\Gamma$.
\end{proofs}

There is also an $^*$-isomorphism $\tilde\FD$ presented by
\beqs (\tilde\FD\circ e^{\overrightarrow{L}})(f_\Gamma)= \bra E_{S^{-1}}(\Gamma), f_\Gamma\ket
=e^{\overrightarrow{R}}(f_\Gamma)
\eqs connected to a flip of the path and surface orientation.

Now, the focus lies on quantum fluxes, which takes values in the enveloping algebra of $\go$.
 
\begin{defi}\label{def universalliefluxalg}Let $\breve S$ be a surface set such that $\breve S$ has the same surface intersection property for a graph $\Gamma$.

Then the \textbf{tensor algebra of flux operators} is defined by \[\TD(\breve S):= \bigoplus_{k=0}^\infty \bar\go_{\breve S,\Gamma}^{\CB\text{ } \otimes_k}\] 
There is a natural inclusion $j:\bar\go^{\CB}_{\breve S,\Gamma}\rightarrow \TD(\breve S)$, $E_S(\Gamma)\mapsto (E_S(\Gamma))^{\otimes^1}$. Denote by $\bar\E_{\breve S,\Gamma}$ the \textbf{universal enveloping $^*$-algebra for flux operators} generated by the quotient of $\TD(\breve S)$ and a two sided ideal $I$ expressed by
\beqs I=\Big\{&j(E_{S_1}(\Gamma))\otimes j(E_{S_2}(\Gamma))-j(E_{S_2}(\Gamma))\otimes j(E_{S_1}(\Gamma))-j(\bra E_{S_1}(\Gamma),E_{S_2}(\Gamma)\ket):\\&E_{S_1},E_{S_2}\in\go^{\CB}_{\breve S},S_K\in\breve S,K=1,2\Big\}\eqs

The antilinear and antimultiplicative involution $^+$ is given by 
\beqs &(E_{S_1}(\Gamma)\times ... \times E_{S_k}(\Gamma))^+=E_{S_k}^+(\Gamma)\times ...\times E_{S_1}^+(\Gamma),\\
&E_{S_K}(\Gamma)^+=-E_{S_K}(\Gamma)\text{ for }E_{S_K}(\Gamma)\in\bar\go_{\breve S,\Gamma}\text{ and }K=1,...,k
\eqs
\end{defi}

Recall the structure of the enveloping algebra of $\go$. 
Moreover there is a bilinear map \[\tau_1 :C^\infty(\Ab_{\Gamma})\times\bar\go^{\CB}_{\breve S,\Gamma}\rightarrow C^\infty(\Ab_{\Gamma})\] such that for $(f_\Gamma,E_S(\Gamma))\in C^\infty(\Ab_\Gamma)\times\bar\go^{\CB}_{\breve S,\Gamma}$ such that it is true that,
\beq \tau_1(f_\Gamma,E_S(\Gamma)) = \bra E_S(\Gamma),f_\Gamma\ket
\eq holds, which is further generalised to \[\tau_2 :C^\infty(\Ab_{\Gamma})\times\bar\go^{\CB}_{\breve S,\Gamma}\otimes\bar\go^{\CB}_{\breve S,\Gamma}\rightarrow C^\infty(\Ab_{\Gamma})\] such that
\beq \tau_2(f_\Gamma,E_{S_1}(\Gamma)\cdot E_{S_2}(\Gamma)) =- \bra E_{S_1}(\Gamma), \bra E_{S_2}(\Gamma),f_\Gamma\ket\ket
\eq 

Hence in general there is a bilinear map $\tau:C^\infty(\Ab_{\Gamma})\times\bar\E_{\breve S,\Gamma}\rightarrow C^\infty(\Ab)$ 
\beq &\tau (f_\Gamma,E_{S_1}(\Gamma)\cdot...\cdot E_{S_n}(\Gamma))= \bra E_{S_1}(\Gamma),\bra ...,\bra E_{S_n}(\Gamma),f_\Gamma\ket\ket ... \ket\eq such that $\tilde\tau:\bar\E_{\breve S,\Gamma}\rightarrow C^\infty(\Ab)$ where $\tilde\tau(E_S(\Gamma))f_\Gamma=\bra E_S(\Gamma),f_\Gamma\ket$ for $f_\Gamma\in C^\infty(\Ab_\Gamma)$ is a unit-preserving homomorphism. 

The following corollary implies that due to the universality structure of $\bar\E_{\breve S}$, this map $\tilde\tau$ is unique.
\begin{cor}\label{cor uniquhom} 
Let $A$ be a unital algebra and $\tilde\tau$ be a homomomorphism of a Lie algebra $\go$ into $A$. Then there exsits a unique unit-preserving homomorphism of the universal enveloping flux algebra $\E$ of $\go$ into $A$ which extends $\tilde\tau$. 
\end{cor}

\begin{lem}Let $\breve S$ be a set of surfaces which has the same intersection surface property for a finite orientation preserved graph system associated to $\Gamma$.

Then $C^\infty(\Ab_\Gamma)$ is a left $\bar\E_{\breve S,\Gamma}$-module algebra. The action of $\bar\E_{\breve S,\Gamma}$ on  $C^\infty(\Ab_\Gamma)$ is given by $E_S(\Gamma)\rhd f_\Gamma:= e^{\overrightarrow{L}}(f_\Gamma)$.
\end{lem}
\begin{proofs}This following from the fact that, $C^\infty(\Ab_\Gamma)$ is a left $\bar\E_{\breve S,\Gamma}$-module, which is defined by the map 
\beqs E_S(\Gamma)\rhd f_\Gamma:= e^{\overrightarrow{L}}(f_\Gamma)=\tau(f_\Gamma,E_S(\Gamma))\text{ for }E_S(\Gamma)\in\bar\E_{\breve S,\Gamma},  f_\Gamma\in C^\infty(\Ab_\Gamma)
\eqs which is obviously bilinear and $1\rhd f_\Gamma=f_\Gamma$ is satisfied. Moreover 	
\beqs E_{S_1}(\Gamma) \rhd( E_{S_2}(\Gamma)\rhd f_\Gamma )= (E_{S_1}(\Gamma) \cdot E_{S_2}(\Gamma))\rhd f_\Gamma
\eqs holds. Furthermore it turns out to be left $\bar\E_{\breve S,\Gamma}$-module algebra, since, additionally, 
\beqs &E_S(\Gamma)\rhd (f_\Gamma k_\Gamma)= (e^{\overrightarrow{L}}(f_\Gamma)) k_\Gamma + f_\Gamma (e^{\overrightarrow{L}}(k_\Gamma))\text{ and }\\
&E_S(\Gamma)\rhd \idf_\Gamma= 0 \text{ for all }E_S(\Gamma)\in\bar \go_{\breve S,\Gamma}
\eqs
yields.
\end{proofs}

\begin{lem}Let $\breve S$ be a set of surfaces which has the appropriate same intersection surface property for a finite orientation preserved graph system associated to $\Gamma$.

Then $C^\infty(\Ab_\Gamma)$ is a right $\bar\E_{\breve S,\Gamma}$-module algebra. The action of $\bar\E_{\breve S,\Gamma}$ on  $C^\infty(\Ab_\Gamma)$ is given by $E_S(\Gamma)\lhd f_\Gamma:= e^{\overrightarrow{R}}(f_\Gamma)$.
\end{lem}

Finally, the definition of the holonomy-flux $^*$-algebra in LQG by the authors \cite[Def.2.7]{LOST06} is rewritten for the case of a fixed graph $\Gamma$. The vector space $C^\infty(\Ab_\Gamma)\otimes\bar\E_{\breve S,\Gamma}$ is equipped with the multiplication
\beq\label{ eq multiholfluxalg} (f_\Gamma^1\otimes E_{S_1}(\Gamma))\cdot (f_\Gamma^2\otimes E_{S_2}(\Gamma)) 
= - \tau(f_\Gamma^2,E_{S_1}(\Gamma))\otimes E_{S_1}(\Gamma)\cdot E_{S_2}(\Gamma)
\eq such that a Lie algebra bracket is derived
\beq\label{eq LiealgstructureHolFlux}
&\bra f_\Gamma^1\otimes E_{S_1}(\Gamma), f_\Gamma^2\otimes E_{S_2}(\Gamma)\ket =\\
&\qquad-(\tau (f_\Gamma^2,E_{S_1}(\Gamma)) -\tau(f_\Gamma^1,E_{S_2}(\Gamma)))F\otimes \bra E_{S_1}(\Gamma),E_{S_2}(\Gamma)\ket
\eq Notice if $S_1$ and $S_2$ are disjoint the commutator on $C^\infty(\Ab_\Gamma)\otimes\bar\E_{\breve S,\Gamma}$ is zero.
Calculate the commutator 
\beq\label{eq commutatordisjointelements} & \bra (f_\Gamma\otimes \idf), (\idf \otimes E_{S}(\Gamma))\ket
= -\tau(f_\Gamma,E_{S}(\Gamma))\otimes E_{S}(\Gamma)
\eq
Additionally, the algebra is equipped with an involution such that this algebra is a unital associative $^*$-algebra.
In this work the algebra is slightly modificated. 

\begin{defi}Let $\breve S$ be a set of surfaces which has the appropriate same intersection surface property for a finite orientation preserved graph system associated to $\Gamma$.

The \textbf{holonomy-flux cross-product $^*$-algebra associated a graph $\Gamma$ and a surface set $\breve S$} is given by the left or right cross-product $^*$-algebra
\[C^\infty(\Ab_\Gamma)\rtimes_{L}\bar\E_{\breve S,\Gamma} \text{ or } C^\infty(\Ab_\Gamma)\rtimes_{R}\bar\E_{\breve S,\Gamma} \] which is defined by the vector space $C^\infty(\Ab_\Gamma)\otimes\bar\E_{\breve S,\Gamma}$ with the multiplication given by
\beqs (f^1_\Gamma\otimes E_{S_1}(\Gamma))\cdot_L(f^2_\Gamma\otimes E_{S_2}(\Gamma))
=f_\Gamma^1(E_{S_1}(\Gamma)\rhd f^2_\Gamma)\otimes E_{S_2}(\Gamma) 
+ f_\Gamma^1 f^2_\Gamma \otimes E_{S_1}(\Gamma)\cdot E_{S_2}(\Gamma)
\eqs or respectively
\beqs (f^1_\Gamma\otimes E_{S_1}(\Gamma))\cdot_R(f^2_\Gamma\otimes E_{S_2}(\Gamma))
=(E_{S_2}(\Gamma)\lhd f^1_\Gamma)f_\Gamma^2\otimes E_{S_1}(\Gamma)
+ f^1_\Gamma f_\Gamma^2\otimes E_{S_1}(\Gamma)\cdot E_{S_2}(\Gamma)
\eqs
and the involution
\beqs (f_\Gamma\rhd E_S(\Gamma))^*=\bar f_\Gamma\rhd E_S(\Gamma)^+
\eqs
or respectively
\beqs (f_\Gamma\lhd E_S(\Gamma))^*=\bar f_\Gamma\lhd E_S(\Gamma)^+
\eqs whenever $E_{S_1}(\Gamma),E_{S_2}(\Gamma),E_{S}(\Gamma)\in\bar\E_{\breve S,\Gamma}$ and $f^1_\Gamma,f^2_\Gamma,f_\Gamma\in C^\infty(\Ab_\Gamma)$.

The \textbf{holonomy-flux cross-product $^*$-algebra associated a surface set $\breve S$} is given by the left or right cross-product $^*$-algebra
\[C^\infty(\Ab)\rtimes_{L} \bar \E_{\breve S}\text{ or } C^\infty(\Ab)\rtimes_{R} \bar \E_{\breve S}\]
which are the inductive limit of the families $\{(C^\infty(\Ab_\Gamma)\rtimes_{L}\bar\E_{\breve S,\Gamma},\beta_{\Gamma,\Gp}\times\check\beta_{\Gamma,\Gp})\}$ or $\{(C^\infty(\Ab_\Gamma)\rtimes_{R}\bar\E_{\breve S,\Gamma},\beta_{\Gamma,\Gp}\times\check\beta_{\Gamma,\Gp})\}$ where $\check\beta_{\Gamma,\Gp}: \bar \E_{\breve S,\Gamma}\rightarrow \bar \E_{\breve S,\Gp}$ are suitable unit-preserving $^*$-homomorphisms for a suitable set $\breve S$ of surfaces that preserve the left or right vector field structure. 
\end{defi}

Summarising, the unital holonomy-flux cross-product $^*$-algebra $C^\infty(\Ab)\rtimes_L\bar \E_{\breve S}$ may be thought of as the universal algebra generated by $C^\infty(\Ab_\Gamma)$ and $\bar\E_{\breve S,\Gamma}$ with respect to the commutator relation
\beq E_S(\Gamma)f_\Gamma= E_S(\Gamma)\rhd f_\Gamma + f_\Gamma E_S(\Gamma)
\eq

Derive for suitable surface $S$ and a graph $\Gamma$ the following commutator relation between elements of the holonomy-flux cross-product $^*$-algebra
\beq &\bra (f^1_\Gamma\otimes E_{S_1}(\Gamma)), (f^2_\Gamma\otimes E_{S_2}(\Gamma))\ket\\
&=
f_\Gamma^1(E_{S_1}(\Gamma)\rhd f^2_\Gamma)\otimes E_{S_2}(\Gamma) + f_\Gamma^1 f^2_\Gamma \otimes E_{S_1}(\Gamma)\cdot E_{S_2}(\Gamma)\\&\quad  
- f_\Gamma^2(E_{S_2}(\Gamma)\rhd f^1_\Gamma)\otimes E_{S_1}(\Gamma)- f_\Gamma^2 f^1_\Gamma\otimes E_{S_2}(\Gamma)\cdot E_{S_1}(\Gamma)\\
&= 
f_\Gamma^1(E_{S_1}(\Gamma)\rhd f^2_\Gamma)\otimes E_{S_2}(\Gamma)   
- f_\Gamma^2(E_{S_2}(\Gamma)\rhd f^1_\Gamma)\otimes E_{S_1}(\Gamma)
+ f_\Gamma^1f^2_\Gamma \otimes \bra E_{S_1}(\Gamma),E_{S_2}(\Gamma)\ket
\eq which can be compared to the definition usually used in LQG, which is illustrated in \eqref{eq LiealgstructureHolFlux}. The definitions do not coincide, since, in LQG the ACZ- holonomy-flux algebra $C^\infty(\Ab_\Gamma)\otimes\bar\E_{\breve S,\Gamma}$ is defined by the multiplication \eqref{ eq multiholfluxalg}. This shows that, the holonomy-flux cross-product $^*$-algebra is a modificated holonomy-flux $^*$-algebra if it is compared with the $^*$-algebra presented in \cite{LOST06}. 

Notice
\beq &\bra (f^1_\Gamma\otimes \idf), (f^2_\Gamma\otimes E_{S}(\Gamma))\ket\\
&=
f_\Gamma^1 f^2_\Gamma\otimes E_{S}(\Gamma) + f_\Gamma^1 f^2_\Gamma \otimes E_{S}(\Gamma)
- f_\Gamma^2(E_{S}(\Gamma)\rhd f^1_\Gamma)\otimes \idf- f_\Gamma^2 f^1_\Gamma\otimes E_{S}(\Gamma)\\
&= 
f_\Gamma^1 f^2_\Gamma\otimes E_{S}(\Gamma) 
- f_\Gamma^2(E_{S}(\Gamma)\rhd f^1_\Gamma)\otimes \idf
\eq holds.
Observe that, the commutator 
\beq\label{eq commutatordisjointelements2} &\bra (f_\Gamma\otimes \idf), (\idf \otimes E_{S}(\Gamma))\ket\\
&=
f_\Gamma \otimes E_{S}(\Gamma) + f_\Gamma \otimes E_{S}(\Gamma)  
- (E_{S}(\Gamma)\rhd f_\Gamma)\otimes \idf -  f_\Gamma\otimes E_{S}(\Gamma)\\
&= 
f_\Gamma \otimes E_{S}(\Gamma)  
- (E_{S}(\Gamma)\rhd f_\Gamma)\otimes \idf 
\eq is different from the commutator \eqref{eq commutatordisjointelements} of the holonomy-flux $^*$-algebra.

Clearly for different surface sets there are a lot of different holonomy-flux cross-product $^*$-algebras. For example let $\breve S$ be a set of $N$ surfaces and let $\Gamma$ be a graph with $N$ independent edges such that every surface $S_i$ in $\breve S$ intersects only one path $\gamma_i$ of a graph $\Gamma$ only once in the target vertex of the path $\gamma_i$, the path $\gamma_i$ lies above and there are no other intersection points of each other path $\gamma_j$ and the surface $S_i$ in $\breve S$ ($i\neq j$). Moreover let $\breve T$ be a set of $N$ surfaces and let $\Gamma$ be a graph with $N$ independent edges such that every surface $T_i$ in $\breve T$ intersects only one path $\gamma_i$ of a graph $\Gamma$ only once in the source vertex of the path $\gamma_i$, the path $\gamma_i$ lies below and there are no other intersection points of each other path $\gamma_j$ and the surface $T_i$ in $\breve T$ ($i\neq j$). 

Then the sets $\breve S$ and $\breve T$ have the simple surface intersection property for $\Gamma$. Then there exists two different holonomy-flux cross-product $^*$-algebras $C^\infty(\Ab)\rtimes_L\bar \E_{\breve S}$ and $C^\infty(\Ab)\rtimes_R\bar \E_{\breve T}$. 

Consider $C^\infty(\Ab_{\Gamma})$ as a $^*$-subalgebra of the analytic holonomy $C^*$-algebra $\Alg_\Gamma:=C(\Ab_{\Gamma})$. 
Moreover refer to Sakai \cite{Sakai} or Bratteli and Robinson \cite{BratteliRobinsonB1} for the definition of $^*$-derivations.

\begin{lem}\label{lem diffop} For any graph $\Gamma$ and a surface $S$, which has the same intersection surface property for a finite orientation preserved graph system associated to  $\Gamma$, the object
\beq i \bra E_S(\Gamma)^+E_S(\Gamma),f_{\Gamma}\ket=:\delta^2_{S,\Gamma}(f_{\Gamma})\eq defines a unbounded symmetric $^*$-derivation $\delta_{S,\Gamma}^2$ on $C(\Ab_{\Gamma})$ with domain $C^\infty(\Ab_{\Gamma})$, in other words\\ $\delta_{S,\Gamma}^2\in\Der(C^\infty(\Ab_{\Gamma}),C(\Ab_{\Gamma}))$. 
\end{lem}
Since multiplier algebra of a unital and commutative $C^*$-algebra is the algebra itself, the elements $E_S(\Gamma)$ are not contained in the multiplier algebra of $\Alg_\Gamma=C(\Ab_\Gamma)$. Hence the derivation defined in the lemma \ref{lem diffop} is not inner.  

Following the notion of infinitesimal representations $\dif U$ in a $C^*$-algebra introduced by Woronowicz \cite[p.8]{WoroNap} the flux operators can be also understood in the following way. 
Recall the unbounded operators $e^{\overleftarrow{R}},e^{\overleftarrow{L}}$ defined in  \eqref{CommRel1}, \eqref{CommRel4}. These operators are infinitesimal representations or differentials of the Lie flux group $\bar G_{\breve S,\Gamma}$ in $\KD(L^2(\Ab_\Gamma,\mu_\Gamma))$. They correspond to the set $\Rep(\bar G_{\breve S,\Gamma},\KD(L^2(\Ab_\Gamma,\mu_\Gamma)))$ of unitary representations $U$ of $\bar G_{\breve S}$, which are analysed in \cite{Kaminski1,KaminskiPHD}. Therefore rewrite
\beqs e^{\overrightarrow{L}}(f_\Gamma):= \dif U(E_S(\Gamma)) f_\Gamma\text{ for }f_\Gamma\in \DD(\dif U)\text{ and }E_S(\Gamma)\in\bar\E_{\breve S,\Gamma}
\eqs
The domain of the infinitesimal representations $\dif U$ is defined by
\beqs \DD(\dif U):=\{f_\Gamma\in C(\Ab_\Gamma): &\text{ the mapping }\bar G_{\breve S,\Gamma}\ni \rho_{S,\Gamma}(\Gamma)\mapsto \|U(\rho_{S,\Gamma}(\Gamma))f_\Gamma\|\\& \text{ is a }C^\infty( \bar G_{\breve S,\Gamma})- \text{ function }\}
\eqs which is a dense subset in $C(\Ab_\Gamma)$.

\subsection{Heisenberg holonomy-flux cross-product $^*$-algebras}\label{subsec Heisenbergalgebras}
The structures of Hopf algebras have been presented by Schm\"udgen and Klimyk \cite{KlimSchmued94}. The authors have rewritten the algera of Quantum mechanics in terms of the Hopf $^*$-algebra $(\Pol(\R^n),\bigtriangleup)$ of coordinate functions. Then $\Pol(\R^n)\rtimes \R^n$ is the Heisenberg algebra of Quantum Mechanics, where the elements are differential operators with polynomial coefficients. In this section similar algebras for Loop Quantum Gravity are studied.

First of all in mathematics a further cross-product, which is called the Heisenberg double, using the properties of bialgebras has been constructed. A particular bialgebra is a Hopf algebra. Let $G$ be either a connected compact Lie group and $\textbf{G}$ a simple matrix Lie group. Therefore consider either the Hopf $^*$-algebra $(C^\infty(G),\bigtriangleup)$, or the Hopf $^*$-algebra $(\Pol(\textbf{G}),\bigtriangleup)$ of coordinate functions on the group $\textbf{G}$, or the Hopf $^*$-algebra $(\Rep(G),\bigtriangleup)$ of representative functions on the group $G$. Then restrict the $^*$-algebra $\Pol(\textbf{G})$ or $\Rep(G)$ to a $^*$-subalgebra of $C^\infty(G)$, which is denoted by $\Pol^\infty(\textbf{G})$ or $\Rep^\infty(G)$. Suppose that $\la .,.\ra: \E\times \Pol^\infty(\textbf{G})\rightarrow \CB$ denotes the dual paring of $(\Pol^\infty(\textbf{G}),\bigtriangleup)$ or respectively $(\Rep^\infty(G),\bigtriangleup)$ and Hopf algebra $(\E,\hat\bigtriangleup)$, where $\E$ denote the universal enveloping flux algebra of $G$. The dual pairing is defined by
\beq \la E,f\ra:= \frac{\dif}{\dif t}\Big\vert_{t=0} f(e_G\exp(t E)) \text{ for }E\in\E\text{ and }f\in \Rep^\infty(G)
\eq
Then the Heisenberg double $\Rep^\infty(G)\rtimes_H \E$ is defined by the bilinear map
\beq E\rhd_H f:= \la E,\idf\ra f + \la E,f\ra =\la E,f\ra 
\eq and the multiplication
\beq (f_1,E_1)\cdot_H (f_2,E_2):= \la E_1,\idf\ra f_1f_2\otimes E_2 + \la \idf,f_2\ra f_1\otimes E_1E_2
\eq 
These objects are used to define similar objects for LQG.
\begin{defi} \label{defi Heisenberg}Let $G$ be either a connected compact Lie group or a simple matrix Lie group. Moreover let $\Ab_\Gamma$ for a graph $\Gamma$ be the set of generalised connections for $G$ such that $\Ab_\Gamma$ is identified with $G^N$ naturally, where $N=\vert E_\Gamma\vert$. Suppose that $\breve S$ has the simple intersection surface property for a finite orientation preserved graph system associated to $\Gamma$. Then $\bar\go_{\breve S,\Gamma}$ is identified with $\go^N$.

The \textbf{Heisenberg representation-holonomy-flux $^*$-algebra of the graph $\Gamma$ and the surface set $\breve S$} is given by \[\Rep^\infty(\Ab_\Gamma)\rtimes_H\bar\E_{\breve S,\Gamma}\]
The \textbf{Heisenberg polynomial-holonomy-flux $^*$-algebra of the graph $\Gamma$ and the surface set $\breve S$} is given by \[\Pol^\infty(\Ab_\Gamma)\rtimes_{H,\Pol}\bar\E_{\breve S,\Gamma}\]
The \textbf{Heisenberg holonomy-flux $^*$-algebra of the graph $\Gamma$ and the surface set $\breve S$} is given by \[C^\infty(\Ab_\Gamma)\rtimes_H\bar\E_{\breve S,\Gamma}\]
\end{defi}
Notice that, an element of $\Pol^\infty(\textbf{G}^N)$
is a matrix element $(\ho_\Gamma)_{ij}$ of a $M\times M$ matrix.
This elements are called coordinate functions $v^{i}_j(\ho_\Gamma)=(\ho_\Gamma)_{ij}$ on a simple matrix Lie group $\textbf{G}$.
Then an element of $\Pol^\infty(\textbf{G}^N)\rtimes\bar\E^N$ is for example given by
\beqs (\ho_\Gamma)_{ij}(\ho_\Gamma)_{kl} \otimes E_S(\Gamma)
\eqs where by natural identification $h_\Gamma:=\ho_\Gamma(\Gamma)$ is an element of $\textbf{G}^N$ and $E_S(\Gamma)$ is an element of the universal enveloping flux algebra  $\bar\E^N$ of the Lie group $\textbf{G}^N$. In this case the map
\beqs (\ho_\Gamma)_{ij}\rhd_{H,\Pol} E_S(\Gamma):=u^i_j(E_S(\Gamma)\rhd \ho_\Gamma)
\eqs is bilinear and defines the left $\bar\E^N$-module algebra $\Pol^\infty(\textbf{G}^N)$. 
Clearly these Heisenberg cross-product algebras defined above are not equivalent to a holonomy-flux cross-product $^*$-algebra and they are in particular Heisenberg doubles in the sense of Schm\"udgen and Klimyk. 

Similarly to the different automorphic actions on the $C^*$-algebra $C(\Ab_\Gamma)$ there are a lot of different Heisenberg doubles depending on the number of intersections and the orientations of the surface and paths.
 
\section{Representations and states of the holonomy-flux cross-product $^*$-algebra}\label{subsec Repholfluxcross}
\paragraph*{Surface-orientation-preserving graph-diffeomorphism-invariant states of the holonomy-flux cross-product $^*$-algebra}\hspace{10pt}

A $^*$-representations of the universal enveloping flux algebra $\bar\E_{\breve S,\Gamma}$ is given by the infinitesimal representation $\dif U$ of a unitary representation $U$ of $\bar G_{\breve S,\Gamma}$ in $C(\Ab_\Gamma)$. In general $^*$-representations can be defined on arbitrary $^*$-algebra, but there is no necessary condition that a unitary representation $U$ of the Lie group $\bar G_{\breve S,\Gamma}$ on a Hilbert space exists such that the commutator is equivalent to the infinitesimal representation. 
Mathematically $^*$-representations of Lie algebras are required to recover the structure of the Lie algebra. \begin{defi}
Let $\DD$ be a dense subspace of a Hilbert space $\HS$. A \textbf{$^*$-representation of a Lie algebra} $\go$ on $\DD$ is a mapping $\pi$ of $\go$ into $L(\DD)$ such that
\begin{enumerate}
 \item $\pi(\alpha X +\beta Y)=\alpha\pi(X)+\beta\pi(Y)$
 \item\label{repofLA} $\pi(\bra X,Y\ket)=\pi(X)\pi(Y)-\pi(Y)\pi(X)$
 \item\label{repofLA2} $ \la \pi(X)\phi,\varphi\ra=\la\phi,\pi(X^+)\varphi\ra$
\end{enumerate} whenever $X,Y\in\go$, $\alpha,\beta\in\R$ and $\phi,\varphi\in\DD$ 
where $L(\DD)$ vector space of linear mappings of $X$ into $X$.
\end{defi}
Notice that, from $\pi(X)\in L(\DD)$ and property \ref{repofLA2} it follows that $\pi(X)\in\Lop^+(\DD)$ (refer to Appendix). 

Let $\breve S$ be a surface set with same surface intersection property for  a finite orientation preserved graph system associated to a graph $\Gamma$.

In LQG the flux operators $E_S(\Gamma)$ are represented as differential operators $\dif U$ on the Hilbert space $\HS_\Gamma:=L^2(\Ab_\Gamma,\mu_\Gamma)$ of square integrable functions on $\Ab_\Gamma$. 
The domain $\DD(\dif U)$ of the infinitesimal representations $\dif U$ is the set of all functions $\psi_\Gamma$ in $L^2(\Ab_\Gamma,\mu_\Gamma)$ such that for each $E_S(\Gamma)\in\bar \go_{\breve S,\Gamma}$ the limit
\beqs \dif U(E_S(\Gamma))\psi_\Gamma :=  \lim_{t\rightarrow 0}\frac{(U(\exp(t E_S(\Gamma)))-\idf)\psi_\Gamma}{t}\eqs exists weakly in $L^2(\Ab_\Gamma,\mu_\Gamma)$. Notice that, there is a domain $\DD(\dif \UD)$ for all infinitesimals that corresponds to unitary representations $U\in\Rep(\bar G_{\breve S,\Gamma},\KD(L^2(\Ab_\Gamma,\mu_\Gamma)))$. The requirement of the existence of the limit is equivalent to the condition that, the function $\rho_S(\Gamma)\mapsto\la U(\rho_{S,\Gamma}(\Gamma)\psi_\Gamma,\phi_\Gamma\ra$ is in $C^\infty( \bar G_{\breve S,\Gamma})$ for each $\phi_\Gamma\in L^2(\Ab_\Gamma,\mu_\Gamma)$. With no doubt $C^\infty(\Ab_ \Gamma)\subset \DD(\dif U)$ holds.

Now recognize a short remark. The domain of the unbounded operator $\dif U$ can be rewritten in the following way.
Let $\{X_{S_1},...,X_{S_d}\}$ be a basis of $\bar\go_{\breve S,\Gamma}$ where $S_1,...,S_d\in\breve S$. Then by a corollary \cite[Cor.10.1.10]{Schmuedgen90} the domain $D(\dif U)$ is equivalent to the set of all elements $\psi_\Gamma\in L^2(\Ab_\Gamma,\mu_\Gamma)$ such that for all $X_{S_k}$ where $k=1,...,d$ and $\phi_\Gamma\in L^2(\Ab_\Gamma,\mu_\Gamma)$ the function $\R\ni t\mapsto \la U(\exp(tX_{S_k})\psi_\Gamma,\phi_\Gamma\ra$ is in $C^\infty(\R)$. Notice that, $t\mapsto U(\exp(tX_k))$ is a unitary representation of the Lie group $\R$ for each element $X_{S_k}$ of the basis of $\bar \go_{\breve S,\Gamma}$, too. Consequently it is assumed that, $U\in\Rep(\R,\KD(\HS_\Gamma))$ for each $X_{S_k}$. The operators $X_{S_k}$ corresponding infinitesimal representation $\dif U(X_{S_K})$ are called the infinitesimal generators of $U$. This reformulation can be used to understand the connection between the construction of the holonomy-flux $^*$-algebra of Lewandowski, Oko\l{}\'{o}w, Sahlmann and Thiemann \cite{LOST06} and the Weyl $C^*$-algebra of Fleischhack \cite{Fleischhack06}.

Summarising the operators $e^{\overleftarrow{L}}$ and $e^{\overleftarrow{R}}$ are defined on a dense linear subspace $\DD(\dif \UD)$ of the Hilbert space $\HS_\Gamma:=L^2(\Ab_\Gamma,\mu_\Gamma)$, and their adjoint operators  $e^{\overrightarrow{R}},e^{\overrightarrow{L}}$ defined on $\DD(\dif \UD^*)$. In particular, $e^{\overleftarrow{L}}$ and $e^{\overleftarrow{R}}$ are elements of the set 
\beqs\Lop^+_\UD(\DD(\dif \UD)):=\{\dif U\in\Lop(\DD(\dif \UD)):\text{ } &U\in\Rep( \bar G_{\breve S,\Gamma},\KD(\HS_\Gamma)),\\& \DD(\dif \UD)\subset \DD(\dif \UD^*), \dif U^*\DD(\dif \UD)\subset \DD(\dif \UD)\}
\eqs where $\Lop^+_\UD(\DD(\dif \UD))\subset \Lop(\DD(\dif \UD))$ and $\Lop(\DD(\dif \UD))$ denotes the set of all linear operators from $\DD(\dif \UD)$ to $\DD(\dif \UD)$ and $\DD(\dif \UD^*)$ the domain of the adjoint of the linear operator $\dif U$. 

Now it is obvious that, $\dif U$ is a $^*$-representation of $\bar\go_{\breve S,\Gamma}$ on $\DD(\dif \UD)$.

In analogy to the result of Schm\"udgen in \cite[Prop 10.1.6]{Schmuedgen90} the following proposition holds.
\begin{prop}Let $\breve S$ be a surface set with same surface intersection property for a finite orientation preserved graph system associated to a graph $\Gamma$.

Let $\bar\E_{\breve S,\Gamma}$ be the universal enveloping Lie flux $^*$-algebra for a surface set $\breve S$ and $U$ a unitary representation of $\bar G_{\breve S,\Gamma}$ on the Hilbert space $L^2(\Ab_\Gamma,\mu_\Gamma)$. 

Then $(\dif U(E_S(\Gamma)))(f_\Gamma):=e^{\overrightarrow{L}}(f_\Gamma)$ defines a $^*$-repesentation $\dif U$ of $\bar\E_{\breve S,\Gamma}$ on a dense subdomain $\DD(\dif U)$ of the Hilbert space $L^2(\Ab_\Gamma,\mu_\Gamma)$. 
\end{prop}
\begin{proofs} 
Let $\tilde\tau_1$ be a homomorphism of the Lie algebra $\bar\go_{\breve S,\Gamma}$ in $C^\infty(\Ab_\Gamma)$. Furthermore $\dif \tilde U$ is a $^*$-homomorphism of $\bar\go_{\breve S,\Gamma}$ into the $O^*$-algebra $\Lop^+(\DD(\dif U))$ such that $\dif \tilde U(\idf)=1$. Then $\dif \tilde U$ defines a $^*$-representation of $\bar\go_{\breve S,\Gamma}$ on the domain $\DD(\dif U)$. There exists a unique extension of $\dif \tilde U$ to an homomorphism $\dif U$ of the $^*$-algebra $\bar\E_{\breve S,\Gamma}$ into the $^*$-algebra $\Lop^+(\DD(\dif U))$, which defines a $^*$-representation of $\bar\E_{\breve S,\Gamma}$ by corollary \ref{cor uniquhom}.

Consequently one shows that, the map $\bar\go_{\breve S,\Gamma}\ni X_S\mapsto \dif U(X_S)$ is a $^*$-representation of $\bar\go_{\breve S,\Gamma}$  on $\DD(\dif U)$. For a suitable surface $S$ and a graph $\Gamma$ set $E_S(\Gamma)=X_S$. First derive that,
\beqs &\la\dif U(X_S)\varphi_\Gamma,\phi_\Gamma\ra
=\frac{\dif}{\dif t}\Big\vert_{t=0}\la U(\exp(tX_S)) \psi_\Gamma,\varphi_\Gamma\ra=\frac{\dif}{\dif t}\Big\vert_{t=0}\la \psi_\Gamma, U^*(\exp(tX_S)) \varphi_\Gamma\ra\\
&=\frac{\dif}{\dif t}\Big\vert_{t=0}\la \psi_\Gamma, U(-\exp(tX_S)) \varphi_\Gamma\ra=-\frac{\dif}{\dif t}\Big\vert_{t=0}\la \psi_\Gamma, U(\exp(tX_S)) \varphi_\Gamma\ra = - \la\varphi_\Gamma,\dif U(X_S)\phi_\Gamma\ra\\
\eqs yields for $\psi_\Gamma,\varphi_\Gamma\in\HS_\Gamma$. Remember that, $X^+_S=-X_S$ for $E_S(\Gamma)\in \go_{\breve S,\Gamma}$ to conclude $\dif U(X_S)^*=-\dif U(X_S)=\dif U(X_S^+)$.

Hence the crucial property is \ref{repofLA}. Let $X\mapsto U(\exp(tX))$ be weakly continuous, then derive
\beqs &\la (\dif U(X)\dif U(Y)- \dif U(Y)\dif U(X))\varphi_\Gamma,\phi_\Gamma\ra\\
&=\frac{\dif}{\dif t}\Big\vert_{t=0} \Big( \frac{\dif}{\dif s}\Big\vert_{s=0} \big\la U(\exp(-tX_S)\exp(-sY_S)) \psi_\Gamma,\varphi_\Gamma\big\ra\Big)\\
&\quad -\frac{\dif}{\dif s}\Big\vert_{s=0} \Big( \frac{\dif}{\dif t}\Big\vert_{t=0} \big\la U(\exp(-sY_S)\exp(-tX_S)) \psi_\Gamma,\varphi_\Gamma\big\ra\Big)\\[3pt]
&=  \big\la \left(Y_S X_S  - X_SY_S \right)\psi_\Gamma,\varphi_\Gamma\big\ra=\big\la\bra Y_S,X_S\ket  \psi_\Gamma,\varphi_\Gamma\big\ra\\
&=\frac{\dif}{\dif t}\Big\vert_{t=0}\big\la U(t\bra Y_S,X_S\ket)  \psi_\Gamma,\varphi_\Gamma\big\ra\\
&=\frac{\dif}{\dif t}\Big\vert_{t=0}\big\la U(-t\bra X_S,Y_S\ket)  \psi_\Gamma,\varphi_\Gamma\big\ra
\eqs
\end{proofs}

Remark that, the unbounded operator $\dif U$ of $\bar \E_{\breve S,\Gamma}$ and the operator $\dif U(E_S(\Gamma))$ for a fixed element $E_S(\Gamma)\in\bar \E_{\breve S,\Gamma}$ are not equivalent, since for example the domains are different. Observe that, for an infinitesimal generator $\dif U(E_S(\Gamma))$ of the strongly continuous one-parameter unitary group $\R\ni t\mapsto U(\exp(tE_S(\Gamma)))$ on $L^2(\Ab_\Gamma,\mu_\Gamma)$ define the self-adjoint $i\dif U(E_S(\Gamma))$ on the domain $D(\dif U(E_S(\Gamma)))$. Clearly the subset $D(\dif U)$ is contained in $D(\dif U(E_S(\Gamma)))$. Therefore different special flux operators $E_S(\Gamma)$ or all flux operators $E_S(\Gamma)$ can be analysed. 

Moreover the operators the $\dif U$ and $\dif U(E_S(\Gamma))$ have different self-adjointness properties. Indeed the $^*$-representation $\dif U$ on  $D(\dif U)$ is self-adjoint \cite[Cor.10.2.3]{Schmuedgen90}, whereas $\dif U(E_S(\Gamma))$ for any hermitian elliptic element $E_S(\Gamma)$ of $\bar\E_{\breve S,\Gamma}$ on the domain $D(\dif U)$ is essentially self-adjoint \cite[Cor. 10.2.5]{Schmuedgen90}. 

Finally for general elliptic elements in $\bar\E_{\breve S,\Gamma}$ the adjoint operator $\dif U(E_S(\Gamma))^*$ is equivalent to the closure w.r.t. the graph topology of $\dif U(E_S^+(\Gamma))$, \cite[Cor.10.2.7]{Schmuedgen90}. For an abelian or compact Lie group $G$ it turns out that, the adjoint $\dif U(E_S(\Gamma))^*$ is equivalent to the closure w.r.t. the graph topology of $\dif U(E_S^+(\Gamma))$ for all $E_S(\Gamma)\in\bar\E_{\breve S,\Gamma}$. 

Summarising the issue of domains of the different differential operators have to be carefully analysed. 

According to the observations of the Lie flux group $C^*$-algebra, the universal enveloping flux $^*$-algebra $\bar\E_{\breve S,\Gamma}$ can be considered. This algebra itself can be shown to be equivalent to the algebra of differential operators on $C^\infty(\bar G_{\breve S,\Gamma})$. Observe that, due to the different structure of $\bar G_{\breve S,\Gamma}$ and $\Ab_\Gamma$ the identification of both sets is valid only for suitable surface sets and graphs.

\begin{prop}Let $G$ be a compact Lie group and the set $\breve S$ has the same intersection surface property for a finite orientation preserved graph system associated to a graph $\Gamma$. Set $N=\vert E_\Gamma\vert$ and identify $\Ab_\Gamma$ with $G^N$ naturally.

Then the universal enveloping Lie flux $^*$-algebra $\bar\E_{\breve S,\Gamma}$ is $^*$-isomorphic to the $O^*$-algebra $\DD_{\breve S}(\bar G_{\breve S,\Gamma})$ of differential operators on $C^\infty(G^N)$ in the Hilbert space $L^2(G^N,\mu_N)$, where $\DD_{\breve S}(\bar G_{\breve S,\Gamma})$ is the algebra of all right-invariant differential operators \(\dif U_{\overleftarrow{L}}(\bar\E_{\breve S,\Gamma})\big\vert_{C^\infty(G^N)}\) on $G^N$.
\end{prop}

Summarising there are different involutive algebras, like the analytic holonomy algebra associated a graph, the universal enveloping Lie flux $^*$-algebra associated a graph and a surface set or the holonomy-flux cross-product $^*$-algebra associated a graph represented on the Hilbert space $\HS_\Gamma$.

\begin{theo}\label{theo repholfluxcrossstar}
Let $\breve S$ be a surface set having the same intersection surface property for a finite orientation preserved graph system associated to $\Gamma$.

There exists the following $^*$-representations of the analytic holonomy $C^*$-algebra $C^\infty(\Ab_\Gamma)$, the universal enveloping Lie flux $^*$-algebra $\bar\E_{\breve S,\Gamma}$ and the holonomy-flux cross-product $^*$-algebra $C^\infty(\Ab_\Gamma)\rtimes_{L}\bar\E_{\breve S,\Gamma}$ for a graph $\Gamma$ and a surface set $\breve S$ on the Hilbert space $\HS_\Gamma=L^2(\Ab_\Gamma,\mu_\Gamma)$ on $C^\infty(\Ab_\Gamma)$:
\beqs 
&\Phi_M(f_\Gamma) \psi_\Gamma = f_\Gamma\psi_\Gamma\text{ for } f_\Gamma\in C^\infty(\Ab_\Gamma) \\
&\Phi_M(f^*_\Gamma) \psi_\Gamma = \overline{f_\Gamma}\psi_\Gamma\text{ for } f_\Gamma\in C^\infty(\Ab_\Gamma) \\
&\dif U (E_S(\Gamma))\psi_\Gamma =\bra E_S(\Gamma), \psi_\Gamma \ket\text{ for } E_S(\Gamma)\in \bar\E_{\breve S,\Gamma} \\
&\dif U(E_S(\Gamma)^+)\psi_\Gamma =\bra E_S(\Gamma)^+, \psi_\Gamma \ket\text{ for } E_S(\Gamma)\in \bar\E_{\breve S,\Gamma} \\
&\pi(f_\Gamma\otimes E_S(\Gamma)) \psi_\Gamma = 
\frac{1}{2} \bra E_S(\Gamma), f_\Gamma \ket \psi_\Gamma +\frac{1}{2} f_\Gamma\bra E_S(\Gamma),\psi_\Gamma\ket \text{ for } f_\Gamma\otimes E_S(\Gamma)\in C^\infty(\Ab_\Gamma)\rtimes\bar\E_{\breve S,\Gamma}\\
&\pi((f_\Gamma\otimes E_S(\Gamma))^*) \psi_\Gamma = 
\frac{1}{2}\bra E_S(\Gamma)^+, f^*_\Gamma \ket \psi_\Gamma + \frac{1}{2} f_\Gamma\bra E_S(\Gamma)^+,\psi_\Gamma\ket
\text{ for } f_\Gamma\otimes E_S(\Gamma)\in C^\infty(\Ab_\Gamma)\rtimes\bar\E_{\breve S,\Gamma}
\eqs whenever $\psi_\Gamma\in C^\infty(\Ab_\Gamma)$.
For two surfaces $S_1\cap S_2=\varnothing$ the representation satisfies
\beqs &\pi(\bra f_\Gamma^1\otimes E_{S_1}(\Gamma), f_\Gamma^2 \otimes E_{S_2}(\Gamma)\ket)\psi_\Gamma 
\\
&= \frac{1}{4} \bra E_{S_2}(\Gamma), f_\Gamma^1\bra E_{S_1}(\Gamma), f^2_\Gamma\ket\ket \psi_\Gamma 
+\frac{1}{4} f_\Gamma^1\bra E_{S_1}(\Gamma), f^2_\Gamma\ket \bra E_{S_2}(\Gamma),\psi_\Gamma\ket\\&\quad
-\frac{1}{4} \bra E_{S_1}(\Gamma), f_\Gamma^2\bra E_{S_2}(\Gamma), f^1_\Gamma\ket \ket \psi_\Gamma 
-\frac{1}{4} f_\Gamma^2\bra E_{S_2}(\Gamma), f^1_\Gamma\ket\bra E_{S_1}(\Gamma),\psi_\Gamma\ket
\eqs whenever $\psi_\Gamma\in C^\infty(\Ab_\Gamma)$.

The representation $\pi$ of the holonomy-flux cross-product $^*$-algebra is called the \textbf{Heisenberg representation of $C^\infty(\Ab_\Gamma)\rtimes_{L}\bar\E_{\breve S,\Gamma}$ on $\HS_\Gamma$}.
\end{theo}
\begin{proofs}
The following computations show that, $\pi$ is a $^*$-representation of $C^\infty(\Ab_\Gamma)\rtimes_L\bar\E_{\breve S,\Gamma}$ on the domain $C^\infty(\Ab_\Gamma)$:
\beqs 
&\pi(\lambda_1)f_\Gamma\otimes \lambda_2 E_S(\Gamma)) \psi_\Gamma
= \frac{1}{2}\lambda_1\lambda_2\left( \bra E_S(\Gamma), f_\Gamma \ket \psi_\Gamma 
-\frac{1}{2} f_\Gamma \bra E_S(\Gamma),\psi_\Gamma\ket \right)\\
&\pi(f^1_\Gamma\otimes E_{S_1}(\Gamma)+f^2_\Gamma\otimes E_{S_2}(\Gamma)) \psi_\Gamma = 
\pi(f^1_\Gamma\otimes E_{S_1}(\Gamma))+\pi(f^2_\Gamma\otimes E_{S_2}(\Gamma)) \psi_\Gamma\\
&\pi((f_\Gamma\otimes E_S(\Gamma))^*) \psi_\Gamma = 
\frac{1}{2} \bra E_S(\Gamma)^+, f^*_\Gamma \ket \psi_\Gamma +\frac{1}{2} f_\Gamma\bra E_S(\Gamma)^+,\psi_\Gamma\ket
= \pi(f_\Gamma\otimes E_S(\Gamma))^* \psi_\Gamma
\eqs for $f_\Gamma,f_\Gamma^1,f_\Gamma^2\in C^\infty(\Ab_\Gamma)$, $E_S(\Gamma),E_{S_1}(\Gamma),E_{S_2}(\Gamma)\in\bar\E_{\breve S,\Gamma}$, $\lambda_1,\lambda_2\in\CB$.

For two surfaces $S_1\cap S_2=\varnothing$ calculate
\beqs &\pi(\bra f_\Gamma^1\otimes E_{S_1}(\Gamma), f_\Gamma^2 \otimes E_{S_2}(\Gamma)\ket)\psi_\Gamma 
\\&=\pi(
f_\Gamma^1\bra E_{S_1}(\Gamma), f^2_\Gamma\ket\otimes E_{S_2}(\Gamma))   
- \pi(f_\Gamma^2\bra E_{S_2}(\Gamma), f^1_\Gamma\ket\otimes E_{S_1}(\Gamma))
+ \pi(f_\Gamma^1f^2_\Gamma \otimes \bra E_{S_1}(\Gamma),E_{S_2}(\Gamma)\ket)\\
&= \frac{1}{4} \bra E_{S_2}(\Gamma), f_\Gamma^1\bra E_{S_1}(\Gamma), f^2_\Gamma\ket\ket \psi_\Gamma 
+\frac{1}{4} f_\Gamma^1\bra E_{S_1}(\Gamma), f^2_\Gamma\ket \bra E_{S_2}(\Gamma),\psi_\Gamma\ket\\&\quad
-\frac{1}{4} \bra E_{S_1}(\Gamma), f_\Gamma^2\bra E_{S_2}(\Gamma), f^1_\Gamma\ket \ket \psi_\Gamma 
-\frac{1}{4} f_\Gamma^2\bra E_{S_2}(\Gamma), f^1_\Gamma\ket\bra E_{S_1}(\Gamma),\psi_\Gamma\ket
\eqs
\end{proofs}

From another point of view the bracket $\bra E_S(\Gamma),.\ket $ defines a $^*$-derivation of the analytic holonomy $C^*$-algebra $C(\Ab_\Gamma)$ for a graph $\Gamma$. Moreover in general such $^*$-derivations can be implemented by automorphisms on $C(\Ab_\Gamma)$. This point of view is more general than the consideration of differential operators. 

For a simplification restrict the following computations to a suitable surface $S$ and a graph $\Gamma:=\{\gamma\}$.
\begin{lem}\label{lem_01}Let $\Phi_M$ be a representation of $C(\Ab_\Gamma)$ on $\HS_\Gamma$ and $\alpha\in\Act(\bar G_{\breve S,\Gamma},C(\Ab_\Gamma))$ defined in \cite[Section 3.1]{Kaminski2} or \cite[Section 6.1]{KaminskiPHD}, where $\rho_{S,\Gamma}(\Gamma)=\exp(tE_S(\Gamma))\in\bar G_{\breve S,\Gamma}$ for a $t\in\R$. Let $\Gamma=\{\gamma\}$ and $S$ be suitable and set $E_S(\Gamma):=X_S$.

Then it is true that,
\beqs\omega_M^\Gamma(\alpha_{\exp(X_S)}^t(f_\Gamma))
&= \int_{\Ab_\Gamma} f_{\Gamma}(\exp(t X_S)\ho_{\Gamma}(\gamma))\dif\mu_{\Gamma}(\ho_{\Gamma}(\gamma))\\
&= \int_{\Ab_\Gamma} f_{\Gamma}(\ho_{\Gamma}(\gamma))\dif\mu_{\Gamma}(\ho_{\Gamma}(\gamma))\\
&=\omega_M^\Gamma(f_\Gamma)
\eqs yields for all $t\in\R$ and $f_\Gamma\in C(\Ab_\Gamma)$.
\end{lem}

There is a general property of states on an arbitrary (untial) $C^*$-algebra $\Alg$ represented on a Hilbert space and a group of $^*$-automorphisms $\alpha$.

\begin{cor}\label{cor derivationgraph} Let  
\[\delta_{S,\Gamma}(f_\Gamma):=i\bra E_{S}(\Gamma_i)^+ E_S(\Gamma),f_\Gamma\ket\text{ for }E_S\in\E_{\breve S,\Gamma},f_\Gamma\in C^\infty(\Ab_\Gamma)\]
be a $^*$-derivation such that $\delta_{S,\Gamma}\in\Der(C^\infty(\Ab_\Gamma),C(\Ab_\Gamma))$. Let $\alpha\in\Act(\bar G_{\breve S,\Gamma},C(\Ab_\Gamma))$.

Then for each element $E_S(\Gamma)\in\bar \E_{\breve S,\Gamma}$ the limit 
\beq\label{norm limit} \tilde\delta_{S,\Gamma}(f_\Gamma):=\lim_{t\rightarrow 0}\frac{\alpha^t_{i\exp(E_{S}(\Gamma_i)^+E_S(\Gamma))}(f_\Gamma)-f_\Gamma}{t}\text{ for }f_\Gamma\in C^\infty(\Ab_\Gamma)
\eq exists in norm topology and $\tilde\delta_{S,\Gamma}=\delta_{S,\Gamma}$.

Then the state $ \omega_M^\Gamma$ on $C(\Ab_\Gamma)$ presented in \cite[Corollary 3.28]{Kaminski2} or \cite[Corollary 6.1.41]{KaminskiPHD} satisfies 
\beq \omega_M^\Gamma(\delta_{S,\Gamma}(f_\Gamma))=0\eq for all $f_\Gamma\in C^\infty(\Ab_\Gamma)$. Hence there is a cyclic vector $\Omega_M^\Gamma$ of the GNS-representation associated to this state such that
\beqs E_{S}(\Gamma_i)^+E_S(\Gamma)\Omega_M^\Gamma=0\text{ for all }E_S\in\E_{\breve S,\Gamma}
\eqs
where $E_S(\Gamma)$ is a skew-adjoint operator with domain $\DD(E_S(\Gamma))$.
\end{cor}
Notice that, the flux operator $iE_S^+(\Gamma)E_S(\Gamma)$ is self-adjoint and positive in $\HS_\Gamma$. The unbounded $^*$-derivation given by $\delta_{S,\Gamma}$ is symmetric.

\begin{proofs}First observe that, $\omega_M^\Gamma(\alpha_{i\exp(E_{S}(\Gamma_i)^+E_{S}(\Gamma))}^t(f_\Gamma))=\omega_M^\Gamma(f_\Gamma)$ holds for all $t\in\R$ and $f_\Gamma\in C(\Ab_\Gamma)$. 
This follows from lemma \ref{lem_01}. For $f_\Gamma\in C(\Ab_\Gamma)$ the $^*$-automorphisms $\alpha$ are implementable as a one-parameter group $\R\ni t\mapsto \alpha^t_{iE_{S}(\Gamma_i)^+E_S(\Gamma)}(f_\Gamma)\in C(\Ab_\Gamma)$ for each $E_S(\Gamma)\in\bar\E_{\breve S,\Gamma}$, which is weakly continuous. 

Then the norm limit \ref{norm limit} exists for suitable $f_\Gamma$ in the domain $C^\infty (\Ab_\Gamma)$. The symmetric derivation is therefore given by
\beqs &\delta_{S,\Gamma}(f_\Gamma)
=\frac{\dif}{\dif t}\Big\vert_{t=0} \alpha^t{iE_{S}(\Gamma_i)^+E_S(\Gamma)}(f_\Gamma)\text{ for }f_\Gamma\in C ^\infty(\Ab_\Gamma)
\eqs The operator $iE_{S}(\Gamma_i)^+E_S(\Gamma)$ is the generator of the unbounded symmetric $^*$-derivation $\delta_{S,\Gamma}$ on $\HS_\Gamma$ by definition. 

Recall the state $\omega_M^\Gamma$ of $C(\Ab_\Gamma)$ presented in \cite[Proposition 3.29]{Kaminski2} or \cite[Proposition 6.1.40]{KaminskiPHD}, then the derivation $\delta_{S,\Gamma}$ satisfies
\beq\label{eq derivcondstate}\omega_M^{\Gamma}(\delta_{S,\Gamma}(f_{\Gamma}))
&=\Big\la\Omega_M^\Gamma, \frac{\dif}{\dif t}\Big\vert_{t=0}\Phi_M(\alpha_{i\exp(E_{S}(\Gamma_i)^+E_{S}(\Gamma))}^t(f_\Gamma)) \Omega_M^\Gamma\Big\ra\\
&=\frac{\dif}{\dif t}\Big\vert_{t=0}\omega_M^{\Gamma}(\alpha_{i\exp(E_{S}(\Gamma_i)^+E_{S}(\Gamma))}^t(f_\Gamma))=0\eq for $f_\Gamma\in C^\infty(\Ab_\Gamma)$. There exists a covariant representation $(\Phi_M,U)$ of $(\bar G_{\breve S,\Gamma},\Alg_\Gamma,\alpha)$ in $\LD(\HS_\Gamma)$ such that
\beqs \Phi_M(\alpha^t_{i\exp(E_{S}(\Gamma_i)^+E_{S}(\Gamma))}(f_\Gamma))=U(\exp(tE_{S}(\Gamma_i)^+E_S(\Gamma)))\Phi_M(f_\Gamma)U(\exp(-tE_{S}(\Gamma_i)^+E_S(\Gamma)))
\eqs and, hence, $U(\exp(tE_{S}(\Gamma_i)^+E_S(\Gamma)))\Omega_M^\Gamma=\Omega_M^\Gamma$ for all $t\in\R$ and $iE_S^+(\Gamma)E_S(\Gamma)\Omega_M^\Gamma=0$.
\end{proofs}

The next derivation can be defined only for a suitable family of graphs and a suitable surface set.

Define the $^*$-derivation on the domain $C ^\infty(\Ab)$ of the $C^*$-algebra $C (\Ab)$ by
\beqs \delta_{S}(f):=i\bra E_{S}(\Gamma_\infty)^+E_{S}(\Gamma_\infty),f\ket\text{ for }f\in C^\infty (\Ab),E_S\in\E_{\breve S}\text{ and }\Gamma_\infty\in\PD_{\Gamma_\infty}
\eqs 

\begin{prop}\label{prop derivationinf}
Let $\Gamma_\infty$ be the inductive limit of a family of graphs $\{\Gamma_i\}$. Let $\breve S$ be a finite set of surfaces in $\Sigma$ such that 
\begin{enumerate}
 \item such that the surface set $\breve S$ has the same surface intersection property for each graph of the family,
\item the inductive limit structure preserves the same surface intersection property for $\breve S$ and
 \item each surface in $\breve S$ intersects the inductive limit graph $\Gamma_\infty$ only in a finite number of vertices.
\end{enumerate}
Then $\PD_{\Gamma_\infty}^{\op}$ is the inductive limit of an inductive family $\{\PD_{\Gamma_i}^{\op}\}$ of finite orientation preserved graph systems. Let $\Ab_{\Gamma_i}$ is identified in the natural way with $G^{N_i}$.

Then the limit 
\beqs \alpha^t_{\exp(iE_S^+(\Gamma_\infty)E_S(\Gamma_\infty))}(f)
:=\lim_{j\rightarrow \infty}\alpha^t_{\exp(iE_S^+(\Gamma_j)E_S(\Gamma_j))}(f)\text{ for }f\in C^\infty (\Ab),E_S\in\E_{\breve S}\text{ and }\Gamma_j\in\PD_{\Gamma_\infty}
\eqs exists for each $t\in\R$ in norm topology. Consequently the limit 
\beqs \tilde\delta_{S}(f):=\lim_{t\rightarrow 0}\frac{\alpha^t_{\exp(iE_S^+(\Gamma_\infty)E_S(\Gamma_\infty))}(f)-f}{t}\text{ for }f\in C^\infty (\Ab)
\eqs exists in norm topology and $ \tilde\delta_{S}= \delta_{S}$ for $S\in\breve S$.

Finally, for each $^*$-derivation $\delta_S$ the state satisfies
\beqs \omega_M(\delta_{S}(f))=0\eqs for all $f\in C^\infty (\Ab)$ and for $S\in\breve S$.
\end{prop}
\begin{proofs}On the inductive limit of the family of $C^*$-algebras $\{(C(\Ab_{\Gamma_i}),\beta_{\Gamma_i,\Gamma_j}):\PD_{\Gamma_i}^{\op}\leq \PD_{\Gamma_j}^{\op}\}$ the action of $\alpha(\rho_{S,\Gamma_i}(\Gamma_i))$ on $C(\Ab_{\Gamma_i})$ for $\rho_{S,\Gamma_i}:=\exp(iE_{S}(\Gamma_i)^+E_{S}(\Gamma_i))\in \bar G_{\breve S,\Gamma_i}$ is non-trivial if each surface $S$ of $\breve S$ intersect the graph $\Gamma_i$ in vertices of the graph $\Gamma_i$. In particular, there is a graph $\Gamma_{j}$ having the maximal number of intersection vertices with any surface $S$ in $\breve S$. Consequently for a graph $\Gamma_{j+1}$ that contains $\Gamma_j$ the flux $E_{S}(\Gamma_{j+1}\setminus\Gamma_j)=0$. Furthermore derive
\beqs \alpha^t_{\exp(iE_S^+(\Gamma_\infty)E_S(\Gamma_\infty))}(f)
&:=\lim_{k\rightarrow \infty}\alpha^t_{\exp(iE_S^+(\Gamma_k)E_S(\Gamma_k))}(f)\\
&=\lim_{k\rightarrow \infty}\alpha^t_{\exp(iE_S^+(\Gamma_k)E_S(\Gamma_k))}(\beta_{\Gp}f_{\Gp})\\
&=\alpha^t_{\exp(iE_S^+(\Gamma_j)E_S(\Gamma_j))}(\beta_{\Gp}f_{\Gp}) \text{ for }f\in C^\infty(\Ab), f=\beta_{\Gp} f_{\Gp}\text{ and }E_S\in\E_{\breve S}
\eqs whenever $\Gamma_k\leq\Gp$ for $1\leq k \leq j$. Furthermore conclude that,
\beqs \delta_{S}(f)
&= i\bra E_{S}(\Gamma_j)^+E_{S}(\Gamma_j),f\ket + \lim_{i\rightarrow\infty}i\bra E^+_{S}(\Gamma_{j+i}\setminus\Gamma_j)E_{S}(\Gamma_{j+i}\setminus\Gamma_j),f\ket\\
&= \beta_{\Gamma_j}\circ(i\bra E_{S}(\Gamma_j)^+E_{S}(\Gamma_j),f_{\Gamma_j}\ket)
\text{ for }f\in C^\infty(\Ab), f=\beta_{\Gamma_j}f_{\Gamma_j}\text{ and }E_S\in\E_{\breve S}
\eqs 
holds.
Hence there is a $^*$-homomorphism $\beta_{\Gamma,\Gamma^\prime}$ from $C^\infty(\Ab_\Gamma)$ to $C^\infty(\Ab_{\Gamma^\prime})$ such that 
$\beta_{\Gamma,\Gamma^\prime}\circ\delta_{S,\Gamma}\circ\beta_{\Gamma,\Gamma^\prime}^{-1}=\delta_{S,\Gamma^\prime}$ is a $^*$-derivation from $C^\infty(\Ab_{\Gamma^\prime})$ into $C(\Ab_{\Gamma^\prime})$ and 
\beq\omega_M^{\Gamma}(\delta_{S,\Gamma}(f_{\Gamma}))=\omega_M^{\Gamma^\prime}(\beta_{\Gamma,\Gamma^\prime}\circ\delta_{S,\Gamma})(f_{\Gamma}))=\omega_M^\Gp(\delta_{S,\Gp}(f_\Gp))=0\eq

Finally, derive
\beq \omega_M(\delta_{S}(f))=\omega_M(\beta_{\Gamma_j}\circ\delta_{S,\Gamma_j})(f_{\Gamma_j}))
=\beta_{\Gamma_j}^*\omega_M^{\Gamma_j}(\delta_{S,\Gamma_j}(f_{\Gamma_j}))
=0
\eq whenever $f\in C^\infty(\Ab)$.
\end{proofs}  

Recall that, $\omega_M$ is graph-diffeomorphism invariant if the natural identification of $\Ab_\Gamma$ with $G^{\vert\Gamma\vert}$ is used. Recall that a real-valued, linear and $^*$-preserving functional $\omega$ on a $^*$-algebra associated to a $^*$-representation $\pi$ has to be positive.

\begin{theo}\label{theo stateholfluxalg}
Let $\Gamma_\infty$ be the inductive limit of a family of graphs $\{\Gamma_i\}$. Let $\breve S$ be a finite set of surfaces in $\Sigma$ such that 
\begin{enumerate}
 \item the surface set $\breve S$ has the same surface intersection property for each graph of the family,
\item the inductive limit structure preserves the same surface intersection property for $\breve S$ and
 \item each surface in $\breve S$ intersects the inductive limit graph $\Gamma_\infty$ only in a finite number of vertices.
\end{enumerate}
Then $\PD_{\Gamma_\infty}^{\op}$ is the inductive limit of a inductive family $\{\PD_{\Gamma_i}^{\op}\}$ of finite orientation preserved graph systems. Let $\Ab_\Gamma$ is identified in the natural way with $G^N$. Denote the center of $\bar\E_{\breve S}$ by $\ZD(\bar\E_{\breve S})$.

The state $\bar\omega_M$ associated to the GNS-representation $(\HS_\Gamma,\pi,\Omega)$ given in theorem \ref{theo repholfluxcrossstar} is a surface-orientation-preserving graph-diffeomorphism invariant state for a fixed set of surfaces $\breve S$ on the holonomy-flux cross-product $^*$-algebra $C^\infty(\Ab)\rtimes_L\ZD(\bar\E_{\breve S})$ such that 
\beqs \bar\omega_M(f\otimes E_S(\Gamma_\infty))&=\beta_{\Gamma_j}^*\bar\omega_M^{\Gamma_j}(f_{\Gamma_j}\otimes E_S(\Gamma_j))=0\text{ for all }E_S\in\ZD(\E_{\breve S})
\eqs 
Moreover the state $\bar\omega_M$ on $C^\infty(\Ab)\rtimes_L\ZD(\bar\E_{\breve S})$ is the unique state, which is surface-orientation-preserving graph-diffeomorphism invariant.
\end{theo}
\begin{proofs}First recall that,
\beqs \zeta_{\sigma}\circ\alpha (\rho_{S,\Gamma}(\Gamma))\neq \alpha (\rho_{S,\Gamma}(\Gamma_\sigma))\circ\zeta_{\sigma}
\eqs yields for every $\sigma\in\mathfrak{B}(\PD_{\Gamma}^{\op})$ and  $\rho_{S,\Gamma}:=\exp(iE_{S}(\Gamma)^+E_{S}(\Gamma))\in G_{\breve S,\Gamma}$. Hence the problem carry over to 
\beqs \delta_{S,\Gamma}\circ\zeta_\sigma\neq\zeta_\sigma\circ\delta_{S,\Gamma}
\eqs Consequently in the following the center $\ZD(\bar\E_{\breve S,\Gamma})$ and the center $\ZD(\bar\E_{\breve S})$ are considered.

Moreover a surface-orientation-preserving graph-diffeomorphism invariant state for a fixed set of surfaces $\breve S$ means that
\beqs \bar\omega_M(\zeta_{\sigma}(f\otimes E_S(\Gamma_\infty)))=\bar\omega_M(f\otimes E_S(\Gamma_\infty))\text{ for all }\sigma\in\mathfrak{B}_{\breve S,\ori}(\PD_{\Gamma_\infty}^{\op})\text{ and }E_S(\Gamma_\infty)\in\ZD(\bar\E_{\breve S})
\eqs holds. To show that $\bar\omega_M$ satisfies this property derive the following 
\beqs \bar\omega_M(f\otimes E_S(\Gamma_\infty))&=\beta_{\Gamma_j}^*\omega_M^{\Gamma_j}(f_{\Gamma_j}\otimes E_S(\Gamma_j))
= \la \Omega_M^{\Gamma_j}, \pi(f_{\Gamma_j}\otimes E_S(\Gamma_j))\Omega_M^{\Gamma_j}\ra\\
&=\frac{1}{2} \la \Omega_M^{\Gamma_j},\bra E_S(\Gamma_j),f_{\Gamma_j} \ket \Omega_M^{\Gamma_j}\ra 
+\frac{1}{2} \la \Omega_M^{\Gamma_j}, f_{\Gamma_j} \dif U(E_S(\Gamma_j))\Omega_M^{\Gamma_j} \ra\\
&=\frac{1}{2} \la \Omega_M^{\Gamma_j},\delta_{S,\Gamma_j}(f_{\Gamma_j})\Omega_M^{\Gamma_j}\ra 
+\frac{1}{2} \la \Omega_M^{\Gamma_j}, f_{\Gamma_j} \frac{\dif}{\dif t}\Big\vert_{t=0} U(\exp(t E_S(\Gamma_j)))\Omega_M^{\Gamma_j} \ra
=0
\eqs for all $f\in C(\Ab)$ and $E_S\in\E_{\breve S}$. Recognize that 
\beqs \bar\omega_M(f\otimes \idf)=\omega_M(f)
\eqs yields whenever $\omega_M$ is a state on $C(\Ab)$. With no doubt a $^*$-derivation is also defined for all $E_S(\Gamma)\in\ZD(\bar\E_{\breve S,\Gamma_j})$ for $j=1,...,\infty$ such that corollary \ref{cor derivationgraph} and proposition \ref{prop derivationinf} hold. Clearly it is true that
\beqs &\bar\omega_{M}((f\otimes E_S(\Gamma_\infty))^* (f\otimes E_S(\Gamma_\infty))) 
=\la \Omega_M,\pi_\Gamma((f^*\otimes E_S(\Gamma_\infty)^+(f\otimes E_S(\Gamma_\infty))) \Omega_M\ra \\
&=  \la \Omega_M^ \Gamma,\frac{1}{4} \bra E_S(\Gamma), f_\Gamma^*\bra E_{S}(\Gamma)^+, f_\Gamma\ket\ket  \Omega_M^\Gamma\ra 
+ \la \Omega_M^\Gamma,\frac{1}{4} f_\Gamma^*\bra E_{S}^+(\Gamma), f_\Gamma\ket \bra E_S(\Gamma),\Omega_M^\Gamma\ket \ra\\&\quad
-\la \Omega_M^\Gamma,\frac{1}{4} \bra E_{S}^+(\Gamma), f_\Gamma\bra E_S(\Gamma), f^
*_\Gamma\ket \ket \Omega_M^\Gamma\ra
-\la \Omega_M^\Gamma,\frac{1}{4} f_\Gamma\bra E_S(\Gamma), f^*_\Gamma\ket\bra E_{S}^+(\Gamma),\Omega_M^\Gamma\ket\ra \\
&= 0
\eqs
yields and therefore $\bar\omega_M$ is a $\mathfrak{B}_{\breve S,\ori}(\PD_{\Gamma_\infty}^{\op})$-invariant state on $C^\infty(\Ab)\rtimes_L\ZD(\bar\E_{\breve S})$. 

Let $\bar\omega^\prime_{M}$ be another state of the holonomy-flux cross-product $^*$-algebra $C^\infty(\Ab)\rtimes_L\bar \ZD(\E_{\breve S})$ such that $\bar\omega^\prime_{M}(f \otimes\idf)=\omega_M(f)$ for all $f\in C(\Ab)$. Recall \cite[Corollary 3.60]{Kaminski2} or \cite[Corollary 6.4.3]{KaminskiPHD} that states that, $\omega_M$ is the unique state on $C(\Ab)$ being invariant under the translation of $\bar G_{\breve S,\Gamma}$ and graph-diffeomorphisms $\Diff(\PD_{\Gamma_i}^{\op})$ for every $1\leq i\leq \infty$. Then it is assumed that, 
\beq\label{eq HF} \omega^\prime_{M}(\alpha_{\exp(E_{S}(\Gamma_\infty)^+E_{S}(\Gamma_\infty))}^{t_0}(f))\neq \omega^\prime_{M}(f)\quad\forall f\in C(\Ab)
\eq yields for some $t_0\in\R$. Consequently, derive $\omega^\prime_{M}(\delta_{S}(f))\neq 0$. 

But from \eqref{eq HF} it follows for a suitable graph-diffeomorphism $(\varphi,\Phi)\in\Diff_{\breve S,\ori}(\PD_{\Gamma_i}^{\op})$ such that $\varphi(S)=S^\prime$ for $S,S^\prime\in\breve S$ that 
\beq\label{eq diffeo}\omega^\prime_{M}(\alpha_{(\varphi,\Phi)}(\alpha_{\exp(E_{S}(\Gamma_\infty))}^{t_0}f))
&=\beta_{\Gamma}^*\omega^{\prime\text{ }\Gamma}_{M}(\alpha_{(\varphi,\Phi)}(\alpha_{\exp(E_{S}(\Gamma))}^{t_0}(f_\Gamma)))\\
&=\beta_{\Phi(\Gamma)}^*\omega^{\prime\text{ }\Phi(\Gamma)}_{M}(\alpha_{\exp(E_{S\circ\varphi}(\Gamma\circ\Phi))}^{t_0}(f_{\Phi(\Gamma)}))\\
&=\beta_{\Gp}^*\omega^{\prime\text{ }\Gamma^\prime}_{M}(\alpha_{\exp(E_{S^\prime}(\Gp))}^{t_0}(f_{\Gamma^\prime}))=\omega^\prime_{M}(\alpha_{\exp(E_{S^\prime}(\Gamma_\infty))}^{t_0}f)\\
&\neq \beta_{\Gamma}^*\omega^{\prime\text{ }\Gamma}_{M}(\alpha_{\exp(E_{S}(\Gamma))}^{t_0}(f_{\Gamma}))=\omega^\prime_{M}(\alpha_{\exp(E_{S}(\Gamma_\infty))}^{t_0}f)
\eq yields whenever $\Phi(\Gamma)=\Gamma^\prime$ and $f\in C(\Ab)$. In other words the state $\omega^\prime$ is only surface preserving graph-diffeomorphism invariant. The state $\omega_{M}^\prime$ is not invariant under surface-orientation-preserving graph-diffeomorphisms and, hence, general graph-diffeomorphisms.
Consequently, the state $\omega^\prime_{M}$ is equal to $\omega_{M}$ where $\omega_{M}(\delta_{S}(f))=0$ for all $f\in C^\infty(\Ab)$ and for all $\delta_{S}\in\Der(C^\infty(\Ab), C(\Ab))$.
\end{proofs}

\paragraph*{Conditions for a surface preserving graph-diffeomorphism-invariant state of the holonomy-flux cross-product $^*$-algebra}\hspace{10pt}

In theorem \ref{theo stateholfluxalg} the uniqueness of the state is referred to the assumption that, $\omega_M$ is surface-orientation graph-diffeomorphism-invariant state of the holonomy-flux cross-product $^*$-algebra $C^\infty(\Ab)\rtimes_L\ZD(\bar \E_{\breve S})$. Now this requirement is relaxed to surface preserving graph-diffeomorphisms. In the case of the holonomy-flux cross-product $C^*$-algebra $C^\infty(\Ab)\rtimes_{\alpha_{\overleftarrow{L}}}\ZD_{\breve S}$, a surface preserving graph-diffeomorphism-invariant state $\omega_{E(\breve S)}$ is presented in \cite[Proposition 4.6]{Kaminski2}, \cite[Proposition 7.2.6]{KaminskiPHD}. Hence the question arise if there exists another state on $C^\infty(\Ab)\rtimes_L\ZD(\bar \E_{\breve S})$ satisfying weaker conditions.

It is assumed that, the flux operators are implemented as $^*$-derivations $\delta_S$ on the domain $\DD(\delta_S)$ of the unital $C^*$-algebra $C(\Ab)$, which are generators of a strongly continuous one-parameter group $t\mapsto \alpha(t)$ of $^*$-automorphisms of $C(\Ab)$. In this case the derivation is of the form $\delta_S(f)=i\bra E_S(\Gamma_\infty)^+E_S(\Gamma_\infty),f\ket$ for $E_S(\Gamma_\infty)\in\bar\go_{\breve S}$, and where $i E_S(\Gamma_\infty)^+E_S(\Gamma_\infty)$ is some unbounded symmetric operator with domain $\DD$ on the Hilbert space $\HS_\infty$ such that $\DD(\delta_S)\DD\subset \DD$. 
Then a new state $\tilde\omega$ on $C^\infty(\Ab)\rtimes_L\ZD(\bar \E_{\breve S})$, which is not of the form $\tilde\omega(\delta_S(f))=0$, is required to satisfy a set of three conditions:

\textbf{First condition:}\\
Require the state $\tilde\omega^\Gamma$ to be $\Diff_{\breve S,\diff}(\PD_\Gamma^{\op})$-invariant, i.o.w. 
\beq\tilde\omega(\alpha_{(\varphi_S,\Phi_{\Gamma_\infty})}(f))
&= \beta_{\Gamma}^*\tilde\omega^\Gamma(\alpha_{(\varphi_S,\Phi_\Gamma)}(f_\Gamma))=\beta_{\Gamma}^*\tilde\omega^\Gamma(f_\Gamma)=\tilde\omega(f)\\
\tilde\omega(\alpha_{(\varphi_S,\Phi_{\Gamma_\infty})}(\delta_{S}(f)))
&=\beta_{\Gamma}^*\tilde\omega^\Gamma(\alpha_{(\varphi_S,\Phi_\Gamma)}(\delta_{S,\Gamma}(f_\Gamma)))
=\beta_{\Gamma}^*\tilde\omega^\Gamma(\delta_{S,\Gamma}(f_\Gamma))=\tilde\omega(\delta_{S}(f))
\eq for all $f\in C(\Ab)$, $f=\beta_\Gamma\circ f_\Gamma$, $\delta_{S}\in\Der(\DD(\delta_S),C(\Ab))$ and $(\varphi_{S},\Phi_{\Gamma})\in\Diff_{\breve S,\diff}(\PD_\Gamma^{\op})$. 

\textbf{Second condition:}\\
Furthermore the state need to have the property  
\beq\label{eq statecond1} \tilde\omega(\delta_{S}(f))\neq 0\quad\forall f\in C(\Ab)
\eq which is equivalent to the requirement that for some $t\in\R$
 \beqs \tilde\omega(\alpha_{\exp(E_S(\Gamma)^+E_S(\Gamma))}(t)(f))\neq \omega(f)\quad\forall f\in C(\Ab)\text{ and }E_S(\Gamma_\infty)\in\bar\E_{\breve S,\Gamma}
\eqs yields.

\textbf{Third condition:}\\
The state is assumed to fulfill
\beq\label{eq statecond2}\vert\tilde\omega(\delta_{S}(f))\vert\leq c\left(\tilde\omega(f^*f)+\tilde\omega(f f^*)\right)^{1/2}\quad\forall f\in C(\Ab)\text{ and some } c>0
\eq 

Clearly the ansatz is to search for $^*$-representations of $\bar \E_{\breve S,\Gamma}$ on $\HS_{\Gamma}$ that are not $\bar G_{\breve S,\Gamma}$-integrable, i.o.w. these representations are not equal to the infinitesimal representation $\dif U$ of some unitary representation $U$ of $\bar G_{\breve S,\Gamma}$. The author does not know such a representation. Note that, in comparison to Dziendzikowski and Oko\l{}\'{o}w \cite{MDziOko09} this representation is called the non-standard representation. But for representations of the universal enveloping algebra, which are not  $\bar G_{\breve S,\Gamma}$-integrable, the relation to the unitary Weyl elements, which define the Weyl algebra for surfaces, is not clear.

\section{Summary of different holonomy-flux cross-product $^*$-algebras}\label{sec summaryofholfluxstar}

In the following section three different $\bar\E_{\breve S,\Gamma}$-module algebras are presented. 

First of all remember that $C^\infty(\Ab_\Gamma)$ is a left (or right) $\bar\E_{\breve S,\Gamma}$-module algebra with the map
\[E_{S,\Gamma}\rhd f_\Gamma:=e^{L}(f_\Gamma)\text{ or resp. }E_{S,\Gamma}\lhd f_\Gamma:=e^{R}(f_\Gamma)\text{ for }f_\Gamma\in C^\infty(\Ab_\Gamma), E_{S,\Gamma}\in\bar\E_{\breve S,\Gamma}\]
where $e^{L}$ (or resp. $e^{R}$) denotes the right- (or left-) invariant vector field.
Then the vector space $C^\infty(\Ab_\Gamma)\otimes \bar\E_{\breve S,\Gamma}$ with the multiplication operation
\beqs (f^1_\Gamma\otimes E_{S_1,\Gamma})\cdot_L(f^2_\Gamma\otimes E_{S_2,\Gamma})
=f_\Gamma^1(E_{S_1,\Gamma}\rhd f^2_\Gamma)\otimes E_{S_2,\Gamma} 
+ f_\Gamma^1 f^2_\Gamma \otimes E_{S_1,\Gamma}\cdot E_{S_2,\Gamma} 
\eqs for $f_\Gamma^i\in C^\infty(\Ab_\Gamma)$,$ E_{S_i,\Gamma}\in\bar\E_{\breve S,\Gamma}$ and $ i=1,2$, or respectively
\beqs (f^1_\Gamma\otimes E_{S_1,\Gamma})\cdot_R(f^2_\Gamma\otimes E_{S_2,\Gamma})
=(E_{S_2,\Gamma}\lhd f^1_\Gamma)f_\Gamma^2\otimes E_{S_1,\Gamma}
+ f^1_\Gamma f_\Gamma^2\otimes E_{S_1}(\Gamma)\cdot E_{S_2,\Gamma}
\eqs for $f_\Gamma^i\in C^\infty(\Ab_\Gamma)$,$ E_{S_i,\Gamma}\in\bar\E_{\breve S,\Gamma}$ and $ i=1,2$, defines the \textbf{holonomy-flux cross-product $^*$-algebra} $C^\infty(\Ab_\Gamma)\rtimes_L \bar\E_{\breve S,\Gamma}$ or resp. $C^\infty(\Ab_\Gamma)\rtimes_R \bar\E_{\breve S,\Gamma}$. The elements satisfy the canonical commutator relation
\[\bra E_{S,\Gamma},f_\Gamma\ket= E_{S,\Gamma}\vartriangleright f_\Gamma\text{ or resp. }\bra E_{S,\Gamma},f_\Gamma\ket= E_{S,\Gamma}\vartriangleleft f_\Gamma\text{ for }f_\Gamma\in C^\infty(\Ab_\Gamma), E_{S,\Gamma}\in\bar\E_{\breve S,\Gamma}\]

For a vector $\psi\in\HS_\Gamma:= L^2(\Ab_\Gamma,\mu_\Gamma)$ there representation $\pi_{\Alg_\Gamma}$ of the analytic holonomy algebra $C^\infty(\Ab_\Gamma)$ is given by
\(\pi_{\Alg_\Gamma}(f_\Gamma)\psi_\Gamma:=f_\Gamma\cdot \psi_\Gamma\) and the representation $\pi_{\bar\E_{\breve S,\Gamma}}$ of the enveloping flux algebra $\bar\E_{\breve S,\Gamma}$ is presented by
\(\pi_{\bar\E_{\breve S,\Gamma}}(E_ {S,\Gamma})\psi_\Gamma:=E_ {S,\Gamma}\rhd\psi_\Gamma\text{ for }f_\Gamma\in C^\infty(\Ab_\Gamma), E_{S,\Gamma}\in\bar\E_{\breve S,\Gamma},\psi_\Gamma\in \HS_\Gamma\)

Then the \textbf{Heisenberg representation} of $C^\infty(\Ab_\Gamma)\rtimes_L \bar\E_{\breve S,\Gamma}$ on the Hilbert space $\HS_\Gamma$ is defined by
\[\pi(\bra E_{S,\Gamma},f_\Gamma\ket ) \psi_\Gamma= (E_{S,\Gamma}\rhd f_\Gamma)\cdot \psi_\Gamma -f_\Gamma\cdot( E_{S,\Gamma}\rhd \psi_\Gamma)\text{ for }f_\Gamma\in C^\infty(\Ab_\Gamma), E_{S,\Gamma}\in\bar\E_{\breve S,\Gamma},\psi_\Gamma\in \HS_\Gamma
\]
Notice that 
\beqs E_{S,\Gamma}\rhd \psi_\Gamma=:\dif U(E_{S,\Gamma})\psi_\Gamma\text{ for }E_{S,\Gamma}\in\bar\E_{\breve S,\Gamma},\psi_\Gamma\in \HS_\Gamma
\eqs defines a representation $\dif U$ of the enveloping algebra $\bar\E_{\breve S,\Gamma}$ on $\HS_\Gamma$. This is also called the infinitesimal representation of the Lie flux group $\bar G_{\breve S,\Gamma}$ on $\HS_\Gamma$. The element $\dif  U(E_{S,\Gamma}^+E_{S,\Gamma})$ itself is an essential self-adjoint operator on this Hilbert space.

Another left (or resp. right) $\bar\E_{\breve S,\Gamma}$-module algebra is given by $C^\infty(\Ab_\Gamma)$ with the bilinear map
\[E_{S,\Gamma}\rhd_H f_\Gamma:=e^{L}(f_\Gamma\vert_{(e_G,...,e_G)})=:\la f_\Gamma,E_{S,\Gamma}\ra\text{ for }f_\Gamma\in C^\infty(\Ab_\Gamma), E_{S,\Gamma}\in\bar\E_{\breve S,\Gamma}\]  and where $\la .,.\ra:C^\infty(\Ab_\Gamma)\otimes \bar\E_{\breve S,\Gamma}\rightarrow \CB$. Then the
\textbf{Heisenberg holonomy-flux cross-product $^*$-algebra} \(C^\infty(\Ab_\Gamma)\rtimes_H\bar\E_{\bar S,\Gamma}\)
is given by the vector space $C^\infty(\Ab_\Gamma)\otimes\bar\E_{\bar S,\Gamma}$ with the multiplication operation
\beq (f_\Gamma,E_{S,\Gamma})\cdot_H(\tilde f_\Gamma,\tilde E_{S,\Gamma}):= \la E_{S,\Gamma},\idf\ra f_\Gamma \tilde f_\Gamma\otimes \tilde E_{S,\Gamma} + \la \idf,\tilde f_\Gamma\ra f_\Gamma\otimes E_{S,\Gamma}\tilde E_{S,\Gamma}\text{ for }f_\Gamma\in C^\infty(\Ab_\Gamma), E_{S,\Gamma}\in\bar\E_{\breve S,\Gamma}
\eq The elements satisfy the canonical commutator relation
\beqs E_{S,\Gamma}f_\Gamma= e^{L}(f_\Gamma\vert_{(e_G,...,e_G)}) - f_\Gamma E_{S,\Gamma}\text{ for }f_\Gamma\in C^\infty(\Ab_\Gamma), E_{S,\Gamma}\in\bar\E_{\breve S,\Gamma}
\eqs This algebra is indeed a Heisenberg double in the sense of Schm\"udgen and Klimyk \cite{KlimSchmued94}. The Heisenberg representation of $C^\infty(\Ab_\Gamma)\rtimes_H\bar\E_{\bar S,\Gamma}$ on $\HS_\Gamma$ is given by
\[\pi(\bra E_{S,\Gamma},f_\Gamma\ket ) \psi_\Gamma= e^{L}(f_\Gamma\vert_{(e_G,...,e_G)})\psi_\Gamma -f_\Gamma\cdot  e^L(\psi_\Gamma)\text{ for }f_\Gamma\in C^\infty(\Ab_\Gamma), E_{S,\Gamma}\in\bar\E_{\breve S,\Gamma},\psi_\Gamma\in \HS_\Gamma
\]

The third possibility is given by the left (or resp. right) $\bar\E_{\breve S,\Gamma}$-module algebra $C^\infty(\Ab_\Gamma)$ with
\[E_{S,\Gamma}\rhd_\epsilon f_\Gamma:=\epsilon(f_\Gamma)E_{S,\Gamma}\text{ for }f_\Gamma\in C^\infty(\Ab_\Gamma), E_{S,\Gamma}\in\bar\E_{\breve S,\Gamma}\]  where the Hopf algebra $(C^\infty(\Ab_\Gamma),S)$ is considered with antipode $S$, comultiplication $\bigtriangleup$ and counit\\ $\epsilon :C^\infty(\Ab_\Gamma)\rightarrow \CB$. 
Then the \textbf{simple Holonomy-flux cross-product $^*$-algebra} \(C^\infty(\Ab_\Gamma)\rtimes_\epsilon\bar\E_{\bar S,\Gamma}\) is defined by the vector space $C^\infty(\Ab_\Gamma)\otimes\bar\E_{\bar S,\Gamma}$ with multiplication operation
\beqs (f_\Gamma\otimes E_{S,\Gamma})\cdot_{\epsilon}(\tilde f_\Gamma\otimes \tilde E_{S,\Gamma})=f_\Gamma\tilde f_\Gamma\otimes E_{S,\Gamma}\tilde E_{S,\Gamma}\text{ for }f_\Gamma\in C^\infty(\Ab_\Gamma), E_{S,\Gamma}\in\bar\E_{\breve S,\Gamma}
\eqs
The elements satisfying the canonical commutator relation given by
\beqs E_{S,\Gamma}f_\Gamma= \epsilon(f_\Gamma)E_{S,\Gamma} - f_\Gamma E_{S,\Gamma}\text{ for }f_\Gamma\in C^\infty(\Ab_\Gamma), E_{S,\Gamma}\in\bar\E_{\breve S,\Gamma}
\eqs The Heisenberg representation of $C^\infty(\Ab_\Gamma)\rtimes_\epsilon\bar\E_{\bar S,\Gamma}$ on $\HS_\Gamma$ is presented by
\[\pi(\bra E_{S,\Gamma},f_\Gamma\ket ) \psi_\Gamma= \epsilon(f_\Gamma)e^L(\psi_\Gamma) -f_\Gamma\cdot  e^L(\psi_\Gamma)\text{ for }f_\Gamma\in C^\infty(\Ab_\Gamma), E_{S,\Gamma}\in\bar\E_{\breve S,\Gamma},\psi_\Gamma\in \HS_\Gamma
\]

The algebra $C^\infty(\Ab)\rtimes_L\DD_{\breve S}(\bar G_{\breve S})$ is an $O^ *$-algebra and is called the \textbf{holonomy-flux cross-product $O^ *$-algebra associated a surface set} $\breve S$ on $C^ \infty(\Ab)$ in $\HS_\infty$.
\section{Tensor products of the holonomy-flux cross-product $^*$-algebra}\label{subsec tensorholflux}

The structure of the holonomy-flux cross-product $^*$-algebra can be slightly modificated in the following way.

\begin{defi}
The \textbf{modified holonomy-flux cross-product $^*$-algebra restricted to a graph $\Gamma$ and a surface set $\breve S$} is given by 
\[ C(\bar G_{\breve S,\Gamma})\otimes \left(C^\infty(\Ab_\Gamma)\rtimes_L \bar\E_{\breve S,\Gamma}\right)\] where the tensor product $\otimes$ is the minimal tensor product of $C^*$-algebras.

The \textbf{modified holonomy-flux cross-product $^*$-algebra associated a surface set $\breve S$} is equivalent to the inductive limit of the family 
\[\Bigg\{\Bigg(C(\bar G_{\breve S,\Gamma})\otimes \left(C^\infty(\Ab_\Gamma)\rtimes_L \bar\E_{\breve S,\Gamma}\right) ,\mathring\beta_{\Gamma,\Gamma^\prime}\times \beta_{\Gamma,\Gamma^\prime}\Bigg):\PD_\Gamma\leq \PD_{\Gp}\Bigg\}\]
\end{defi}

Then for the state $\check\omega_M^\Gamma$ on $C(\bar G_{\breve S,\Gamma})\otimes \left(C^\infty(\Ab_\Gamma)\rtimes_L \bar\E_{\breve S,\Gamma}\right)$ it is true that
\beq &\check\omega^\Gamma_M\left(f(\rho_{S,\Gamma}(\Gamma))\delta_{S,\Gamma}(f_\Gamma)\right)\\
&=\int_{\bar G_{\breve S,\Gamma}}\dif\mu_{\breve S,\Gamma}(\rho_{S,\Gamma}(\Gamma)) f(\rho_{S,\Gamma}(\Gamma))\delta(\rho_{S,\Gamma}(\Gamma),\exp(E_S(\Gamma))) \la \Omega_M^\Gamma,  \delta_{S,\Gamma}(f_\Gamma)\Omega_M^\Gamma\ra\\
&= f(\exp(E_S(\Gamma)) \omega_M^\Gamma\left( \delta_{S,\Gamma}(f_\Gamma)\right)
\eq holds where $\delta(g_1,g_2)$ is the delta function on $\bar G_{\breve S,\Gamma}$.

\begin{defi}
The \textbf{modified intersection-holonomy-flux cross-product $^*$-algebra restricted to a graph $\Gamma$ and a surface set $\breve S$} is given by 
\[ C(V_\Gamma^S)\otimes \left(C^\infty(\Ab_\Gamma)\rtimes_L \bar\E_{\breve S,\Gamma}\right)\] where $V_\Gamma^S=V_\Gamma\cap S$ and the tensor product $\otimes$ is the minimal tensor product of $C^*$-algebras.
\end{defi}

Then for the state $\hat\omega_M^\Gamma$ on $C(V_\Gamma^S)\otimes \left(C^\infty(\Ab_\Gamma)\rtimes_L \bar\E_{\breve S,\Gamma}\right)$ it is true that
\beq &\hat\omega^\Gamma_M\left(f(v_1,...,v_{M})\delta_{S,\Gamma}(f_\Gamma)\right)= f(v_1,...,v_{M}) \omega_M^\Gamma\left( \delta_{S,\Gamma}(f_\Gamma)\right)
\eq yields where $v_1,...,v_{M}\in V_\Gamma^S$.

Clearly these states are not surface-orientation-preserving graph-diffeomorphism invariant, but the states are surface preserving graph-diffeomorphism invariant.

\section{Comparison table}\label{sec tableQMholflux}

The author has argued in \cite{Kaminski0} that, for example for Quantum Mechanics different algebras are obtained by using different generating sets of abstract operators. The aim of the construction of the Weyl $C^*$-algebra, which has been invented in \cite{Kaminski1}, and the holonomy-flux cross-product $^*$-algebras defined in this article, is to use a common setup. Both algebras are generated by functions depending on holonomies along paths, and group- or Lie algebra-valued quantum flux operators. These abstract operators satisfy some canonical commutator relations, which are called Heisenberg relations if the unbounded configuration and momentum operators are studied, or Weyl relations if the bounded configuration and momentum operators are used. Clearly by choosing different sets of operators other $^*$-algebras or respectively $C^*$-algebas can be constructed. For example if the functions depending on the quantum flux group associated to surfaces, and the functions depending on holonomies along paths, are considered, then these operators generate the holonomy-flux cross-product $C^*$-algebra, which has been presented in \cite{Kaminski2}. The abstract operators are represented on a common Hilbert spaces as self-adjoint Hilbert space operators. The exponentiated Lie algebra-valued quantum flux operator are implemented by an unitary weakly continuous representation of the group $\R$ on the Hilbert space. The Lie algebra-valued quantum flux operator is related to the infinitesimal representation of this unitary weakly continuous representation. This flux operator is unbounded and self-adjoint. In a Hilbert space independent framework automorphisms of and derivations for the algebra of quantum variables play a fundamental role. In particular strongly continuous one-parameter group of $^*$-automorphisms defines a derivation, which is given by the commutator of two self-adjoint operators. 

\begin{landscape}
\begin{longtable}[ht]{|l|l|l|}
\hline
&&\\
&Quantum Mechanics & Weyl alg. for surfaces and holonomy-flux cross-prod. $^*$-alg.\\
\hline \hline
&&\\
Configuration space & $\R^n$ & $\Ab:=\limpPD \Ab_\Gamma$\\[5pt]
Momentum space & $\R^n$ & $\bar\E_{\breve S}$ or $\bar G_{\breve S}$ ( where $G$ compact connected Lie group)\\[5pt]
Configuration variable I & $x_i$ & $\ho(\gamma_i)$ for $\ho\in\Hom(\PD,G)$,  $\gamma_i\in\PD$\\[5pt]
Configuration variable II & $f(x_i)$ for $f\in C_0(\R^n)$& $f(\ho(\gamma_i))$ for $f\in C(\Ab)$\\[5pt]
Momentum variable I & $p_i$ & $E_{S_j}(\gamma_i)$ for $E_{S_j}\in\E_{\breve S}$,  $\gamma_i\in\PD$, $S_j\in\breve S$\\[5pt]
Momentum variable II & $\exp(tp_i)$ & $\rho_{S_j}(\gamma_i)$ for $\rho_{S_j}\in G_{\breve S}$,  $\gamma_i\in\PD$, $S_j\in\breve S$\\[5pt]
Dynamical Hamiltonian& $H=\sum_{i}\frac{p_i^2}{2m}+V(x_1,..,x_n)$& $H=\sum_{i}Tr(\left(\ho(\alpha_i)-\ho(\alpha_i)^{-1}\right)\ho(\gamma_i)[\ho(\gamma_i)^{-1},V])$\\[2pt]
&& $V=\sum_{jk,l}E_{S_1}(\gamma_j)E_{S_2}(\gamma_k)E_{S_3}(\gamma_l)$\\[3pt]
&& for $\alpha_i,\gamma_x\in\PD$, $S_m\in\breve S$\\[5pt]
Hilbert space& $\HS:=L^2(\R^n,\times_{1\leq k\leq n}\dif x_k)$& $\HS_\infty:=L^2(\Ab,\dif\mu_{\infty})$\\[5pt]
self-adjoint Hilbert space operator& $\pi(x_i)=x_i$& $\pi(\ho(\gamma_i))=\ho(\gamma_i)$\\[5pt]
unitary Hilbert space operator& $\pi(\exp(tp_i))=U_{p_i}(t)$ & $\pi(\exp(tE_{S_j}(\gamma_i)))=U_t(E_{S_j}(\gamma_i))$\\[5pt]
&$\big(U_{p_i}(t)\psi\big)(x_i):=\psi(x_i-tp_i)$ for $\psi\in\HS$& $\big(U_t(E_{S_j}(\gamma_i))\psi\big)(\ho(\gamma_i)):=\psi(\exp(tE_{S_j}(\gamma_i))\ho(\gamma_i))$ for $\psi\in\HS$\\[5pt]
unitary Hilbert space operator&&$\pi(R_{\sigma})=V_\sigma$ for $\sigma\in\mathfrak{B}(\PD_\Gamma)$, $t\circ\sigma\in\Diff(V_\Gamma)$\\
&&\\
$^*$-automorphism &&$\pi(\zeta_\sigma(f))=V_\sigma\pi(f)V_\sigma^*$, $\zeta_\sigma\in\Aut(C(\Ab))$\\[5pt]
unitary transformation & Fourier transform $\FD$ &\\[5pt]
self-adjoint Hilbert space operator& $\pi(p_i) = -i\frac{\partial}{\partial x_i}$ such that & \\[5pt]
&$\pi(p_i)\psi =\FD p_i \FD^{-1} \psi = -i\frac{\partial}{\partial x_i}\psi$ for $\psi\in D(p_i)$& \\[5pt]
self-adjoint Hilbert space operators & $\pi(H)$, $\pi(\exp(tH)=:U_H(t)$& \\[5pt]
strongly continuous $1$-parameter unitary group& $\R\ni t\mapsto U_H(t)$ such that&$\R\ni t\mapsto U_t(E_{S_j}(\gamma_i))$\\[5pt]
self-adjoint Hilbert space operator & & $\pi(E_{S_j}(\gamma_i)^+E_{S_j}(\gamma_i))=-i\frac{\dif}{\dif t}U_t(E_{S_j}(\gamma_i)^+E_{S_j}(\gamma_i))$ such that\\[5pt]
Stone's theorem &$\frac{\dif}{\dif t}\Big\vert_{t=0}U_H(t)\psi=iU_H(t)\Big\vert_{t=0}\pi(H)\psi$ for $\psi\in D(H)$&  $\frac{\dif}{\dif t}\Big\vert_{t=0}U_t(E_{S_j}(\gamma_i)^+E_{S_j}(\gamma_i))\psi =i\pi(E_{S_j}(\gamma_i)^+E_{S_j}(\gamma_i))\psi$  \\[5pt]
&$\frac{\dif}{\dif t}\psi_t=i\pi(H)\psi_t$ for $\psi_t:=U_H(t)\psi$ & for $\psi\in D(E_{S_j}(\gamma_i)^+E_{S_j}(\gamma_i))$\\[5pt]
%
\hline\newpage\hline &&\\
unitary Hilbert space operator& $\pi(\exp(s x_i))=V_{x_i} (s)$& \\[5pt]
self-adjoint Hilbert space operator& & $\pi(f)=f$ for $f\in C(\Ab)$ \\[5pt]
strongly continuous $1$-parameter group &&$\R\ni t\mapsto \alpha_t(E_{S_j}(\gamma_i))\in \Aut( C(\Ab))$ \\
of $^*$-automorphisms&&\\
Canonical Commutatur Relations&$p_ix_j-x_jp_i=-i\delta_{ij}$ (Heisenberg relations)& 
$\bra E_{S_j}(\gamma_i),\ho(\gamma_i)\ket = i\frac{\dif}{\dif t}\Big\vert_{t=0}\exp(t E_{S_j}(\gamma_i))\ho(\gamma_i)-\ho(\gamma_i)E_{S_j}(\gamma_i)$\\[5pt]
%
%
&$V_{x_i}(s)U_{p_j}(t)=\exp(st\delta_{ij})U_{p_j}(t)V_{x_i}(s)$ (Weyl rel.)& $\bra E_{S_j}(\gamma_i),f\ket = i\frac{\dif}{\dif t}\Big\vert_{t=0}\alpha_t(E_{S_j}(\gamma_i))(f)$\\[5pt]
%
%
$^*$-automorphism &&$\pi(\alpha(\rho_{S_j}(\gamma_i))(f))=U(\rho_{S_j}(\gamma_i))\pi(f)U^*(\rho_{S_j}(\gamma_i))$\\[3pt]
&&$\alpha(\rho_{S_j}(\gamma_i))\in\Aut(C(\Ab))$\\[5pt]
strongly contin. $1$-parameter group&&$\mathfrak{B}(\PD_\Gamma)\ni\sigma\mapsto \zeta_\sigma\in\Aut(C(\Ab))$\\[5pt]
of $^*$-automorphism&&$\alpha(\rho_{S_j}(\gamma_i))\circ\zeta_\sigma=\zeta_\sigma\circ\alpha(\rho_{S_j}(\gamma_i))$\\[3pt]
&&  for all $\sigma\in\mathfrak{B}(\PD_\Gamma)$ and $\rho_{S_j}\in \ZD_{\breve S,\gamma_i}$\\[5pt]
Uniqueness of the GNS-representation& $(\HS,\pi,\Omega)$ irreducible, cyclic and regular  &$(\HS_\infty,\pi,\Omega)$ irreducible and regular \\[3pt]
&representation of $C_0(\R^n)$&GNS-representation of $C(\Ab)$\\[5pt]
&such that $\omega(f)=\la \Omega,\pi(f)\Omega\ra$ and $U_H(t)\Omega=\Omega$& such that $\omega_{M}(f)=\la \Omega,\pi(f)\Omega\ra$, \\[3pt]
&& $V_\sigma\Omega=\Omega$  for all $\sigma\in\mathfrak{B}(\PD_\Gamma)$\\[5pt]
&&and $U(\rho_{S_j}(\gamma_i))\Omega=\Omega$ for all $\rho_{S_j}\in \ZD_{\breve S}$ and $\gamma_i\in\PD$\\[10pt]
&&$(\HS_\infty,\Phi,\Omega_{M})$ irreducible and regular GNS-repr. of $\WF\text{eyl}_\ZD(\breve S)$\\[5pt]
&&such that $\bar\omega_M(W)=\la \Omega_M,\Phi(W)\Omega_M\ra$\\[3pt]
&& $V_\sigma\Omega_M=\Omega_M$  for all $\sigma\in\mathfrak{B}(\PD_\Gamma)$\\[5pt]
&&and $U(\rho_{S_j}(\gamma_i))\Omega_M=\Omega_M$ for all $\rho_{S_j}\in \ZD_{\breve S}$ and $\gamma_i\in\PD$\\[5pt]
&&(w.r.t. natural or non-standard identification of conf. space)\\[5pt]
symmetric $^*$-derivation with domain $D(\delta)$&& $\delta_{S_j}(f):=i\bra E_{S_j}(\gamma_i)^+E_{S_j}(\gamma_i),f\ket$\\[5pt]
&&$\omega_M(\delta_{S_j}(f))=0$\\[10pt]
Uniqueness of the state&& $\tilde\omega_M$ is the unique state on $C^\infty(\Ab)\rtimes_{L} \ZD(\bar \E_{\breve S})$\\[5pt]
&& such that $\tilde\omega_M\circ\alpha_\sigma= \tilde\omega_M$ for all $\sigma\in\mathfrak{B}(\PD_\Gamma^{\op})$\\[3pt]
&&$\tilde\omega_M\circ\alpha_t(E_{S_j}(\gamma_i)^+E_{S_j}(\gamma_i))= \tilde\omega_M$ for all $E_{S_j}(\gamma_i)\in \ZD(\bar \E_{\breve S})$\\[5pt]
\hline
\end{longtable} 
\end{landscape}

\section{Appendix}
The theory of $O^*$-algebras has been developed by Schm\"udgen \cite{Schmuedgen90} and Inoue \cite{Inoue}. In this appendix only the basic objects are collected.

\subsection*{Definition of $O^*$-algebras}\label{app OSA}
Let $\DD$ be a dense subspace in a Hilbert space with inner product $\la .,.\ra$. By $\Lop(\DD)$ (respect. $\Lop_c(\DD)$) denote the set of all (closable) linear operators from $\DD$ to $\DD$ and 
\beqs \Lop^+(\DD)=\{A\in\Lop(\DD): \DD\subset\DD(A^*), A^*\DD\subset\DD\}\eqs 

Then with the operations $AB$, $A+B$ and $\lambda A$ the set $\Lop(\DD)$ forms an algebra. The set $\Lop^+(\DD)$ forms a $^*$-algebra with involution $A\mapsto A^+= A^*\vert_{\DD}$.

\begin{prop}\cite[Prop 2.1.10]{Schmuedgen90}
Let $A\in\Lop^+(\DD)$ and let $A$ be closed. Then $\Lop^+(\DD)$ is equal to the algebra $\LD(\HS)$ of bounded linear operators on a Hilbert space $\HS$. 
\end{prop}
Observe 
\[\Lop^+(\DD)\subset\Lop_c(\DD)\subset\Lop(\DD)\]

\begin{defi}
A subalgebra of $\Lop(\DD)$ contained in $\Lop_c(\DD)$ is said to be an\textbf{ $O$-algebra on $\DD$ in $\HS$}, and a $^*$-subalgebra of $\Lop^+(\DD)$ is said to be an \textbf{$O^*$-algebra on $\DD$ in $\HS$}. 
\end{defi}

\subsection*{Representations of $^*$-algebras}\label{app rep} 
\begin{defi}Let $\Alg$ be a $^*$-algebra with unit $\idf$ and let $\DD$ be a dense subspace of a Hilbert space $\HS$.
  
The map $\pi:\Alg\rightarrow\Lop(\DD)$ is a \textbf{$^*$-representation of a $^*$-algebra $\Alg$} on a Hilbert space $\HS$ if
\begin{enumerate}
 \item there exists a dense subset $\DD$ of $\HS$ such that
\beqs \DD\subset \bigcap_{A\in\Alg}\left(D(\pi(A))\cap D(\pi(A)^*)\right)\eqs
 \item for every $A,B\in\Alg$ and $\lambda\in\CB$ 
\beqs \pi(A+B)=\pi(A)+\pi(B),\quad &\pi(\lambda A)=\lambda\pi(A)\\
\pi(AB)=\pi(A)\pi(B),\quad & \pi(A^*)=\pi(A)^*\\
\pi(\idf)=I&
\eqs
\end{enumerate}
\end{defi}

\section*{Acknowledgements}
The work has been supported by the Emmy-Noether-Programm (grant FL 622/1-1) of the Deutsche Forschungsgemeinschaft.

\end{document}